\documentclass[a4paper,english]{lipics-report}
\input{macro.sty}

\begin{document}
\title{Regular Transducer Expressions for Regular Transformations}
\author[1]{Vrunda Dave}
\author[2]{Paul Gastin}
\author[1]{Shankara Narayanan Krishna}
\affil[1]{Dept of CSE, IIT Bombay, India\\
  \texttt{vrunda,krishnas@cse.iitb.ac.in}}
\affil[2]{LSV, ENS Paris-Saclay, CNRS, Universit{\'e} Paris-Saclay, France\\
  \texttt{paul.gastin@lsv.fr}}
\authorrunning{Vrunda Dave, P. Gastin, S. Krishna} 

\Copyright{John Q. Open and Joan R. Access}

\maketitle
\abstract

Functional MSO transductions, deterministic two-way transducers, as well as
streaming string transducers are all equivalent models for regular functions.
In this paper, we show that every regular function, either on finite words or on
infinite words, captured by a deterministic two-way transducer, can be described
with a regular transducer expression (RTE).  For infinite words, the
transducer uses Muller acceptance and $\omega$-regular look-ahead.
\RTEs are constructed from constant functions using the combinators if-then-else
(deterministic choice), Hadamard product, and unambiguous versions of the Cauchy
product, the 2-chained Kleene-iteration and the 2-chained omega-iteration.  Our
proof works for transformations of both finite and infinite words, extending the
result on finite words of Alur et al.\ in LICS'14.
In order to construct an RTE associated with a deterministic two-way Muller
transducer with look-ahead, we introduce the notion of transition monoid for
such two-way transducers where the look-ahead is captured by some backward
deterministic B\"uchi automaton.  Then, we use an unambiguous version of Imre
Simon's famous forest factorization theorem in order to derive a ``good''
($\omega$-)regular expression for the domain of the two-way transducer.
``Good'' expressions are unambiguous and Kleene-plus as well as
$\omega$-iterations are only used on subexpressions corresponding to
\emph{idempotent} elements of the transition monoid.  The combinator expressions
are finally constructed by structural induction on the ``good''
($\omega$-)regular expression describing the domain of the transducer.

\section{Introduction}\label{sec:intro}

One of the most fundamental results in theoretical computer science is that the
class of regular languages corresponds to the class of languages recognised by
finite state automata, to the class of languages definable in MSO, and to the
class of languages whose syntactic monoid is finite.  Regular languages are also
those that can be expressed using a regular expression; this equivalence is
given by the Kleene's theorem.  This beautiful correspondence between machines,
logics and algebra in the case of regular languages paved the way to
generalizations of this fundamental theory to regular transformations
\cite{EH01}, where, it was shown that regular transformations are those which
are captured by two-way transducers and by MSO transductions a la Courcelle.
Much later, streaming string transducers (SSTs) were introduced \cite{fst10} as
a model which makes a single pass through the input string and use a finite set
of variables that range over strings from the output alphabet.  \cite{fst10}
established the equivalence between SSTs and MSO transductions, thereby showing
that regular transformations are those which are captured by either SSTs,
two-way transducers or MSO transductions.  This theory was further extended to
work for infinite string transformations \cite{lics12}; the restriction from MSO
transductions to first-order definable transductions, and their equivalence with
aperiodic SSTs and aperiodic two-way transducers has also been established over
finite and infinite strings \cite{fst14}, \cite{fst16}.  Other generalizations
such as \cite{sstt17}, extend this theory to trees.  Most recently, this
equivalence between SSTs and logical transductions are also shown to hold good
even when one works with the origin semantics \cite{icalp17}.

Moving on, an interesting generalization pertains to the characterization of the
output computed by two-way transducers or SSTs (over finite and infinite words)
using regular-like expressions.  For the strictly lesser expressive case of
sequential one-way transducers, this regex characterization of the output is
obtained as a special case of Sch\"{u}tzenberger's famous equivalence
\cite{DrosteBook} between weighted automata and regular weighted expressions.
The question is much harder when one looks at two-way transducers, due to the
fact that the output is generated in a one-way fashion, while the input is read
in a two-way manner.  The most recent result known in this direction is
\cite{lics14}, which provides a set of combinators, analogous to the operators
used in forming regular expressions.  These combinators are used to form
\emph{combinator expressions} which compute the output of an additive cost
register automaton (ACRA) over finite words.  ACRAs are generalizations of SSTs
and compute a partial function from finite words over a finite alphabet to
values from a monoid $(\mathbb{D}, +,0)$ (SSTs are ACRAs where $(\mathbb{D},
+,0)$ is the free monoid $(\Gamma^*, ., \epsilon)$ for some finite output
alphabet $\Gamma$).  The combinators introduced in \cite{lics14} form the basis
for a declarative language DReX \cite{popl15} over finite words, which can
express all regular string-to-string transformations, and can also be
efficiently evaluated.

\noindent {\bf Our Contributions}.  We generalize the result of \cite{lics14}.
Over finite words, we work with two-way deterministic transducers (denoted 2DFT,
see Figure~\ref{fig:intro} left) while over infinite words, the model considered
is a deterministic two-way transducer with regular look-ahead, equipped with the
Muller acceptance condition.  For example, Figure~\ref{fig:intro} right gives an
\twoDMTla ($\mathsf{la}$ stands for look-ahead and $\mathsf{M}$ in the
$\mathsf{2DMT}$ for Muller acceptance).

\begin{figure}[b]
  \raisebox{4.0mm}{\includegraphics[scale=0.72,page=1]{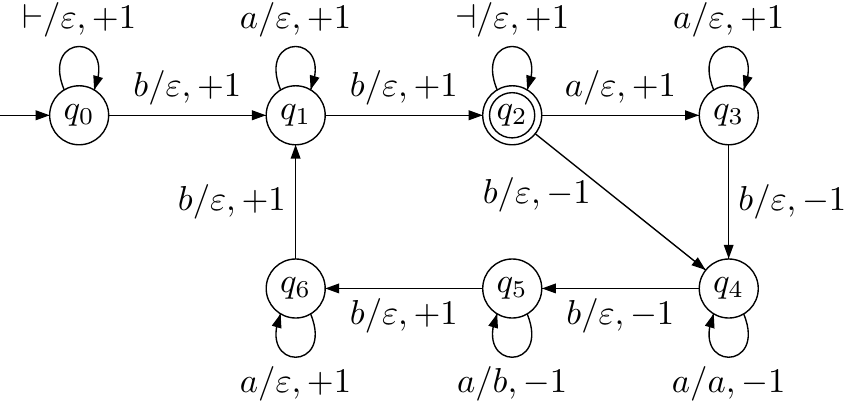}}
  \hfill\includegraphics[page=3,scale=0.72]{gpicture-pics.pdf}
  
  \caption{On the left, a 2DFT $\Aa$ with $\sem{\A}(ba^{m_1}ba^{m_2}b \dots
  a^{m_k}b)=a^{m_2}b^{m_1}a^{m_3}b^{m_2} \dots a^{m_k}b^{m_{k-1}}$.
  On the right, an \twoDMTla $\Aa'$ with $\sem{\Aa'}(u_1\#u_2\#\dots
  \#u_n\#v)=u^R_1u_1\#u^R_2u_2\# \dots \#u^R_nu_n\#v$ where $u_1, \dots, u_n \in
  (a+b)^*$, $v \in (a+b)^{\omega}$ and $u^R$ denotes the reverse of $u$. The
  Muller acceptance set is $\{\{q_5\}\}$.  The look-ahead expressions $\Sigma^* \#
  \Sigma^{\omega}$ and $(\Sigma \backslash \{\#\})^{\omega}$ are used to check if
  there is a $\#$ in the remaining suffix of the input word.}
  \label{fig:intro}
\end{figure} 
 
In both cases of finite words and infinite words, we come up with a set of
combinators using which, we form \emph{regular transducer expressions} (\RTE)
characterizing the output of the two-way transducer (2DFT/\twoDMTla).  
  
\noindent {\bf The Combinators}.  
We describe our basic combinators that form the building blocks of \RTEs.  The
semantics of an \RTE is a partial function
$f\colon\Sigma^{\infty}\to\Gamma^{\infty}$ whose domain is denoted $\dom{f}$.
\begin{itemize}
  \item We first look at the case of finite words and describe the basic
  combinators.  The constant function is one which maps all strings in
  $\Sigma^*$ to some fixed value $d$.  Given a string $w \in \Sigma^*$, the
  if-then-else combinator $\Ifthenelse{K}{f}{g}$ checks if $w$ is in the regular
  language $K$ or not, and appropriately produces $f(w)$ or $g(w)$.  The
  unambiguous Cauchy product $f \boxdot g$ when applied on $w \in \Sigma^*$
  produces $f(u)\cdot g(v)$ if $w=u\cdot v$ is an unambiguous decomposition of
  $w$ with $u\in\dom{f}$ and $v\in\dom{g}$.  The unambiguous Kleene-plus
  $f^{\boxplus}$ when applied to $w \in \Sigma^*$ produces $f(u_1) \cdots
  f(u_n)$ if $w=u_1 \cdots u_n$ is an unambiguous factorization of $w$, with
  each $u_i \in \dom{f}$.  The Hadamard product $f \odot g$ when applied to $w$
  produces $f(w)\cdot g(w)$.  Finally, the unambiguous 2-chained Kleene-plus
  $\twoplus{K}{f}$ when applied to a string $w$ produces as output
  $f(u_1u_2)\cdot f(u_2u_3) \cdots f(u_{n-1}u_n)$ if $w$ can be unambiguously
  written as $u_1u_2 \cdots u_n$, with each $u_i \in K$, for the regular
  language $K$.  We also have the reverses $f \rldot g, \rlplus{f}$ and
  $\rltwoplus{K}{f}$: $[f \rldot g](w)$ produces $g(v)\cdot f(u)$ if $w$ is the
  unambiguous concatenation $u\cdot v$ with $u\in\dom{f}$ and $v\in\dom{g}$,
  $\rlplus{f}(w)$ produces $f(u_n) \cdots f(u_1)$ if $w$ is the unambiguous
  catenation $u_1 \cdots u_n$ with $u_i\in\dom{f}$ for all $i$, and,
  $\rltwoplus{K}{f}(w)$ produces $f(u_{n-1}u_n) \cdots f(u_1u_2)$ if $w$ is the
  unambiguous catenation $u_1\cdots u_n$ with $u_i \in K$ for all $i$.
  
  \item In the case of infinite words, the Cauchy product $f \boxdot g$ works on
  $w \in \Sigma^{\omega}$ if $w$ can be written unambiguously as $u\cdot v$ with
  $u\in\dom{f}\cap\Sigma^*$ and $v\in\dom{g}\cap\Sigma^{\omega}$.  Another
  difference is in the use of the Hadamard product: for $w \in \Sigma^{\omega}$,
  $f \odot g$ produces $f(w)\cdot g(w)$ if $f(w)$ is a finite string.  Note that
  these are sound with respect to the concatenation semantics for infinite
  words.  Indeed, we also have $\omega$-iteration and two-chained 
  $\omega$-iteration: $f^\omega(w)=f(u_1)f(u_2)\cdots$ if $w\in\Sigma^\omega$ 
  can be unambiguously decomposed as $w=u_1u_2\cdots$ with 
  $u_i\in\dom{f}\cap\Sigma^{*}$ for all $i\geq1$. Moreover,
  $\twoomega{K}{f}(w)=f(u_1u_2)f(u_2u_3)\cdots$ if $w\in\Sigma^\omega$ 
  can be unambiguously decomposed as $w=u_1u_2\cdots$ with 
  $u_i\in K$ for all $i\geq1$, where $K\subseteq\Sigma^{*}$ is regular.
  
  \item An \RTE is formed using all the above basic combinators.
  
  \item As an example, consider the \RTE $C={\lrplus{C_4}\lrdot C_2^{\omega}}$
  with $C_4=\Ifthenelse{((a+b)^*\#)}{(\rlplus{C_3} \odot \lrplus{C_1})}{\bot}$,
  $C_2=\Ifthenelse{a}{a}{(\Ifthenelse{b}{b}{\bot})}$.  Here,
  $C_1=\Ifthenelse{a}{a}{(\Ifthenelse{b}{b}{(\Ifthenelse{\#}{\#}{\bot})})}$, and
  $C_3=\Ifthenelse{a}{a}{(\Ifthenelse{b}{b}{(\Ifthenelse{\#}{\epsilon}{\bot})})}$.
  Then $\dom{C_1}=\dom{C_3}=(a+b+\#)$, $\dom{C_2}=(a+b)$, $\dom{C_4}=(a+b)^*\#$
  and, for $u\in(a+b)^*$, $\sem{C_4}(u\#)=u^Ru\#$ where $u^R$ denotes the
  reverse of $u$.  This gives $\dom{C} = [(a+b)^*\#]^+(a+b)^{\omega}$ with
  $\sem{C}(u_1\#u_2\#\cdots u_n\#v)=u_1^Ru_1\#u_2^Ru_2\# \cdots \#u_n^Ru_n\#v$
  when $u_i \in (a+b)^*$ and $v \in (a+b)^{\omega}$.  The \RTE
  $C'=\Ifthenelse{(a+b)^{\omega}}{C_2^{\omega}}{C}$ corresponds to the \twoDMTla
  $\Aa'$ in Figure~\ref{fig:intro}; that is, $\sem{C'}=\sem{\Aa'}$.
   
  \item The combinators proposed in \cite{lics14} also require unambiguity in
  concatenation and iteration.  The \emph{base function} $L/d$ in \cite{lics14}
  maps all strings in language $L$ to the constant $d$, and is undefined for
  strings not in $L$.  This can be written using our if-then-else
  $\Ifthenelse{L}{d}{\bot}$.
  The \emph{conditional choice} combinator $f \triangleright g$ of \cite{lics14}
  maps an input $\sigma$ to $f(\sigma)$ if it is in $dom(f)$, and otherwise it
  maps it to $g(\sigma)$.  This can be written in our if-then-else as
  $\Ifthenelse{\dom{f}}{f}{g}$.  The \emph{split-sum} combinator $f \oplus g$ of
  \cite{lics14} is our Cauchy product $f \boxdot g$.  The \emph{iterated sum}
  $\Sigma f$ of \cite{lics14} is our Kleene-plus $f^{\boxplus}$.  The
  \emph{left-split-sum} and \emph{left-iterated sum} of \cite{lics14} are
  counterparts of our reverse Cauchy product $f \rldot g$ and reverse
  Kleene-plus $\rlplus{f}$.  The \emph{sum} $f+g$ of two functions in
  \cite{lics14} is our Hadamard product $f \odot g$.  Finally, the \emph{chained
  sum} $\Sigma(f, L)$ of \cite{lics14} is our two-chained Kleene-plus
  $\twoplus{L}{f}$.  In our case, the terminology is all inspired from weighted
  automata literature, and the unambiguity comes from the use of the unambiguous
  factorization of the domain into good expressions, and we also extend our \RTEs
  to infinite words.
  
\end{itemize}

Our main result is that two-way deterministic transducers and regular 
transducer expressions are effectively equivalent, both for finite and infinite 
words. See Appendix~\ref{app:motiv} for a practical example using transducers.  

\begin{theorem}\label{thm:main-intro}
  (1) Given an \RTE (resp.\ \oRTE) we can effectively construct an equivalent
  2DFT (resp.\ an \twoDMTla).
  Conversely,
  (2) given a 2DFT (resp.\ an \twoDMTla) we can effectively construct an equivalent 
  \RTE (resp.\ \oRTE).
\end{theorem}

The proof of (1) is by structural induction on the \RTE. The construction of an
\RTE starting from a two-way deterministic transducer $\Aa$ is quite involved.
It is based on the transition monoid $\TrMon(\Aa)$ of the transducer.  This is a
classical notion for two-way transducers over finite words, but not for two-way
transducers \emph{with look-ahead} on infinite words (to the best of our
knowledge).  So we introduce the notion of transition monoid for \twoDMTla.  We
handle the look-ahead with a backward deterministic B\"uchi automaton (\BDBA),
also called \emph{complete unambiguous} or \emph{strongly unambiguous} B\"uchi
automata \cite{Carton:2003rt,Wilke-fsttcs17}.  The translation of $\Aa$ to an \RTE is
crucially guided by a ``good'' rational expression induced by the transition
monoid of $\Aa$.  These ``good'' expressions are obtained thanks to an
unambiguous version \cite{GastinKrishna-UFF} of the celebrated forest
factorization theorem due to Imre Simon~\cite{Simon_1990}.  The unambiguous forest
factorization theorem implies that, given a two-way transducer $\Aa$, any input
word $w$ in the domain of $\Aa$ can be factorized unambiguously
following a ``good'' rational expression induced by the transition monoid of
$\Aa$.  This unambiguous factorization then guides the construction of the \RTE
corresponding to $\Aa$.  This algebraic backdrop facilitates a uniform treatment
in the case of infinite words and finite words.  As a remark, it is not apriori
clear how the result of \cite{lics14} extends to infinite words using the
techniques therein.
  
\smallskip\noindent {\bf Goodness of Rational Expressions}.  The goodness of a
rational expression over alphabet $\Sigma$ is defined using a morphism $\varphi$
from $\Sigma^*$ to a monoid $(S,.,1_S)$.  A rational expression $F$ is good iff
(i) it is unambiguous and (ii) for each subexpression $E$ of $F$, the image of
all strings in $L(E)$ maps to a single monoid element $s_E$, and (iii) for each
subexpression $E^+$ of $F$, $s_E$ is an idempotent.  Note that unambiguity
ensures the functionality of the output computed.
Good rational expressions might be useful in settings beyond two-way transducers.  
 
\smallskip\noindent {\bf Computing the \RTE}.  
As an example, we now show how one computes an \RTE equivalent to the 2DFT 
$\Aa$ on the left of Figure~\ref{fig:intro}.
\begin{enumerate}[nosep]
  \item We work with the morphism $\TrMorph\colon\Sigma^*\to\TrMon$ which maps
  words $w \in \Sigma^*$ to the transition monoid $\TrMon$ of $\Aa$.  An element
  $X \in \TrMon$ is a set consisting of triples $(p, d, q)$, where $d$ is a
  direction $\{\leftleft, \rightright, \leftright, \rightleft\}$.  Given a word
  $w \in \Sigma^*$, a triple $(p,\leftleft,q) \in \TrMorph(w)$ iff 
  when starting in state $p$ on the left most symbol of $w$, the run of  
  $\Aa$ leaves $w$ on the left in state $q$.
  The other directions
  $\rightright$ (start at the rightmost symbol of $w$ in state $p$ and leave $w$
  on the right in state $q$), $\rightleft$ and $\leftright$ are similar.  In 
  general, we have
  $w\in\dom{\Aa}$ iff 
  on input $\leftend\, w\, \rightend$, starting on $\leftend$ in the initial 
  state of $\Aa$, the run exits on the right of $\rightend$ in some 
  final state of $\Aa$.
  With the automaton $\Aa$ on the left of Figure~\ref{fig:intro} we have
  $w\in\dom{\Aa}$ iff $(q_0, \leftright, q_2) \in \TrMorph(w)$.

  \item For each $X\in\TrMon$ such that $(q_0,\leftright,q_2)\in X$, we find an
  \RTE $C_X$ whose domain is $\TrMorph^{-1}(X)$ and such that
  $\sem{\Aa}(w)=\sem{C_X}(w)$ for all $w\in\TrMorph^{-1}(X)$.  The \RTE
  corresponding to $\sem{\Aa}$ is the disjoint union of all these \RTEs and is
  written using the if-then-else construct iterating over for all such elements
  $X$. For instance, if the monoid elements containing $(q_0,\leftright,q_2)$ 
  are $X_1,X_2,X_3$ then we set $C=\Ifthenelse{\TrMorph^{-1}(X_1)}{C_{X_1}}{(
  \Ifthenelse{\TrMorph^{-1}(X_2)}{C_{X_2}}{(
  \Ifthenelse{\TrMorph^{-1}(X_3)}{C_{X_3}}\bot)})}$ where $\bot$ stands for 
  a nowhere defined function, i.e., $\dom{\bot}=\emptyset$.

  \item Consider the language $L=(ba^{+})^+b \subseteq \dom{\Aa}$.  Notice that
  the regular expression $(ba^{+})^{+}b$ is not ``good''.  For instance,
  condition (ii) is violated since $\TrMorph(bab)\neq\TrMorph(babab)$.  Indeed,
  we can seen in Figure~\ref{fig:run} that if we start on the right of $bab$ in
  state $q_3$ then we exist on the left in state $q_5$:
  $(q_3,\rightleft,q_5)\in\TrMorph(bab)$.  On the other hand, if we start on the
  right of $babab$ in state $q_3$ then we exist on the right in state $q_2$:
  $(q_3,\rightright,q_2)\in\TrMorph(babab)$. Also,
  $(q_5,\leftright,q_1)\in\TrMorph(bab)$ while
  $(q_5,\leftright,q_2)\in\TrMorph(babab)$.
  It can be seen that $\TrMorph(a)$\footnote{$\TrMorph(a)=\{
  (q_1, \leftright, q_1), (q_1, \rightright, q_1),
  (q_2, \leftright, q_3), (q_2, \rightright, q_3),
  (q_3, \leftright, q_3), (q_3, \rightright, q_3), 
  (q_4, \rightleft, q_4), (q_4, \leftleft, q_4), \\
  (q_5, \rightleft, q_5), (q_5, \leftleft, q_5),
  (q_6, \leftright, q_6), (q_6, \rightright, q_6)\}$}
  is an idempotent, hence $\TrMorph(a^{+})=\TrMorph(a)$.  We deduce also
  $\TrMorph(ba^{+}b)=\TrMorph(bab)$\footnote{$\TrMorph(ba^{+}b)=\{
  (q_0, \leftright, q_2), (q_0, \rightright, q_1),
  (q_1, \leftleft, q_5), (q_1, \rightright, q_2), 
  (q_2, \leftleft, q_4),(q_2, \rightleft, q_5),
  (q_3, \leftleft, q_4), (q_3, \rightleft, q_5), \\
  (q_4, \leftleft, q_5), (q_4, \rightright, q_1),
  (q_5, \leftright, q_1), (q_5, \rightright, q_6),
  (q_6, \leftright, q_2), (q_6, \rightright, q_1)\}$}.
  Finally, we have $\TrMorph((ba^{+})^{n}b)=\TrMorph(babab)$%
  \footnote{$\TrMorph(ba^{+}ba^{+}b)=\{
  (q_0, \leftright, q_2), (q_0, \rightright, q_1), 
  (q_1, \leftleft, q_5), (q_1, \rightright, q_2),
  (q_2, \leftleft, q_4), (q_2, \rightright, q_2),
  (q_3, \leftleft, q_4), (q_3, \rightright, q_2), \\
  (q_4, \leftleft, q_5), (q_4, \rightright, q_1),
  (q_5, \leftright, q_2), (q_5, \rightright, q_6),
  (q_6, \leftright, q_2), (q_6, \rightright, q_1)\}$} for all $n\geq2$.  
  Therefore, to obtain the \RTE corresponding to $L$, we compute \RTEs
  corresponding to $ba^{+}b$ and $(ba^{+})^{+}ba^{+}b$ satisfying conditions 
  (i) and (ii) of ``good'' rational expressions.

  \item  While $ba^{+}b$ is good since $\TrMorph(a)$ is an idempotent,  
  $(ba^{+})^{+}ba^{+}b$
  is not good, the reason being that $\TrMorph(ba^{+})$ is not an idempotent. 
  We can check that $\TrMorph(ba^{+}ba^{+})$%
  \footnote{$\TrMorph(ba^{+}ba^{+})=\{(q_0, \leftright, q_3), 
  (q_1, \leftleft, q_5), (q_1, \rightright, q_1),
  (q_2, \leftleft, q_4), (q_2, \rightright, q_3),
  (q_3, \leftleft, q_4), (q_3, \rightright, q_3), \\
  (q_4, \leftleft, q_5), (q_4, \rightright, q_1),
  (q_5, \leftright, q_1), (q_5, \rightright, q_6),
  (q_6, \leftright, q_3), (q_6, \rightright, q_6)\}$} 
  is still not idempotent, while $\TrMorph((ba^{+})^{i})=\TrMorph((ba^{+})^3)$  for all $i\geq 3$, 
  (see Figure~\ref{fig:run}: we only need to argue for
  $(q_0, \leftright, q_3), (q_5, \leftright, q_3)$ and $(q_6, \leftright, q_3)$
  in $\TrMorph((ba)^{i})$, $i \geq 3$, all other entries trivially carry over).
  In particular, $\TrMorph((ba^{+})^3)$ \label{idempotent}
  is an idempotent%
  \footnote{$\TrMorph((ba^{+})^3)=\{(q_0, \leftright, q_3),
  (q_1, \leftleft, q_5), (q_1, \rightright, q_1),
  (q_2, \leftleft, q_4), (q_2, \rightright, q_3),
  (q_3, \leftleft, q_4), (q_3, \rightright, q_3), \\
  (q_4, \leftleft, q_5), (q_4, \rightright, q_1),
  (q_5, \leftright, q_3), (q_5, \rightright, q_6),
  (q_6, \leftright, q_3), (q_6, \rightright, q_6)\}$}.
  Thus, to compute the \RTE for $L=(ba^{+})^+b$, we consider the \RTEs
  corresponding to the ``good'' regular expressions $E_1=ba^{+}b$,
  $E_2=ba^{+}ba^{+}b$, $E_3=[(ba^{+})^3]^+ b$, $E_4=[(ba^{+})^3]^+ ba^{+}b$ and
  $E_5=[(ba^{+})^3]^+ ba^{+}ba^{+}b$.

  \begin{figure}[h]
    \centerline{
    \includegraphics[scale=0.20]{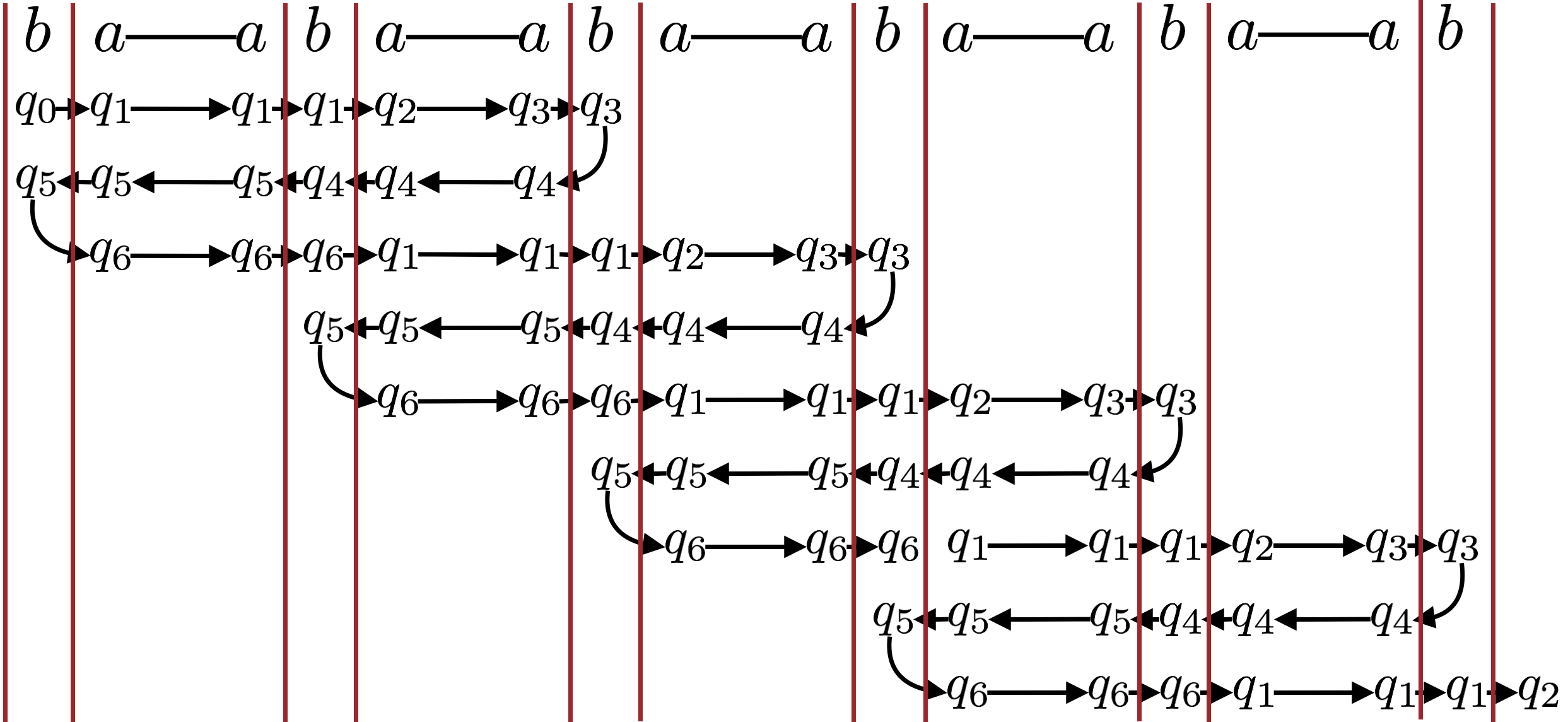}}
    \caption{Run of $\Aa$ on an input word in $(ba^{+})^{+}b$.}
    \label{fig:run}
  \end{figure}

  \item  We define by induction, for each ``good''
  expression $E$ and ``step'' $x=(p,d,q)$ in the monoid element $X=\TrMorph(E)$
  associated with $E$, an \RTE $C_E(x)$ whose domain is $E$ and, given a word
  $w\in E$, it computes $\sem{C_E(x)}(w)$ the output of $\A$ when running step
  $x$ on $w$.  For instance, if $E=a$ and $x=(q_5,\rightleft,q_5)$ the output is
  $b$ so we set $C_a(q_5,\rightleft,q_5)=(\Ifthenelse{a}{b}\bot)$.  The
  if-then-else ensures that the domain is $a$.  Similarly, we get the \RTE
  associated with all atomic expressions and steps.  For instance,
  $C_b(q_1,\leftright,q_2)=(\Ifthenelse{b}{\varepsilon}\bot)
  =C_b(q_3,\leftleft,q_4)$.  For $u,v\in\Sigma^{*}$, we introduce the macro
  $u/v=\Ifthenelse{u}{v}\bot$. 
  We have $\dom{u/v}=\{u\}$ and $\sem{u/v}(u)=v$.
  
  We turn to the good expression $a^{+}$.  If we start on the right of a
  word $w\in a^{+}$ from state $q_5$ then we read the word from right to left 
  using always the step $(q_5,\rightleft,q_5)$. Therefore, we have
  $C_{a^{+}}(q_5,\rightleft,q_5)=\rlplus{(C_a(q_5,\rightleft,q_5))}=\rlplus{(a/b)}$.
  Similarly, 
  $C_{a^{+}}(q_4,\rightleft,q_4)=\rlplus{(a/a)}$,
  $C_{a^{+}}(q_1,\leftright,q_1)=\lrplus{(a/\varepsilon)}
  =C_{a^{+}}(q_6,\leftright,q_6)$. 
  Now if we start on the left of a word $w\in a^{+}$ from state $q_2$ then we
  first take the step $(q_2,\leftright,q_3)$ and 
  then we iterate the step $(q_3,\leftright,q_3)$.  Therefore, we have
  $C_{a^{+}}(q_2,\leftright,q_3)=\Ifthenelse{a}{C_a(q_2,\leftright,q_3)}
  {(C_a(q_2,\leftright,q_3)\lrdot\lrplus{(C_a(q_3,\leftright,q_3))})}=
  \Ifthenelse{a}{(a/\varepsilon)}
  {\big((a/\varepsilon)\lrdot\lrplus{(a/\varepsilon)}\big)}$,
  which is equivalent to the \RTE $\lrplus{(a/\varepsilon)}$.
  
  We consider now $E=ba^{+}ba^{+}$ and the step $x=(q_0,\leftright,q_3)$. We 
  have (see Figure~\ref{fig:run})
  \begin{align*}
    C_E(x) &=
    C_b(q_0,\leftright,q_1)\lrdot C_{a^{+}}(q_1,\leftright,q_1)\lrdot
    C_b(q_1,\leftright,q_2)\lrdot C_{a^{+}}(q_2,\leftright,q_3)
    \\
    &=(b/\varepsilon)\lrdot\lrplus{(a/\varepsilon)}\lrdot
    (b/\varepsilon)\lrdot\lrplus{(a/\varepsilon)} 
    \\
    &\approx (\Ifthenelse{ba^{+}ba^{+}}{\varepsilon}\bot) \,.
  \end{align*}
  More interesting is the step $y=(q_4,\rightright,q_1)$ \label{step:y}
  since on a word $w\in 
  E$, the run which starts on the right in state $q_4$ goes all the way to the 
  left until it reads the first $b$ in state $q_5$ and then moves to the right 
  until it exists in state $q_1$ (see Figure~\ref{fig:run}). Therefore, we have
  \begin{align*}
    C_E(y) &=
    \big((b/\varepsilon)\rldot C_{a^{+}}(q_5,\rightleft,q_5)\rldot
    C_b(q_4,\rightleft,q_5)\rldot C_{a^{+}}(q_4,\rightleft,q_4)\big) \odot{}
    \\
    &\hspace{4mm}
    \big(C_b(q_5,\leftright,q_6)\lrdot C_{a^{+}}(q_6,\leftright,q_6)\lrdot
    C_b(q_6,\leftright,q_1)\lrdot C_{a^{+}}(q_1,\leftright,q_1)\big) 
    \\
    &= \big((b/\varepsilon)\rldot\rlplus{(a/b)}\rldot
    (b/\varepsilon)\rldot\rlplus{(a/a)}\big) \odot
    \big((b/\varepsilon)\lrdot\lrplus{(a/\varepsilon)}\lrdot
    (b/\varepsilon)\lrdot\lrplus{(a/\varepsilon)}\big)
    \\
    &\approx (b/\varepsilon)\rldot\rlplus{(a/b)}\rldot
    (b/\varepsilon)\rldot\rlplus{(a/a)}
    \,.
  \end{align*}
  The leftmost $(b/\varepsilon)$ in the first line is used to make sure that 
  the input word belongs to $E=ba^{+}ba^{+}$.
  Composing these steps on the right with $b$, we obtain the \RTE 
  $C_2=C_{E_2}(q_0,\leftright,q_2)$ which describes the behaviour of $\Aa$ on the 
  subset $E_2=ba^{+}ba^{+}b\subseteq\dom{\Aa}$:
  \begin{align*}
    C_2 &= \big( C_E(x) \lrdot C_b(q_3,\leftleft,q_4) \big) \odot 
    \big( C_E(y) \lrdot C_b(q_1,\leftright,q_2) \big)
    \\
    &= \big( C_E(x) \lrdot (b/\varepsilon) \big) \odot 
    \big( C_E(y) \lrdot (b/\varepsilon) \big) 
    \\
    &\approx \big( (b/\varepsilon)\rldot\rlplus{(a/b)}\rldot
    (b/\varepsilon)\rldot\rlplus{(a/a)} \big) \lrdot (b/\varepsilon) \,.  
  \end{align*}
  Therefore, $\sem{C_2}(ba^{m_1}ba^{m_2}b)=a^{m_2}b^{m_1}=\sem{\Aa}(ba^{m_1}ba^{m_2}b)$.
  
  In Appendix~\ref{app:intro}, we show the computation of the \RTE
  $C_{E_3}(q_0,\leftright,q_2)$ for $E_3=[(ba^{+})^{3}]^{+}b \subseteq
  \dom{\Aa}$.
\end{enumerate}

\section{Finite Words}
\label{sec:fin}

We start with the definition of two-way automata and transducers for the case of
finite words.

\subsection{Two-way automata and transducers}

Let $\Sigma$ be a finite input
alphabet and let $\leftend, \rightend$ be two special symbols not in $\Sigma$.
We assume that every input string $w \in \Sigma^*$ is presented as $\leftend w
\rightend$, where $\leftend, \rightend$ serve as left and right delimiters that
appear nowhere else in $w$.  We write
$\Sigma_{\leftend\rightend}=\Sigma\cup\{\leftend,\rightend\}$.  A two-way
automaton $\Aa=(Q,\Sigma,\delta,I,F)$ has a finite set of states $Q$,
subsets $I,F\subseteq Q$ of initial and final states 
and a transition relation $\delta
\subseteq Q \times \Sigma_{\leftend\rightend} \times Q \times
\{-1,1\}$.  The -1 represents the reading head moving to the left, while a 1
represents the reading head moving to the right.
The reading head cannot move left when it is on $\leftend$.
A configuration of $\Aa$ is represented by $w_1 q w_2$ where $q \in Q$ and
$w_1w_2 \in {\leftend \Sigma^* \rightend}$.  If $w_2=\varepsilon$ the
computation has come to an end.  Otherwise, the reading head of $\Aa$ is
scanning the first symbol of $w_2\neq\varepsilon$ in state $q$.  If $w_2=aw'_2$
and if $(q,a,q',-1)\in\delta$ (hence $a\neq\leftend$), then there is a
transition from the configuration $w'_1b q a w'_2$ to $w'_1 q' b a w'_2$.
Likewise, if $(q,a,q',1)\in\delta$, we obtain a transition from $w_1 q a w'_2$
to $w_1 a q' w'_2$.  A run of $\Aa$ is a sequence of transitions; it is
accepting if it starts in a configuration $p \leftend w\rightend$ with $p\in I$
and ends in a configuration $\leftend w \rightend q$ with $q \in F$.
The language $\Lang{\Aa}$ or domain $\dom{\Aa}$ of $\Aa$ is the set of all words
$w\in\Sigma^*$ which have an accepting run in $\Aa$.

To extend the definition of a two-way automaton $\Aa$ into a two-way transducer,
$(Q, \Sigma, \delta, I, F)$ is extended to $(Q, \Sigma, \Gamma, \delta, I, F)$
by adding a finite output alphabet $\Gamma$ and the definition of the transition
relation as a \emph{finite} subset $\delta \subseteq Q \times
\Sigma_{\leftend\rightend} \times Q \times \Gamma^* \times \{-1,1\}$.  The
output
produced on each transition is appended to the right of the output produced so
far.  $\Aa$ defines a relation $\sem{\Aa}=\{(u,w) \mid u \in \Lang{\Aa}$ and $w$
is the output produced on an accepting run of $u\}$.

The transducer $\Aa$ is said to be functional if for each input $u\in\dom{\Aa}$,
at most one output $w$ can be produced.  In this case, for each $u\in\dom{\Aa}$
in the domain, there is exactly one $w\in\Gamma^*$ such that
$(u,w)\in\sem{\Aa}$.  We also denote this by $\sem{\Aa}(u)=w$.
We consider a special symbol $\bot\notin\Gamma$ that will stand for
\emph{undefined}.  We let $\sem{\Aa}(u)=\bot$ when $u\notin\dom{\Aa}$.
Thus, the semantics of a functional transducer $\Aa$ is a map 
$\sem{\Aa}\colon\Sigma^*\to\D=\Gamma^*\cup\{\bot\}$ such that 
$u\in\dom{\Aa}$ iff $\sem{\Aa}(u)\neq\bot$.

We use non-deterministic unambiguous two-way transducers
(2NUFT) in some proofs.  A two-way transducer is unambiguous if each string $u \in \Sigma^*$
has at most one accepting run.  Clearly, 2NUFTs are functional.  A deterministic
two-way transducer (2DFT) is one having a single initial state and where, from
each state, on each symbol $a\in\Sigma_{\leftend\rightend}$, at most one
transition is enabled.  In that case, the transition relation is a partial 
function $\delta\colon Q\times \Sigma_{\leftend\rightend}
\to Q \times \Gamma^* \times \{-1,1\}$.
2DFTs are by definition unambiguous.  It is known
\cite{Chytil1977} that 2DFTs are equivalent to 2NUFTs.

A 1DFT (1NUFT) represents a deterministic (non-deterministic unambiguous)
transducer where the reading head only moves to the right.

\begin{example}
  On the left of Figure~\ref{fig:intro}, a two-way transducer $\A$ is given with
  $\dom{\A}=(ba^{*})^{+}b$, $\sem{\A}(ba^{m_1}b)=\varepsilon$ and
  $\sem{\A}(ba^{m_1}ba^{m_2}b \cdots a^{m_k}b)= a^{m_2}b^{m_1}a^{m_3}b^{m_2}
  \cdots a^{m_k}b^{m_{k-1}}$ for $k\geq2$.
\end{example}

\subsection{Regular Transducer Expressions}\label{sec:RegExp}

Let $\Sigma$ and $\Gamma$ be finite input and output alphabets.  Recall that
$\bot\notin\Gamma$ is a special symbol that stands for \emph{undefined}.  We
define the output monoid as $\D=\Gamma^*\cup\{\bot\}$ with the usual
concatenation on words, $\bot$ acting as a zero: $d\cdot\bot=\bot\cdot
d=\bot$ for all $d \in \D$.  The unit is the empty word $\unitD=\varepsilon$.

We define \emph{Regular Transducer Expressions} (\RTE) from $\Sigma^*$
to $\D$ using some basic combinators.   
The syntax of \RTE is defined with the following grammar:
$$
C ::= d \mid \Ifthenelse{K}{C}{C} \mid C\odot C \mid C\lrdot C \mid C \rldot C
\mid \lrplus{C} \mid \rlplus{C} 
\mid \twoplus{K}{C} \mid \rltwoplus{K}{C}
$$
where $d\in\D$ ranges over output values, and $K\subseteq\Sigma^*$ ranges over
regular languages of \emph{finite words}.
The semantics of an \RTE $C$ is a function $\sem{C}\colon\Sigma^*\to\D$ defined
inductively following the syntax of the expression, starting from constant
functions.  Since $\bot$ stands for \emph{undefined}, we define the
\emph{domain} of a function $f\colon\Sigma^*\to\D$ by
$\dom{f}=f^{-1}(\D\setminus\{\bot\})=\Sigma^*\setminus f^{-1}(\bot)$.

\begin{description}
  \item[Constants.] For $d\in\D$, we let $\sem{d}$ be the constant map defined by 
  $\sem{d}(w)=d$ for all $w\in\Sigma^*$.
  
  We have $\dom{\sem{d}}=\Sigma^*$ if $d\neq\bot$ and $\dom{\sem{\bot}}=\emptyset$.
\end{description}
Each regular combinator defined above allows to combine functions from $\Sigma^*$ to
$\D$.  
For functions $f,g\colon\Sigma^*\to\D$, $w \in \Sigma^*$ and 
a regular language $K\subseteq\Sigma^*$, we define the following combinators. 
\begin{description}
  \item[If then else.]  $(\Ifthenelse{K}{f}{g})(w)$ is defined as
  $f(w)$ for $w\in K$, and $g(w)$ for $w\notin K$.
  
  We have $\dom{\Ifthenelse{K}{f}{g}}=(\dom{f}\cap K)\cup(\dom{g}\setminus K)$.

  \item[Hadamard product.]$(f \odot g)(w)=f(w)\cdot g(w)$ (recall that 
  $(\D,\cdot,\unitD)$ is a monoid).
  
  We have $\dom{f\odot g}=\dom{f}\cap\dom{g}$.

  \item[Unambiguous Cauchy product and its reverse.] If $w$ admits a unique 
  factorization $w=u\cdot v$ with $u\in\dom{f}$ and $v\in\dom{g}$ then we set 
  $(f\lrdot g)(w)=f(u)\cdot g(v)$ and $(f\rldot g)(w)=g(v)\cdot f(u)$. Otherwise, we set 
  $(f\lrdot g)(w)=\bot=(f\rldot g)(w)$.
  
  We have $\dom{f\lrdot g}=\dom{f\rldot g}\subseteq\dom{f}\cdot\dom{g}$ and the 
  inclusion is strict if the concatenation of $\dom{f}$ and $\dom{g}$ is 
  ambiguous.

  \item[Unambiguous Kleene-plus and its reverse.]  If $w$ admits a unique
  factorization $w=u_1\cdot u_2 \cdots u_n$ with $n\geq1$ and $u_i\in\dom{f}$
  for all $1\leq i\leq n$ then we set $\lrplus{f}(w)=f(u_1)\cdot f(u_2) \cdots
  f(u_n)$ and $\rlplus{f}(w)=f(u_n)\cdots f(u_2)\cdot f(u_1)$.  Otherwise, we
  set $\lrplus{f}(w)=\bot=\rlplus{f}(w)$.
  
  We have $\dom{\lrplus{f}}=\dom{\rlplus{f}}\subseteq\dom{f}^+$ and the
  inclusion is strict if the Kleene iteration $\dom{f}^+$ of $\dom{f}$ is
  ambiguous. Notice that $\dom{\lrplus{f}}=\emptyset$ when 
  $\varepsilon\in\dom{f}$.

  \item[Unambiguous 2-chained Kleene-plus and its reverse.]  If $w$ admits a
  unique factorization $w=u_1\cdot u_2 \cdots u_n$ with $n\geq1$ and $u_i\in K$
  for all $1\leq i\leq n$ then we set $\twoplus{K}{f}(w)=f(u_1u_2)\cdot
  f(u_2u_3) \cdots f(u_{n-1}u_n)$ and $\rltwoplus{K}{f}(w)=f(u_{n-1}u_n)\cdots
  f(u_2u_3)\cdot f(u_1u_2)$ (if $n=1$, the empty product gives the unit of $\D$:
  $\twoplus{K}{f}(w)=\unitD=\rltwoplus{K}{f}(w)$).  Otherwise, we set
  $\twoplus{K}{f}(w)=\bot=\rltwoplus{K}{f}(w)$.

  Again, we have $\dom{\twoplus{K}{f}}=\dom{\rltwoplus{K}{f}}\subseteq K^+$ and
  the inclusion is strict if the Kleene iteration $K^+$ of $K$ is ambiguous.
  Notice that, even if $w\in K^+$ admits a unique factorization $w=u_1\cdot u_2
  \cdots u_n$ with $u_i\in K$ for all $1\leq i\leq n$, $w$ is not necessarily 
  in the domain of $\twoplus{K}{f}$ or $\rltwoplus{K}{f}$. For $w$ to be in 
  this domain, it is further required that 
  $u_1u_2,u_2u_3,\ldots,u_{n-1}u_n\in\dom{f}$.
  Notice that we have $\dom{\twoplus{K}{f}}=\dom{\rltwoplus{K}{f}}=K^+$ when
  $K^+$ is unambiguous and $K^2\subseteq\dom{f}$.
\end{description}

\begin{lemma}\label{lem:dom-regular}
  The domain of an \RTE $C$ is a regular language $\dom{C}\subseteq\Sigma^*$.
\end{lemma}

\begin{remark}
  Notice that the reverse Cauchy product is redundant, it can be expressed 
  with the Hadamard product and the Cauchy product:
  $$
  f\rldot g =
  ((\Ifthenelse{\dom{f}}{\varepsilon}\bot)\lrdot g) \odot
  (f\lrdot(\Ifthenelse{\dom{g}}{\varepsilon}\bot)) \,.
  $$
  The unambiguous Kleene-plus is also redundant, it can be expressed with the 
  unambiguous 2-chained Kleene-plus:
  $$
  \lrplus{f} = 
  \twoplus{\dom{f}}{f\lrdot(\Ifthenelse{\dom{f}}{\varepsilon}\bot)}
  \odot (\Ifthenelse{(\dom{f}^{*}}{\varepsilon}\bot)\lrdot f) \,.
  $$
\end{remark}

\begin{example}
	\label{finiteRTE-eg}
  Consider the \RTEs 
  $C_1{=}{(\Ifthenelse{[(a+b)^+\#]}{\varepsilon}\bot){\lrdot}(\Ifthenelse{(a+b)^+}{\mathsf{copy}}\bot)}$,
  $C_2=\#$ and
  $C_3=(\Ifthenelse{(a+b)^+}{\mathsf{copy}}\bot){\lrdot}(\Ifthenelse{[\#(a+b)^+]}\varepsilon\bot)$, 
  where $\mathsf{copy}=\lrplus{(\Ifthenelse{a}{a}(\Ifthenelse{b}{b}{\bot}))}$.
  Then, $\dom{\sem{C_2}}=\Sigma^*$, $\dom{\sem{\mathsf{copy}}}=(a+b)^+$ and
  $\dom{\sem{C_1}}=\dom{\sem{C_3}}=(a+b)^+\#(a+b)^+$.
  Moreover, $\sem{C_1\odot C_2 \odot C_3}(u\#v)=v\#u$ for all 
  $u,v\in(a+b)^+$. 
\end{example}

\begin{example}
  Consider the \RTEs 
  $C_a=(\Ifthenelse{b}\varepsilon\bot)\lrdot\lrplus{(\Ifthenelse{a}{a}\bot)}$ and
  $C_b=(\Ifthenelse{b}\varepsilon\bot)\lrdot\lrplus{(\Ifthenelse{a}{b}\bot)}$.
  We have $\dom{\sem{C_a}}=ba^+=\dom{\sem{C_b}}$ and $\sem{C_a}(ba^n)=a^n$ and 
  $\sem{C_b}(ba^n)=b^n$.
  We deduce that $\dom{\sem{C_b\rldot C_a}}=ba^+ba^+$ and 
  $\sem{C_b\rldot C_a}(ba^nba^m)=a^mb^n$.

  Consider the expression $C=\twoplus{ba^+}{C_b\rldot 
  C_a}\lrdot(\Ifthenelse{b}\varepsilon\bot)$. 
  Then, $\dom{\sem{C}}=(ba^+)^+b$, $\sem{C}(ba^mb)=\varepsilon$ and
  $\sem{C}(ba^{m_1}ba^{m_2}b\cdots a^{m_k}b)=a^{m_2}b^{m_1}a^{m_3}b^{m_2}\cdots
  a^{m_k}b^{m_{k-1}}$ for $k\geq2$.
\end{example}

\begin{theorem}\label{thm:2DFT=RTE}
  2DFTs and \RTEs define the same class of functions. More precisely,
  \begin{enumerate}[nosep]
    \item\label{item:RTEto2DFT} given an \RTE $C$, we can construct a 2DFT $\A$ such that 
    $\sem{\A}=\sem{C}$,

    \item\label{item:2DFTtoRTE} given a 2DFT $\A$, we can construct an \RTE C
    such that $\sem{\A}=\sem{C}$.
  \end{enumerate}
\end{theorem}

The proof of \eqref{item:RTEto2DFT} is given in the next section, while the
proof of \eqref{item:2DFTtoRTE} will be given in Section~\ref{sec:2DFT-RTE}
after some preliminaries in Section~\ref{sec:TrM-2NFA} on transition monoids for
2DFTs and the unambiguous forest factorization theorem.

\subsection{\RTE to 2DFT}

In this section, we prove Theorem~\ref{thm:2DFT=RTE}\eqref{item:RTEto2DFT}, 
i.e., we show that given an \RTE $C$, we can construct a 2DFT $\Aa$
such that $\sem{\Aa}=\sem{C}$.  We do this by structural induction on \RTEs,
starting with constant functions, and then later showing that 2DFTs are closed
under all the combinators used in \RTEs.
  
\noindent
\textbf{Constant functions:} We start with the constant function $d\in\D$ for
which it is easy to construct a 2DFT $\Aa$ such that $\sem{d}=\sem{\Aa}$.  For
$d=\bot$, we take $\Aa$ such that $\dom{\Aa}=\emptyset$ (for instance we use a
single state and an empty transition function).  Assume now that
$d\in\Gamma^*$.  The 2DFT scans the word up to the right end marker, outputs
$d$ and stops.  Formally, we let $\Aa = (\{q\},\Sigma,\Gamma,\delta,q,\{q\})$
s.t.\ $\delta(q,a) = (q, \varepsilon, +1)$ for all $a\in\Sigma \cup \{\leftend\}$ and
$\delta(q,\rightend) = (q,d,+1)$.  Clearly, $\sem{\Aa}(w)=d$ for all
$w\in\Sigma^*$.

The inductive steps follow directly from:

\begin{lemma}\label{lem:induction-RTE-to-2DFT}
  Let $K\subseteq\Sigma^*$ be regular, and let $f$ and
  $g$ be \RTEs with $\sem{f}=\sem{M_f}$ and $\sem{g}=\sem{M_g}$ for 2DFTs $M_f$
  and $M_g$ respectively. Then, one can construct
  \begin{enumerate}
    \item\label{item:ifthenelse} a 2DFT $\Aa$ such that $\sem{\Ifthenelse{K}{f}g}=\sem{\Aa}$.  
    
    \item\label{item:Hadamard} a 2DFT $\Aa$ such that $\sem{\Aa}=\sem{f \odot g}$.  
    
    \item\label{item:Cauchy} 2DFTs $\Aa$, $\Bb$ such that $\sem{\Aa}=\sem{f \lrdot g}$ and
    $\sem{\Bb}=\sem{f \rldot g}$.  
    
    \item\label{item:Kleene} 2DFTs $\Aa$, $\Bb$ such that $\sem{\Aa}=\sem{\lrplus{f}}$ and
    $\sem{\Bb}=\sem{\rlplus{f}}$.  
   
    \item\label{item:twoplus} 2DFTs $\Aa$, $\Bb$ such that $\sem{\Aa}=\sem{\twoplus{K}{f}}$ and
    $\sem{\Bb}=\sem{\rltwoplus{K}{f}}$.  
  \end{enumerate}    
\end{lemma}
  
\begin{proof}
  \eqref{item:ifthenelse} \textbf{If then else.}
  Let $\mathcal{B}$ be a complete DFA that accepts the regular language $K$.
  The idea of the proof is to construct a 2DFT $\Aa$ which first runs
  $\mathcal{B}$ on the input $w$ until the end marker $\rightend$ is reached in
  some state $q$ of $\mathcal{B}$.  Then, $w\in K$ iff $q\in F$ is some
  accepting state of $\mathcal{B}$. The automaton $\Aa$ moves left all
  the way to $\leftend$, and starts running either $M_f$ or $M_g$ depending on 
  whether $q\in F$ or not. 
  Since $\mathcal{B}$ is complete, it is clear that
  $\dom{\Aa}=\dom{\Ifthenelse{K}{f}g}$ and the output of $\Aa$ coincides with
  $\sem{M_f}$ iff the input is in $K$, and otherwise coincides with $\sem{M_g}$.

  \medskip\noindent\eqref{item:Hadamard} \textbf{Hadamard product.}
  Given an input $w$, the constructed 2DFT $\Aa$ first runs $M_f$.  Instead of
  executing a transition $p\xrightarrow{\rightend/\gamma,+1}q$ with $q$ a final
  state of $M_f$, it executes $p\xrightarrow{\rightend/\gamma,-1}\mathsf{reset}$
  where $\mathsf{reset}$ is a new state.  While in the $\mathsf{reset}$ state,
  it moves all the way back to $\leftend$ and it starts running $M_g$ by
  executing $\mathsf{reset}\xrightarrow{\leftend/\gamma',+1}q'$ if
  $\delta_g(q_0,\leftend) = (q', \gamma', +1)$ where $\delta_g$ is the transition
  function of $M_g$ and $q_0$ is the initial state of $M_g$.  The final states
  of $\Aa$ are those of $M_g$,
  and its initial state is the initial state of $M_f$.  Clearly,
  $\dom{\Aa}=\dom{M_f}\cap\dom{M_g}$ and the output of $\Aa$ is the
  concatenation of the outputs of $M_f$ and $M_g$.

  \medskip\noindent\eqref{item:Cauchy} \textbf{Cauchy product.}
  The domain of a 2DFT is a regular language, accepted by the 2DFA obtained by
  ignoring the outputs.  Since 2DFAs are effectively equivalent to (1)DFAs, we
  can construct from $M_f$ and $M_g$ two DFAs 
  $\mathcal{C}_f=(Q_f,\Sigma,\delta_f,s_f,F_f)$ and 
  $\mathcal{C}_g=(Q_g,\Sigma,\delta_g,s_g,F_g)$
  such that $\Lang{\mathcal{C}_f}=\dom{f}$ and $\Lang{\mathcal{C}_g}=\dom{g}$.
   
  Now, the set $K$ of words $w$ having at least two factorizations
  $w=u_1v_1=u_2v_2$ with $u_1,u_2\in\dom{f}$, $v_1,v_2\in\dom{g}$ and $u_1\neq
  u_2$ is also regular.  This is easy since $K$ can be written as
  $K=\bigcup_{p\in F_f,q\in Q_g} L_p\cdot M_{p,q}\cdot R_q$ where
  \begin{itemize}
    \item $L_p$ is the set of words which admit a run in $\mathcal{C}_f$
    from its initial state to the final state $p\in F_f$,

    \item $M_{p,q}$ is the set of words which admit a run in $\mathcal{C}_f$
    from state $p$ to some final state in $F_f$, and also admit a run in 
    $\mathcal{C}_g$ from its initial state to state $q\in Q_g$,
  
    \item $R_{q}$ is the set of words which admit a run in $\mathcal{C}_g$
    from state $q$ to some final state in $F_g$, and also admit a run in 
    $\mathcal{C}_g$ from its initial state to some final state in $F_g$.
  \end{itemize}
  Therefore, we have $\dom{f\lrdot g}=\dom{f\rldot g}
  =(\dom{f}\cdot\dom{g})\setminus K$ is a regular language and we construct a
  complete DFA $\mathcal{C}=(Q,\Sigma,\delta,q_0,F)$ which accepts this language.
  
  \begin{enumerate}
    \item From $\mathcal{C}_f$, $\mathcal{C}_g$ and $\mathcal{C}$ we construct a
    1NUFT $\mathcal{D}$ such that $\dom{\mathcal{D}}=\dom{f\lrdot g}$ and on an
    input word $w=u\cdot v$ with $u\in\dom{f}$ and $v\in\dom{g}$ it produces the
    output $u\#v$ where $\#\notin\Sigma$ is a new symbol.
    On an input word $w\in\Sigma^*$, the transducer $\mathcal{D}$ runs a copy of
    $\mathcal{C}$.
    Simultaneously, $\mathcal{D}$ runs a copy of $\mathcal{C}_f$ on some prefix
    $u$ of $w$, copying each input letter to the output.  Whenever $\mathcal{C}_f$
    is in a final state after reading $u$, the transducer $\mathcal{D}$ may
    non-deterministically decide to stop running $\mathcal{C}_f$, to output $\#$,
    and to start running $C_g$ on the corresponding suffix $v$ of $w$ ($w=u\cdot
    v$) while copying again each input letter to the output.
    The transducer $\mathcal{D}$ accepts if $\mathcal{C}$ accepts $w$ and
    $\mathcal{C}_g$ accepts $v$. Then, we have 
    $u\in\Lang{\mathcal{C}_f}=\dom{f}$, $v\in\Lang{\mathcal{C}_g}=\dom{g}$ and 
    $w=u\cdot v\in\Lang{\mathcal{C}}=\dom{f\lrdot g}$. The output produced by 
    $\mathcal{D}$ is $u\#v$.
    The only non-deterministic 
    choice in an accepting run of $\mathcal{D}$ is unambiguous since a word
    $w\in\Lang{\mathcal{C}}=\dom{f\lrdot g}$ has a unique
    factorization $w=u\cdot v$ with $u\in\dom{f}$ and $v\in\dom{g}$.
    
    \item We construct a 2DFT $\Tt$ which takes as input words of the form
    $u\#v$ with $u,v\in\Sigma^*$, runs $M_f$ on $u$ and then $M_g$ on $v$.  To
    do so, $u$ is traversed in either direction depending on $M_f$, and the
    symbol $\#$ is interpreted as the right end marker $\rightend$.  We explain
    how $\Tt$ simulates a transition of $M_f$ moving to the right of
    $\rightend$, producing some output $\gamma$ and going to a state $q$.  If
    $q$ is not final, then $\Tt$ moves to the right of $\#$ and then all the way
    to the end and rejects.  If $q$ is final, then $\Tt$ stays on $\#$
    (simulated by moving right and then back left), producing the output
    $\gamma$, but goes to the initial state of $M_g$ instead.  $\Tt$ then runs
    $M_g$ on $v$, interpreting $\#$ as $\leftend$.  When $M_g$ moves to the
    right of $\rightend$, $\Tt$ does the same and accepts iff $M_g$ accepts.
    
    \item In a similar manner, we construct a 2DFT $\Tt'$ which takes as input
    strings of the form $u\#v$, first runs $M_g$ on $v$ and then runs $M_f$ on
    $u$.  Assume that $M_g$ wants to move to the right of $\rightend$ going to
    state $q$.  If $q$ is not final then $\Tt'$ also moves to the right of
    $\rightend$ and rejects.  Otherwise, $\Tt'$ traverses back to $\leftend$ and
    runs $M_f$ on $u$.  When $M_f$ wants to move to the right of $\#$ going to
    some state $q$ and producing $\gamma$, $\Tt'$ moves also to the right of 
    $\#$ producing $\gamma$ and then all the way right producing $\varepsilon$.
    After moving to the right of $\rightend$, it accepts if $q$ is a final state
    of $M_f$ and rejects otherwise.
  \end{enumerate}
  We construct a 2NUFT $\Aa'$ as the composition of $\mathcal{D}$ and $\Tt$.
  The composition of a 1NUFT and a 2DFT is a 2NUFT \cite{Chytil1977}, hence
  $\Aa'$ is a 2NUFT. Moreover, $\sem{\Aa'}=\sem{f\lrdot g}$.  Using the
  equivalence of 2NUFT and 2DFT, we can convert $\Aa'$ into an equivalent 2DFT
  $\Aa$.
  In a similar way, to obtain $\sem{f\rldot g}$, the 2NUFT $\Bb'$ is obtained as
  a composition of $\mathcal{D}$ and $\Tt'$ and is then converted to an 
  equivalent 2DFT $\Bb$.

  \medskip\noindent\eqref{item:Kleene} \textbf{Kleene-plus.}
  The proof is similar to case \eqref{item:Cauchy}.  First, we show
  that $\dom{\lrplus{f}}$ is regular.  Notice that if $\varepsilon\in\dom{f}$
  then $\dom{\lrplus{f}}=\emptyset$, hence we assume below that
  $\varepsilon\notin\dom{f}$.  As in case \eqref{item:Cauchy}, the
  language $K$ of words $w$ having at least two factorizations $w=u_1v_1=u_2v_2$
  with $u_1,u_2\in\dom{f}$, $v_1,v_2\in\dom{f}^*$ and $u_1\neq u_2$ is regular.
  Hence, $K'=\dom{f}^*\cdot K$ is regular and contains all words in $\dom{f}^+$
  having several factorizations as products of words in $\dom{f}$.  We deduce
  that $\dom{\lrplus{f}}=\dom{f}^+\setminus K'$ is regular and we can construct 
  a complete DFA $\mathcal{C}$ recognizing this domain.
  
  As in case \eqref{item:Cauchy}, from $\mathcal{C}_f$ and
  $\mathcal{C}$, we construct a 1NUFT $\mathcal{D}$ which takes as input $w$ and
  outputs $u_1\#u_2\#\cdots\#u_n$ iff there is an unambiguous decomposition of
  $w$ as $u_1u_2 \cdots u_n$, with each $u_i\in\dom{f}$.
  We then construct a 2DFT $\Tt$ that takes as input words of the form
  $u_1\#u_2\#\cdots\#u_n$ with each $u_i\in\Sigma^*$ and runs $M_f$ on each
  $u_i$ from left to right, i.e., starting with $u_1$ and ending with $u_n$.
  The transducer $\Tt$ interprets $\#$ as $\leftend$ (resp.\ $\rightend$) when
  it is reached from the right (resp.\ left).  The simulation by $\Tt$ reading
  $\#$ of a transition of $M_f$ moving to the right of $\rightend$ is as in case
  \eqref{item:Cauchy}, except that $\Tt$ goes to the initial state of $M_f$.
  
  The 2NUFT $\Aa'$ is then obtained as the composition of $\mathcal{D}$ with the
  2DFT $\Tt$.  Finally, a 2DFT $\Aa$ equivalent to the 2NUFT $\Aa'$ is
  constructed.
  Likewise, $\Bb$ is obtained using the composition of $\mathcal{D}$ with a 2DFT
  $\Tt'$ that runs $M_f$ on each factor $u_i$ from right to left.

  \medskip\noindent\eqref{item:twoplus} \textbf{2-chained Kleene-plus.}
  As in case \eqref{item:Kleene}, we construct the 1NUFT $\mathcal{D}$ which takes
  as input $w$ and outputs $u_1\#u_2\#\cdots\#u_n$ iff there is an
  unambiguous decomposition of $w$ as $u_1u_2 \cdots u_n$, with each $u_i\in K$.
  We then construct a 2DFT $\mathcal{D'}$ that takes as input words of the form
  $u_1\#u_2\#\cdots\#u_n$ with each $u_i\in\Sigma^*$ and produces
  $u_1u_2\#u_2u_3\#\cdots\#u_{n-1}u_n$.
  The 2NUFT $\Aa'$ is then obtained as the composition of $\mathcal{D}'$ with
  the 2DFT $\Tt$ constructed for case \eqref{item:Kleene}.
  Finally, a 2DFT $\Aa$ equivalent to the 2NUFT $\Aa'$ is constructed.
  The output produced by $\Aa$ is thus $\sem{M_f}(u_1u_2)\cdot
  \sem{M_f}(u_2u_3)\cdots \sem{M_f}(u_{n-1}u_n)$.
  We proceed similarly for $\Bb$.
\end{proof}

\subsection{Unambiguous forest factorization}

In Section~\ref{sec:2DFT-RTE}, we prove that, given a 2DFT $\Aa$, we can
obtain an \RTE $C$ such that $\sem{\Aa}=\sem{C}$.  We use the fact that any
$w\in\Sigma^*$ in the domain of $\Aa$ can be factorized unambiguously into a
good rational expression.  The unambiguous factorization of words in $\Sigma^*$
guides the construction of the combinator expression for $\sem{\Aa}(w)$ over
$\Gamma$ in an inductive way. 

For rational expressions over $\Sigma$ we will use the following syntax:
$$
F ::= \emptyset \mid \varepsilon \mid a \mid F\cup F \mid F\cdot F \mid F^+
$$
where $a\in\Sigma$. For reasons that will be clear below, we prefer to use the 
Kleene-plus instead of the Kleene-star, hence we also add $\varepsilon$ 
explicitely in the syntax. An expression is said to be $\varepsilon$-free if it 
does not use $\varepsilon$.

Let $(S,\cdot,\unitS)$ be a \emph{finite} monoid and $\varphi\colon\Sigma^*\to S$ be a
morphism.
We say that a rational expression $F$ is $\varphi$-\emph{good} (or simply
\emph{good} when $\varphi$ is clear from the context) when
\begin{enumerate}[nosep]
  \item the rational expression $F$ is unambiguous,
  
  \item for each subexpression $E$ of $F$ we have $\varphi(\Lang{E})=\{s_E\}$ is
  a singleton set,

  \item for each subexpression $E^+$ of $F$ we have $s_E\cdot s_E=s_E$ is an
  idempotent.
\end{enumerate}
Notice that $\emptyset$ cannot be used in a good expression
since it does not satisfy the second condition.

\begin{theorem}[Unambiguous Forest Factorization \cite{GastinKrishna-UFF}]\label{thm:U-forest}
  For each $s\in S$, there is an $\varepsilon$-free \emph{good} rational
  expression $F_s$ such that
  $\Lang{F_s}=\varphi^{-1}(s)\setminus\{\varepsilon\}\subseteq\Sigma^+$.
  Therefore, $G=\varepsilon\cup\bigcup_{s\in S}F_s$ is an unambiguous rational
  expression over $\Sigma$ such that $\Lang{G}=\Sigma^*$.
\end{theorem}

Theorem~\ref{thm:U-forest} can be seen as an unambiguous version of Imre 
Simon's forest factorization theorem \cite{Simon_1990}. Its proof, which
can be found in \cite{GastinKrishna-UFF}, follows the same lines of the recent 
proofs of Simon's theorem, see e.g.\ \cite{Colcombet_2010,ColcombetFactForest}.

In the rest of the section, we assume Theorem~\ref{thm:U-forest}, and use it in
obtaining an \RTE corresponding to $\Aa$.  For the purposes of this paper, we
work with the transition monoid of the two-way transducer.

\subsection{Transition monoid of 2NFAs}\label{sec:TrM-2NFA}

Consider a 2-way possibly non-deterministic automaton (2NFA) $\A$.  Let $\TrMon$
be the transition monoid of $\A$ which is obtained by
quotienting the free monoid $(\Sigma^*,\cdot,\varepsilon)$ by a congruence which
equate words behaving alike in the underlying automaton. 
In a one way automaton, the canonical morphism $\TrMorph\colon\Sigma^*\to\TrMon$
is such that $\TrMorph(w)$ consists of the set of pairs $(p,q)$ such that there
is a run
from state $p$ to state $q$ reading $w$.  In the case of two-way automaton, we
also consider the starting side (left/right) and ending side (left/right) of the
reading head while going from state $p$ to $q$.  Hence, an element of $\TrMon$
is a set $X$ of tuples $(p,d,q)$ with $p,q\in Q$ states of $\A$ and
$d\in\{\leftright,\leftleft,\rightright,\rightleft\}$ a direction amongst
``left-left'' ($\leftleft$), ``left-right'' ($\leftright$),
``right-left''($\rightleft$) and ``right-right''($\rightright$).

In the case of two-way automata, the canonical morphism
$\TrMorph\colon\Sigma^*\to\TrMon$ is such that $\TrMorph(w)$ is the set of
triples $(p,d,q)$ which are compatible with $w$.  For instance,
$(p,\leftright,q)\in\TrMorph(w)$ iff $\A$ has a run starting in state $p$ on the
left of $w$ and which exits $w$ on its right and in state $q$.  Likewise,
$(p,\rightright,q)\in\TrMorph(w)$ iff $\A$ has a run starting in state $p$ on
the right of $w$ and which exits $w$ on its right and in state $q$.  The
explanation is similar for other directions.  It is well-known that $\TrMon$ is
a monoid and that $\TrMorph$ is a morphism.

Consider the 2DFT $\Aa$ on the left of Figure~\ref{fig:intro} and its underlying
input 2DFA $\Bb$.  In the transition monoid of $\Bb$, we have
$\TrMorph(abb)=\{(q_1,\leftleft,q_5),(q_1,\rightright,q_2),
(q_2,\leftleft,q_4),(q_2,\rightleft,q_5),
(q_3,\leftleft,q_4),(q_3,\rightleft,q_5),
(q_4,\leftleft,q_4),(q_4,\rightright,q_1),
(q_5,\leftleft,q_5),(q_5,\rightright,q_6),
(q_6,\leftright,q_2),(q_6,\rightright,q_1)\}$.

Let $(p,d,q)\in\TrMorph(w)$.  If $w=a \in \Sigma$, then we know that reading $a$
in state $p$, $\A$ may move in direction $d$ and enter state $q$.  If
$w=w_1\cdot w_2$ for $w_1, w_2 \in \Sigma^+$, then we can possibly decompose
$(p,d,q)$ into several ``steps'' depending on the behaviour of $\A$ on $w$
starting in state $p$.  As an example, see Figure~\ref{steps}, where we
decompose $(p, \leftright, q)\in\TrMorph(w)$.  We show only those elements of
$\TrMorph(w_1)$ and $\TrMorph(w_2)$ which help in the decomposition; the
pictorial depiction is visually intuitive.

\begin{figure}[h]
  \centerline{\includegraphics[scale=0.28]{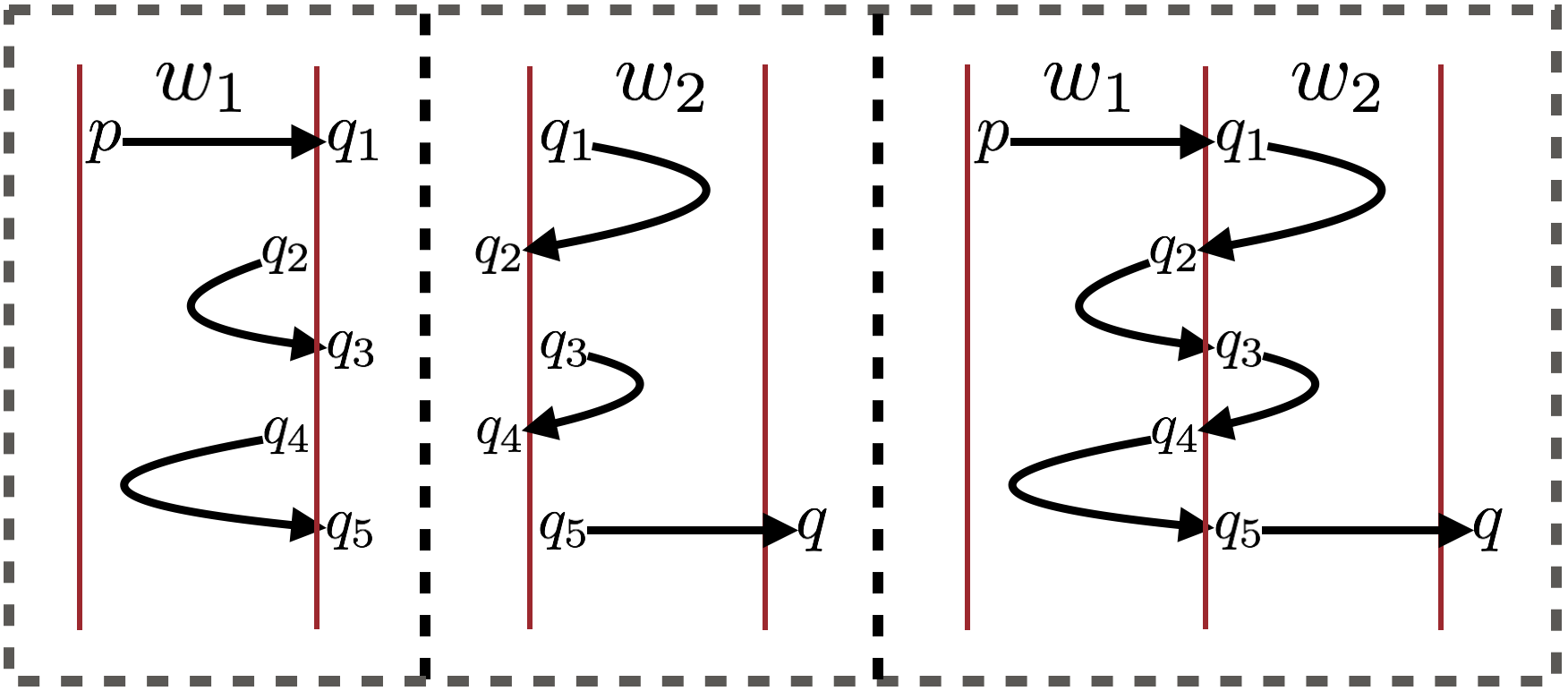}}
  \caption{The first and second pictures are illustrations of subsets of
  $\TrMorph(w_1)$ and $\TrMorph(w_2)$ respectively.  $(p, \leftright, q_1),
  (q_2, \rightright, q_3), (q_4, \rightright, q_5) \in \TrMorph(w_1)$ while
  $(q_1, \leftleft, q_2), (q_3, \leftleft, q_4), (q_5, \leftright, q) \in
  \TrMorph(w_2)$.  The third picture shows that $(p, \leftright, q) \in
  \TrMorph(w_1\cdot w_2)$: $(p, \leftright, q)$ consists of ``steps''
  $(p,\leftright,q_1), (q_1,\leftleft,q_2), (q_2,\rightright,q_3),
  (q_3,\leftleft,q_4), (q_4,\rightright,q_5), (q_5,\leftright,q)$ alternately
  from $\TrMorph(w_1)$ and $\TrMorph(w_2)$.}
  \label{steps}
\end{figure}

\begin{example}
  Let $\Sigma=\{a,b\}$ and let $\Aa$ be the following 1DFT:
  \raisebox{-4mm}{\includegraphics[page=2,scale=1]{gpicture-pics.pdf}}\,. \\
  Let $\TrMon$ be the transition monoid of $\Aa$ and let
  $\TrMorph\colon\Sigma^*\to\TrMon$ be the canonical morphism.  The expression
  $F=a^+(ba)^+$ is not $\TrMorph$-good: one of the reasons why $F$ is not
  $\TrMorph$-good is that the subexpression $a^+$ is such that $\TrMorph(a)$ is
  not an idempotent; the same is true for the subexpression $(ba)^+$.  The
  expression $F'=aba \cup aaba \cup a(aa)^+ba \cup a(baba)^+ \cup
  a(aa)^+(baba)^+$ is not $\TrMorph$-good, even though each of the expressions
  $aba, aaba, a(aa)^+ba, a(baba)^+$ and $a(aa)^+(baba)^+$ are $\TrMorph$-good.
  $F'$ is not $\TrMorph$-good since $\TrMorph(\Lang{F'})$ is not a singleton.
  The expression $F''=aba \cup (aa)^{+}\cup a(aa)^+ba $ is $\TrMorph$-good.
  \label{eg:tm}
\end{example}

\subsection{2DFT to RTE}\label{sec:2DFT-RTE}

Consider a deterministic and complete 2-way transducer $\A$.  Let $\TrMon$ be
the transition monoid of the underlying input automaton.  We can apply the
unambiguous factorization theorem to the morphism
$\TrMorph\colon\Sigma^*\to\TrMon$ in order to obtain, for each $s \in \TrMon$,
an $\varepsilon$-free good rational expression $F_s$ for
$\TrMorph^{-1}(s)\setminus\{\varepsilon\}$.  We use the unambiguous expression
$G=\varepsilon\cup\bigcup_{s\in\TrMon}F_s$ as a \emph{guide} when constructing
\RTEs corresponding to the 2DFT $\A$.

\begin{lemma}\label{lem:C_E}
  Let $F$ be an $\varepsilon$-free $\TrMorph$-good rational expression and let
  $\TrMorph(F)=s_F$ be the corresponding element of the transition monoid
  $\TrMon$ of $\A$.  We can construct a map $C_F\colon s_F\to\RTE$ such that
  for each step $x=(p,d,q)\in s_F$ the following invariants hold:
  \begin{enumerate}[nosep,label=($\mathsf{I}_{\arabic*}$),ref=$\mathsf{I}_{\arabic*}$]
    \item\label{Inv1} $\dom{C_F(x)}=\Lang{F}$,
    
    \item\label{Inv2} for each $u\in\Lang{F}$, $\sem{C_F(x)}(u)$ is the output
    produced by $\A$ when running step $x$ on $u$ (i.e., running $\A$ on $u$
    from $p$ to $q$ following direction $d$).
  \end{enumerate}
\end{lemma}

\begin{proof}
  The proof is by structural induction on the rational expression.  For each
  subexpression $E$ of $F$ we let $\TrMorph(E)=s_E$ be the corresponding element
  of the transition monoid $\TrMon$ of $\A$.  We start with atomic regular
  expressions.  Since $F$ is $\varepsilon$-free and $\emptyset$-free, we do not
  need to consider $E=\varepsilon$ or $E=\emptyset$.

\begin{description}
  \item[atomic] Assume that $E=a\in\Sigma$ is an atomic subexpression.  Since
  the 2DFT $\A$ is deterministic and complete, for each state $p\in Q$ we have
  \begin{itemize}[nosep]
    \item either $\delta(p,a)=(q,\gamma,1)$ and we let
    $C_a((p,\leftright,q))=C_a((p,\rightright,q))=\Ifthenelse{a}{\gamma}\bot$,
  
    \item or $\delta(p,a)=(q,\gamma,-1)$ and we let
    $C_a((p,\leftleft,q))=C_a((p,\rightleft,q))=\Ifthenelse{a}{\gamma}\bot$.
  \end{itemize}
  Clearly, invariants \eqref{Inv1} and \eqref{Inv2} hold for all 
  $x\in\TrMorph(a)=s_E$.
  
  \item[Union] Assume that $E=E_1\cup E_2$.  Since the expression is good, we
  deduce that $s_E=s_{E_1}=s_{E_2}$.  For each $x\in s_E$ we define
  $C_E(x)=\Ifthenelse{E_1}{C_{E_1}(x)}{C_{E_2}(x)}$.  
  Since $E$ is unambiguous
  we have $\Lang{E_1}\cap\Lang{E_2}=\emptyset$.  Using \eqref{Inv1} for $E_1$
  and $E_2$, we deduce that
  \begin{align*}
    \dom{C_E(x)} & = (\Lang{E_1}\cap\dom{C_{E_1}(x)}) \cup 
    (\dom{C_{E_2}(x)}\setminus\Lang{E_1}) 
     = \Lang{E_1}\cup\Lang{E_2}=\Lang{E} \,.
  \end{align*}
  Therefore, \eqref{Inv1} holds for $E$.  Now, for each $u\in\Lang{E}$, either
  $u\in\Lang{E_1}$ and $\sem{C_E(x)}(u)=\sem{C_{E_1}(x)}(u)$ or $u\in\Lang{E_2}$
  and $\sem{C_E(x)}(u)=\sem{C_{E_2}(x)}(u)$.  In both cases, applying
  \eqref{Inv2} for $E_1$ or $E_2$, we deduce that $\sem{C_E(x)}(u)$ is the
  output produced by $\A$ when running step $x$ on $u$.

  \item[concatenation] Assume that $E=E_1\cdot E_2$ is a concatenation.  
  Since the expression is good, we deduce that
  $s_E=s_{E_1}\cdot s_{E_2}$. Let $x\in s_E$.
  \begin{itemize}[nosep]
    \item If $x=(p,\leftright,q)$ then, by definition of the product in the 
    transition monoid $\TrMon$, there is a unique sequence of steps 
    $x_1=(p,\leftright,q_1)$,
    $x_2=(q_1,\leftleft,q_2)$,
    $x_3=(q_2,\rightright,q_3)$,
    $x_4=(q_3,\leftleft,q_4)$, \ldots,
    $x_i=(q_{i-1},\rightright,q_i)$,
    $x_{i+1}=(q_i,\leftright,q)$ 
    with $i\geq1$, 
    $x_1,x_3,\ldots,x_i\in s_{E_1}$ and
    $x_2,x_4,\ldots,x_{i+1}\in s_{E_2}$ (see Figure \ref{steps}). 
    We define
    $$
    C_E(x)=(C_{E_1}(x_1)\lrdot C_{E_2}(x_2))\odot
    (C_{E_1}(x_3)\lrdot C_{E_2}(x_4))\odot
    \cdots\odot(C_{E_1}(x_i)\lrdot C_{E_2}(x_{i+1})) \,.
    $$
    Notice that when $i=1$ we simply have $C_E(x)=C_{E_1}(x_1)\lrdot 
    C_{E_2}(x_2)$ with $x_2=(q_1,\leftright,q)$.
    
    The concatenation $\Lang{E}=\Lang{E_1}\cdot\Lang{E_2}$ is unambiguous.
    Therefore, for all $y\in s_{E_1}$ and $z\in s_{E_2}$, using \eqref{Inv1} for
    $E_1$ and $E_2$, we obtain $\dom{C_{E_1}(y)\lrdot C_{E_2}(z)}=\Lang{E}$.  We
    deduce that $\dom{C_E(x)}=\Lang{E}$ and \eqref{Inv1} holds for $E$.

    Now, let $u\in\Lang{E}$ and let $u=u_1u_2$ be its unique factorization with
    $u_1\in\Lang{E_1}$ and $u_2\in\Lang{E_2}$.  The step $x=(p,\leftright,q)$
    performed by $\A$ on $u$ is actually the concatenation of steps $x_1$ on
    $u_1$, followed by $x_2$ on $u_2$, followed by $x_3$ on $u_1$, followed by
    $x_4$ on $u_2$, \ldots, until $x_{i+1}$ on $u_2$.  Using \eqref{Inv2} for
    $E_1$ and $E_2$, we deduce that the output produced by $\A$ while running
    step $x$ on $u$ is
    $$
    \sem{C_{E_1}(x_1)}(u_1) \cdot \sem{C_{E_2}(x_2)}(u_2) 
    \cdots
    \sem{C_{E_1}(x_i)}(u_1) \cdot \sem{C_{E_2}(x_{i+1})}(u_2)
    = \sem{C_E(x)}(u) \,.
    $$
    \item If $x=(p,\leftleft,q)$ then, following the definition of the product
    in the transition monoid $\TrMon$, we distinguish two cases.  
    
    Either $x\in s_{E_1}$ and we let
    $C_E(x)=C_{E_1}(x)\lrdot(\Ifthenelse{E_2}{\varepsilon}{\bot})$.  Since
    $\dom{\Ifthenelse{E_2}{\varepsilon}{\bot}}=\Lang{E_2}$, we deduce as above that
    $\dom{C_E(x)}=\Lang{E}$.  Moreover, let $u\in\Lang{E}$ and $u=u_1u_2$ be its
    unique factorization with $u_1\in\Lang{E_1}$ and $u_2\in\Lang{E_2}$.  The
    step $x=(p,\leftleft,q)$ performed by $\A$ on $u$ reduces to the step $x$ on
    $u_1$.  Using \eqref{Inv2} for $E_1$, we deduce that the output produced by
    $\A$ while making step $x$ on $u$ is
    $\sem{C_{E_1}(x)}(u_1)=\sem{C_E(x)}(u)$.

    \begin{figure}[h]
      \begin{center}
        \includegraphics[scale=0.28]{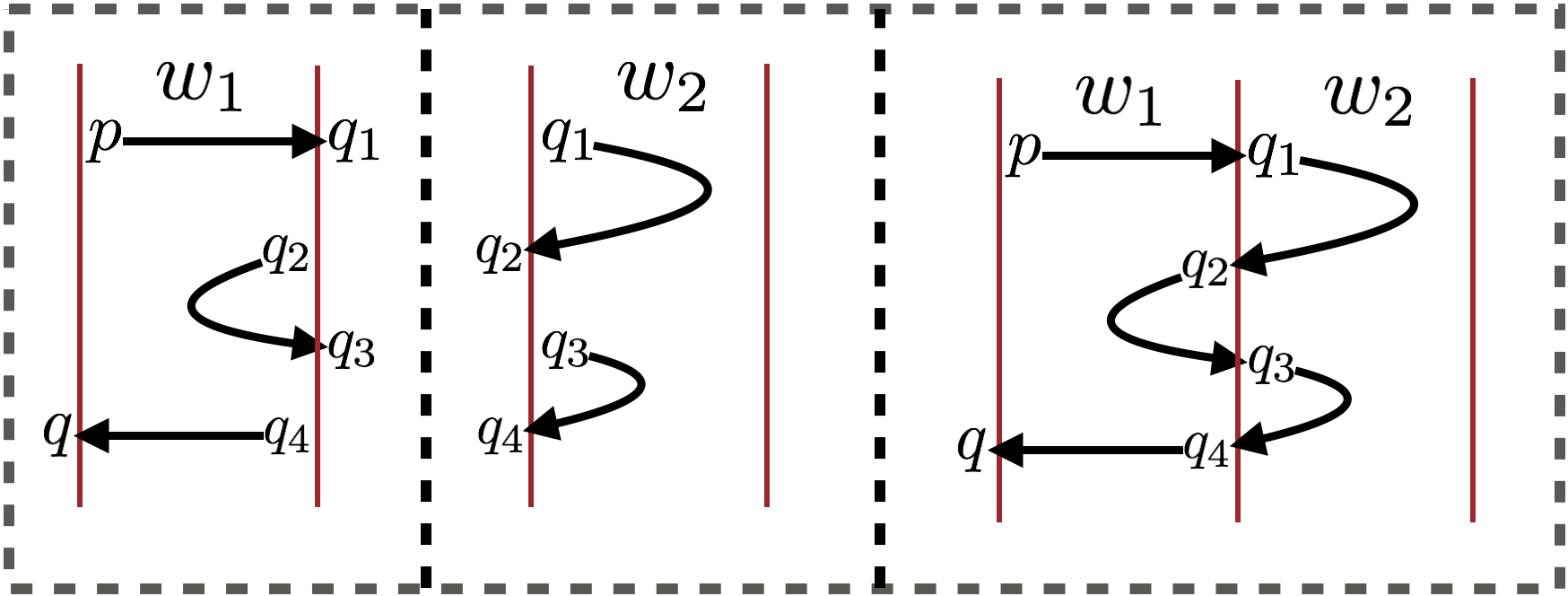}
      \end{center}
      \caption{Let $w=w_1\cdot w_2\in\Lang{E}$ with $w_1\in\Lang{E_1}$,
      $w_2\in\Lang{E_2}$.  We have $(p,\leftright,q_1), (q_2, \rightright, q_3),
      (q_4, \rightleft,q) \in \TrMorph(w_1)$ and $(q_1, \leftleft, q_2),
      (q_3,\leftleft, q_4) \in \TrMorph(w_2)$.  Then $(p,\leftleft,q)$ is
      composed of ``steps'' $(p,\leftright,q_1), (q_1,\leftleft,q_2),
      (q_2,\rightright,q_3), (q_3,\leftleft,q_4), (q_4,\rightleft,q)$
      alternately from $\TrMorph(w_1)$ and $\TrMorph(w_2)$.}
      \label{steps-1}
    \end{figure}

    Or there is a unique sequence of steps (see Figure~\ref{steps-1})
    $x_1=(p,\leftright,q_1)$,
    $x_2=(q_1,\leftleft,q_2)$,
    $x_3=(q_2,\rightright,q_3)$,
    $x_4=(q_3,\leftleft,q_4)$, \ldots,
    $x_i=(q_{i-1},\rightleft,q)$ with $i\geq3$, 
    $x_1,x_3,\ldots,x_i\in s_{E_1}$ and
    $x_2,x_4,\ldots,x_{i-1}\in s_{E_2}$. We define
    $$
    C_E(x)=(C_{E_1}(x_1)\lrdot C_{E_2}(x_2))\odot
    (C_{E_1}(x_3)\lrdot C_{E_2}(x_4))\odot
    \cdots\odot(C_{E_1}(x_i)\lrdot(\Ifthenelse{E_2}{\varepsilon}{\bot}) ) \,.
    $$
    As for the first item, we can prove that invariants \eqref{Inv1} and 
    \eqref{Inv2} are satisfied for $E$.
  
    \item The cases $x=(p,\rightleft,q)$ or $x=(p,\rightright,q)$ are handled
    symmetrically.
    For instance, when $x=(p,\rightleft,q)$, the unique
    sequence of steps is
    $x_1=(p,\rightleft,q_1)$,
    $x_2=(q_1,\rightright,q_2)$,
    $x_3=(q_2,\leftleft,q_3)$,
    $x_4=(q_3,\rightright,q_4)$, \ldots,
    $x_i=(q_{i-1},\leftleft,q_i)$,
    $x_{i+1}=(q_i,\rightleft,q)$ 
    with $i\geq1$, 
    $x_1,x_3,\ldots,x_i\in s_{E_2}$ and
    $x_2,x_4,\ldots,x_{i+1}\in s_{E_1}$ (see Figure \ref{steps-2}).  We define
    \begin{align*}
      C_E(x)=((\Ifthenelse{E_1}{\varepsilon}\bot)\lrdot C_{E_2}(x_1))
      &\odot (C_{E_1}(x_2)\lrdot C_{E_2}(x_3))\odot\cdots\odot{}
      \\
      (C_{E_1}(x_{i-1})\lrdot C_{E_2}(x_{i}))
      &\odot (C_{E_1}(x_{i+1})\lrdot(\Ifthenelse{E_2}{\varepsilon}{\bot})) \,.
    \end{align*}
  \end{itemize}
  
  \begin{figure}[h]
    \begin{center}
      \includegraphics[scale=0.28]{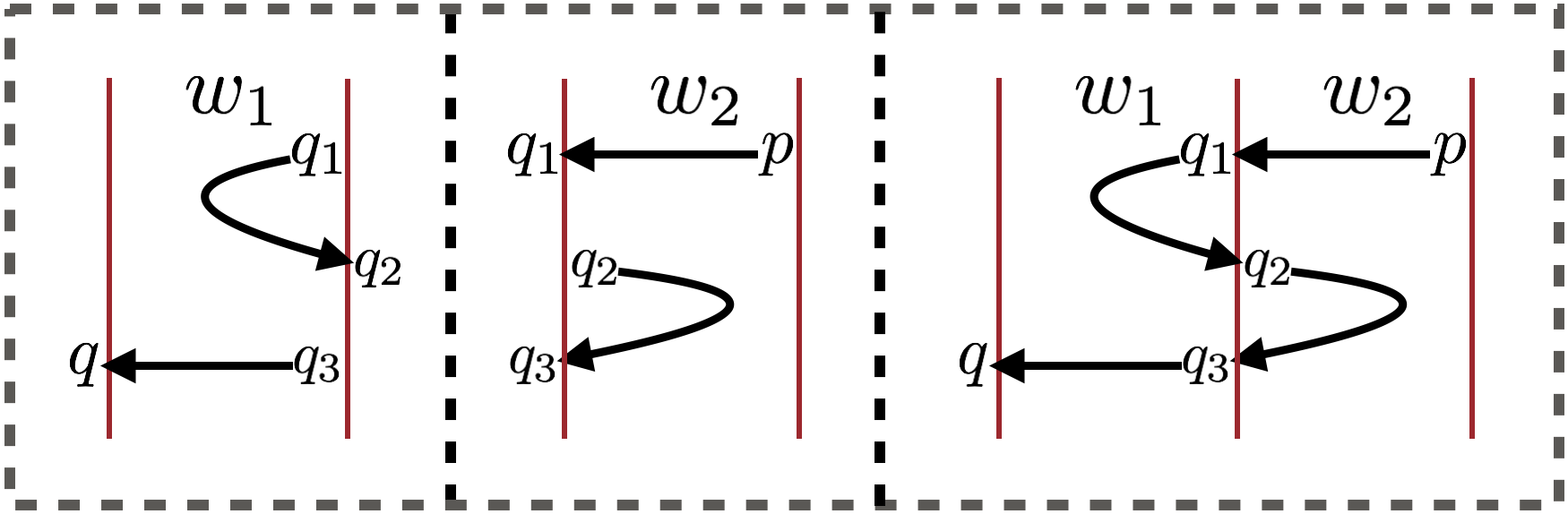}
    \end{center}
    \caption{Let $w=w_1\cdot w_2\in\Lang{E}$ with $w_1\in\Lang{E_1}$,
    $w_2\in\Lang{E_2}$.  We have $(p,\rightleft,q_1), (q_2, \leftleft, q_3) \in
    \TrMorph(w_2)$ and $(q_1, \rightright, q_2), (q_3,\rightleft, q) \in
    \TrMorph(w_1)$.  Then $(p,\rightleft,q)\in\TrMorph(w)$ is composed of
    ``steps'' $(p,\rightleft,q_1), (q_1, \rightright, q_2), (q_2, \leftleft,
    q_3), (q_3,\rightleft,q)$ alternately from $\TrMorph(w_2)$ and
    $\TrMorph(w_1)$.}
    \label{steps-2}
  \end{figure}
  
  \item[Kleene-plus] Assume that $E=F^+$.  Since the expression is good, we
  deduce that $s_E=s_F=s$ is an idempotent of the transition monoid $\TrMon$.
  Let $x\in s$.
  \begin{itemize}[nosep]
    \item If $x=(p,\leftleft,q)$.  Since $F^+$ is unambiguous, a word
    $u\in\Lang{F^+}$ admits a unique factorization $u=u_1u_2\cdots u_n$ with
    $n\geq1$ and $u_i\in\Lang{F}$.  Now, $\TrMorph(u_1)=s_E$ and since
    $x=(p,\leftleft,q)\in s_E$ the unique run $\rho$ of $\A$ starting in state
    $p$ on the left of $u_1$ exits on the left in state $q$.  Therefore,
    the unique run of $\A$ starting in state $p$ on the left of $u=u_1u_2\cdots
    u_n$ only visits $u_1$ and is actually $\rho$ itself. Therefore, we set 
    $C_E(x)=C_F(x)\lrdot(\Ifthenelse{F^*}{\varepsilon}\bot)$ and we can easily check 
    that (\ref{Inv1}--\ref{Inv2}) are satisfied.
    
    \item Similarly for $x=(p,\rightright,q)$ 
    we set $C_E(x)=(\Ifthenelse{F^*}{\varepsilon}\bot)\lrdot C_F(x)$.
    
    \item If $x=(p,\leftright,q)$.  Recall that $s$ is an idempotent, hence
    $x\in s^2$.  We distinguish two cases.
    
    Either $y=(q,\leftright,q)\in s$ and we set
    $C_E(x)=\Ifthenelse{F}{C_F(x)}{\big(C_F(x)\lrdot\lrplus{(C_F(y))}\big)}$.
    
    \begin{figure}[t]
      \begin{center}
        \includegraphics[scale=0.29]{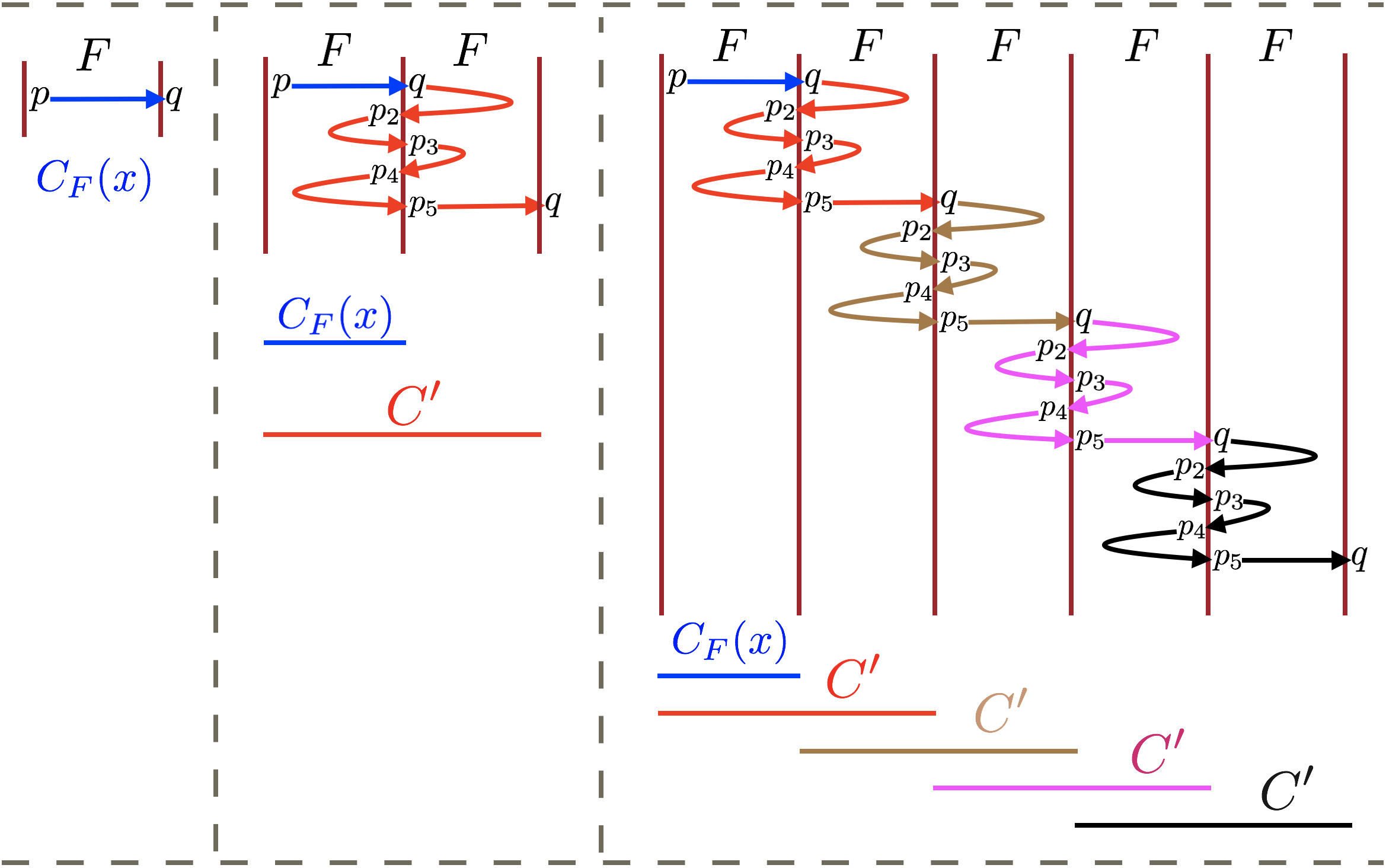}
      \end{center}
      \caption{In the Kleene-plus $E=F^+$, a step $x=(p,\leftright,q)\in s_E$ on
      some $u=u_1u_2\cdots u_n$ with $u_\ell\in\Lang{F}$ is obtained by composing
      the following steps in $s_F$: $x_1=x$, $x_2=(q,\leftleft,p_2)$,
      $x_3=(p_2,\rightright,p_3)$, $x_4=(p_3,\leftleft,p_4)$,
      $x_5=(p_4,\rightright,p_5)$, $x_6=(p_5,\leftright,q)$.}
      \label{kleene}
    \end{figure}

    Or there exists a unique sequence of steps in $s$:
    $x_1=x$,
    $x_2=(q,\leftleft,p_2)$,
    $x_3=(p_2,\rightright,p_3)$,
    $x_4=(p_3,\leftleft,p_4)$, \ldots,
    $x_i=(p_{i-1},\rightright,p_i)$,
    $x_{i+1}=(p_i,\leftright,q)$ 
    with $i\geq3$ (see Figure~\ref{kleene}). We define
    \begin{align*}
      C_E(x) &= 
      \big(C_F(x)\lrdot(\Ifthenelse{F^*}{\varepsilon}\bot)\big)
      \odot\twoplus{F}{C'}
      \\
      C' &= 
      \big((\Ifthenelse{F}{\varepsilon}\bot)\lrdot C_F(x_2)\big)
      \odot(C_F(x_3)\lrdot C_F(x_4))\odot
      \cdots\odot(C_F(x_i)\lrdot C_F(x_{i+1})) 
    \end{align*}
    Since the expression is good, the Kleene-plus $E=F^+$ is unambiguous.  We
    have $\dom{C_F(x_j)}=\Lang{F}$ for $1\leq j\leq i+1$ by \eqref{Inv1}.  Also
    $\dom{\Ifthenelse{F^*}{\varepsilon}\bot}=\Lang{F^*}$.  Since $F^+$ is unambiguous,
    the concatenation $F\cdot F^*$ is also unambiguous and we get
    $\dom{C_F(x)\lrdot(\Ifthenelse{F^*}{\varepsilon}\bot)}=\Lang{F}\cdot\Lang{F^*}=\Lang{E}$.
    Also, the product $F\cdot F$ is unambiguous and we deduce that
    $\dom{C_F(x_j)\lrdot C_F(x_{j+1})}=\Lang{F}^2$ for $1\leq j\leq i$ and
    $\dom{(\Ifthenelse{F}{\varepsilon}\bot)\lrdot C_F(x_2)}=\Lang{F}^2$.  Therefore,
    $\dom{C'}=\Lang{F}^2$ and using once again that $F^+$ is unambiguous, we
    deduce that $\dom{\twoplus{F}{C'}}=\Lang{F^+}=\Lang{E}$.  We deduce that
    $\dom{C_E(x)}=\Lang{E}$ and \eqref{Inv1} holds for $E$.
    
    Let now $u\in\Lang{F^+}=\dom{C_E(x)}$.  We have to show that the output
    $\gamma\in\D$ produced by $\A$ when running step $x$ on $u$ is
    $\sem{C_E(x)}(u)$.  There is a unique factorization $u=u_1u_2\cdots u_n$ with
    $n\geq1$ and $u_\ell\in\Lang{F}$ for $1\leq\ell\leq n$.  
    
    Assume first that $n=1$ (see Figure~\ref{kleene} left).  By definition, we
    have $\sem{\twoplus{F}{C'}}(u)=\varepsilon$ and
    $\sem{C_F(x)\lrdot(\Ifthenelse{F^*}{\varepsilon}\bot)}(u)=\sem{C_F(x)}(u)$ which,
    by induction, is the output $\gamma$ produced by $\A$ running step $x$ on
    $u$.  Therefore, $\sem{C_E(x)}(u)=\gamma\cdot\varepsilon=\gamma$.
    
    Assume now that $n\geq2$ (see Figure~\ref{kleene} middle for $n=2$ and right
    for $n=5$).  For $1\leq\ell\leq n$ and $1\leq j\leq i+1$, we denote
    $\gamma^\ell_j=\sem{C_F(x_j)}(u_\ell)$ the output produced by $\A$ when
    running step $x_j$ on $u_\ell$.  We can check (see Figure~\ref{kleene}) that
    the output $\gamma$ produced by $\A$ when running $x$ on $u=u_1u_2\cdots
    u_n$ is
    $$
    \gamma=\gamma_1^1
    (\gamma_2^{2}\gamma_3^{1}\gamma_4^{2}\cdots\gamma_i^{1}\gamma_{i+1}^{2})
    (\gamma_2^{3}\gamma_3^{2}\gamma_4^{3}\cdots\gamma_i^{2}\gamma_{i+1}^{3})
    \cdots
    (\gamma_2^{n}\gamma_3^{n-1}\gamma_4^{n}\cdots\gamma_i^{n-1}\gamma_{i+1}^{n}) \,.
    $$
    
    We have
    $\sem{C'}(u_\ell u_{\ell+1})=\gamma_2^{\ell+1}\gamma_3^\ell\gamma_4^{\ell+1}
    \cdots\gamma_i^{\ell}\gamma_{i+1}^{\ell+1}$
    for $1\leq\ell<n$. Therefore, we obtain
    $\gamma=\gamma_1^1\sem{C'}(u_1u_2)\sem{C'}(u_2u_3)\cdots\sem{C'}(u_{n-1}u_n)$.
    Since $\sem{C_F(x)\lrdot(\Ifthenelse{F^*}{\varepsilon}\bot)}(u)=\gamma_1^1$ we 
    deduce that $\gamma=\sem{C_E(x)}(u)$.
    
    \item The case of $x=(p,\rightleft,q)$ can be handled similarly. 
    \qedhere
  \end{itemize}
\end{description}
\end{proof}

Lemma~\ref{lem:C_E} is the main ingredient in the construction of an \RTE 
equivalent to a 2DFT.

\begin{proof}[Proof of Theorem~\ref{thm:2DFT=RTE}\eqref{item:2DFTtoRTE}]
  First, we let $C_\varepsilon=\sem{\A}(\varepsilon)\in\Gamma^*\cup\{\bot\}$.
  Then, we will define for each $s\in\TrMon$, an \RTE $C_s$ such that
  $\dom{C_s}=\dom{\A}\cap(\TrMorph^{-1}(s)\setminus\{\varepsilon\})$ and
  $\sem{C_s}(u)=\sem{\A}(u)$ for all $u\in\dom{C_s}$. Assuming an arbitrary 
  enumeration $s_1,s_2,\ldots,s_m$ of $\TrMon$, we define the final \RTE as
  $$
  C_\A=
  \Ifthenelse{\varepsilon}{C_\varepsilon}{
  (\Ifthenelse{\TrMorph^{-1}(s_1)}{C_{s_1}}{
  (\Ifthenelse{\TrMorph^{-1}(s_2)}{C_{s_2}}{\cdots
  (\Ifthenelse{\TrMorph^{-1}(s_{m-1})}{C_{s_{m-1}}}{C_{s_m}})
  })})} \,.
  $$
  It remains to define the \RTE $C_s$ for $s\in\TrMon$.
  We first define \RTEs for steps in the 2DFT $\A$ on some input $\leftend u$
  with $u\in\TrMorph^{-1}(s)\setminus\{\varepsilon\}$.  Such a step must exit on
  the right since there are no transitions of $\A$ going left when reading
  $\leftend$.  So either the step $(q_0,\leftright,q)$ starts on the left in the
  initial state $q_0$ and exits on the right in some state $q$.  Or the step
  $(p,\rightright,q)$ starts on the right in state $p$ and exits on the right in
  state $q$. See Figure~\ref{fig:leftend}.
  
  Let $s_\leftend$ be the set of steps $(p,\leftright,q),(p,\rightright,q)$ such
  that there is a transition $\delta(p,\leftend)=(q,\gamma_p,+1)$ in $\A$.
  From the initial state $q_0$ of $\A$, there is a unique sequence of steps
  $x_1=(q_0,\leftright,q_1)$,
  $x_2=(q_1,\leftleft,q_2)$,
  $x_3=(q_2,\rightright,q_3)$,
  $x_4=(q_3,\leftleft,q_4)$, \ldots,
  $x_i=(q_{i-1},\rightright,q_i)$,
  $x_{i+1}=(q_i,\leftright,q)$ 
  with $i\geq1$, 
  $x_1,x_3,\ldots,x_i\in s_\leftend$ and
  $x_2,x_4,\ldots,x_{i+1}\in s$ (see Figure~\ref{fig:leftend} left). 
  We define
  $$
  C_{\leftend F_s}((q_0,\leftright,q))=\gamma_{q_0}\odot C_{F_s}(x_2)\odot
  \gamma_{q_2}\odot C_{F_s}(x_4)\odot
  \cdots\odot\gamma_{q_{i-1}}\odot C_{F_s}(x_{i+1}) \,.
  $$
  Notice that when $i=1$ we simply have $C_{\leftend
  F_s}((q_0,\leftright,q))=\gamma_{q_0}\odot C_{F_s}((q_1,\leftright,q))$.
  Since $\dom{C_{F_s}(x_i)}=\Lang{F_s}=\TrMorph^{-1}(s)\setminus\{\varepsilon\}$
  for $i=2,4,\ldots,i+1$, we deduce that $\dom{C_{\leftend
  F_s}((q_0,\leftright,q))} =\TrMorph^{-1}(s)\setminus\{\varepsilon\}$.
  Moreover, for each $u\in\TrMorph^{-1}(s)\setminus\{\varepsilon\}$, the output
  produced by $\A$ performing step $(q_0,\leftright,q)$ on $\leftend u$ is
  $\sem{C_{\leftend F_s}((q_0,\leftright,q))}(u)$.
   
  \begin{figure}[h]
    \begin{center}
      \includegraphics[scale=0.18]{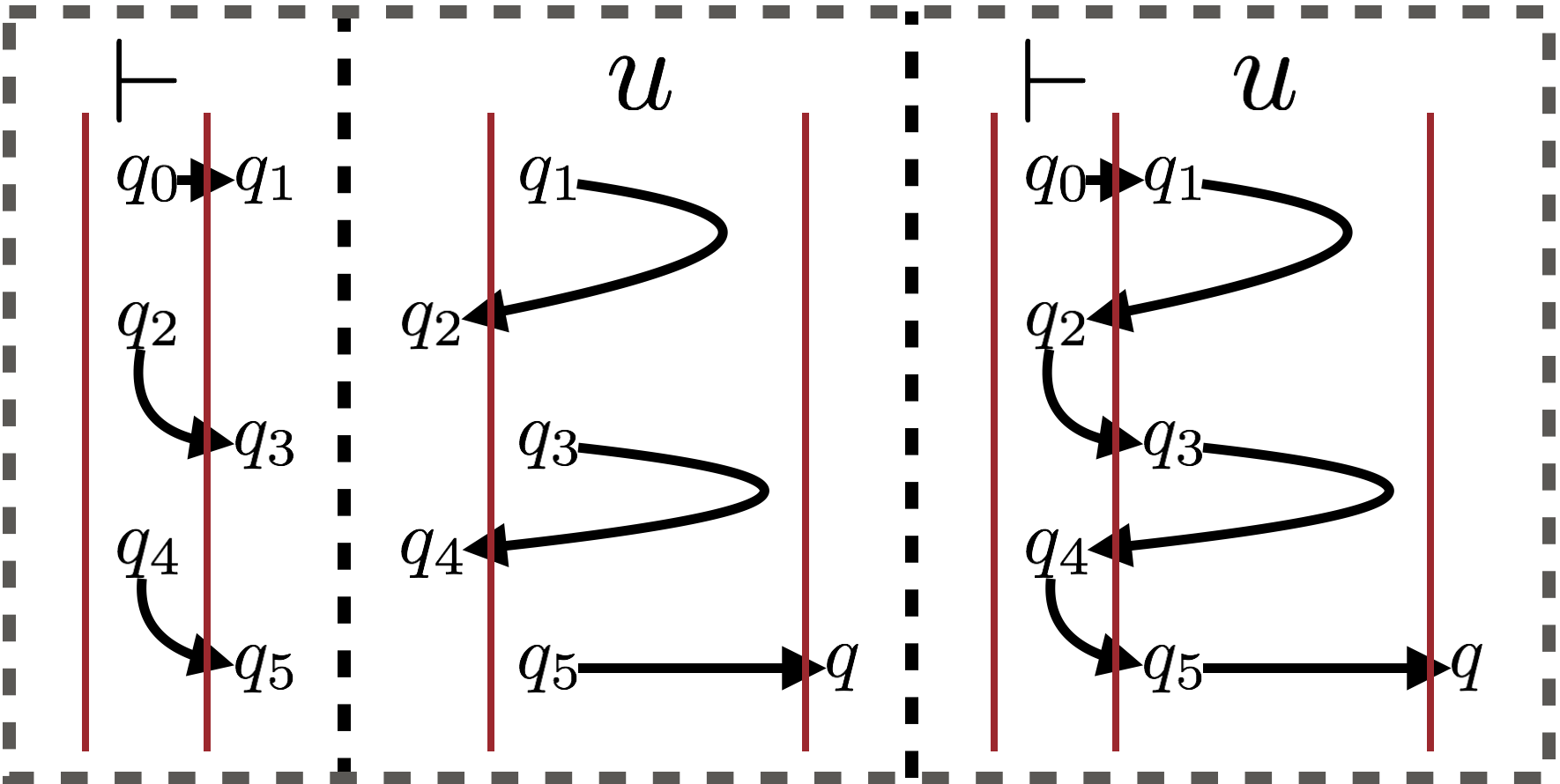}
      \hfil
      \includegraphics[scale=0.18]{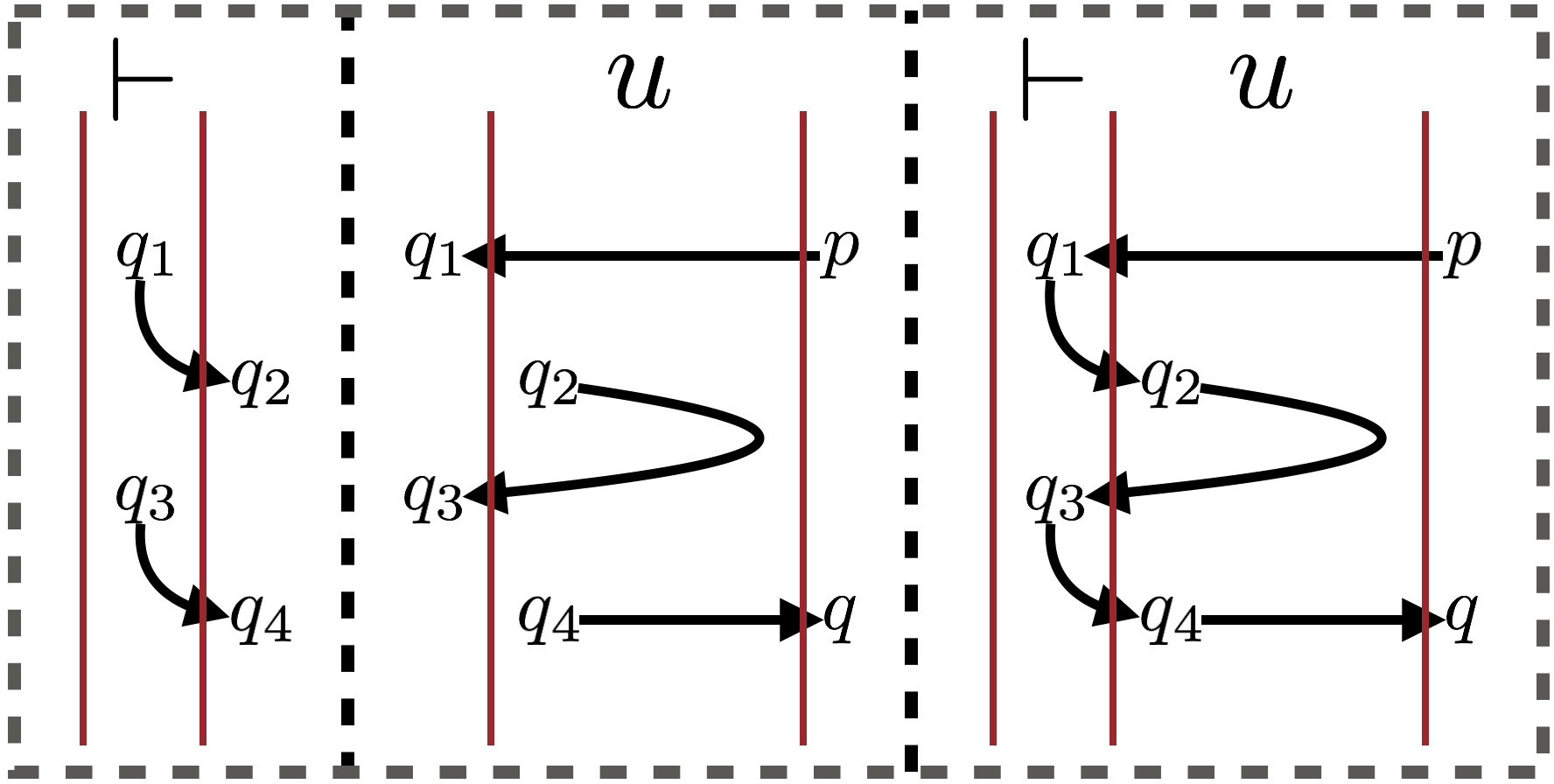}
    \end{center}
    \caption{}
    \label{fig:leftend}
  \end{figure}
  
  Let $p$ be a state of $\A$. Either there is a step $(p,\rightright,q)\in s$ 
  and we let $C_{\leftend F_s}((p,\rightright,q))=C_{F_s}((p,\rightright,q))$.
  Or, there is a unique sequence of steps
  $x_1=(p,\rightleft,q_1)$,
  $x_2=(q_1,\rightright,q_2)$,
  $x_3=(q_2,\leftleft,q_3)$,
  $x_4=(q_3,\rightright,q_4)$, \ldots,
  $x_i=(q_{i-1},\leftright,q)$
  with $i\geq3$, 
  $x_1,x_3,\ldots,x_i\in s$ and
  $x_2,x_4,\ldots,x_{i-1}\in s_\leftend$ (see Figure~\ref{fig:leftend} right).  
  We define
  $$
  C_{\leftend F_s}((p,\rightright,q))=C_{F_s}(x_1)\odot\gamma_{q_1}\odot 
  C_{F_s}(x_3)\odot\gamma_{q_3}\odot 
  \cdots\odot\gamma_{q_{i-2}}\odot C_{F_s}(x_{i}) \,.
  $$
  As above, we have $\dom{C_{\leftend F_s}((p,\rightright,q))}
  =\TrMorph^{-1}(s)\setminus\{\varepsilon\}$.  Moreover, for each
  $u\in\TrMorph^{-1}(s)\setminus\{\varepsilon\}$, the output produced by $\A$
  performing step $(p,\rightright,q)$ on $\leftend u$ is $\sem{C_{\leftend
  F_s}((p,\rightright,q))}(u)$.

  \begin{figure}[h]
    \begin{center}
      \includegraphics[scale=0.17]{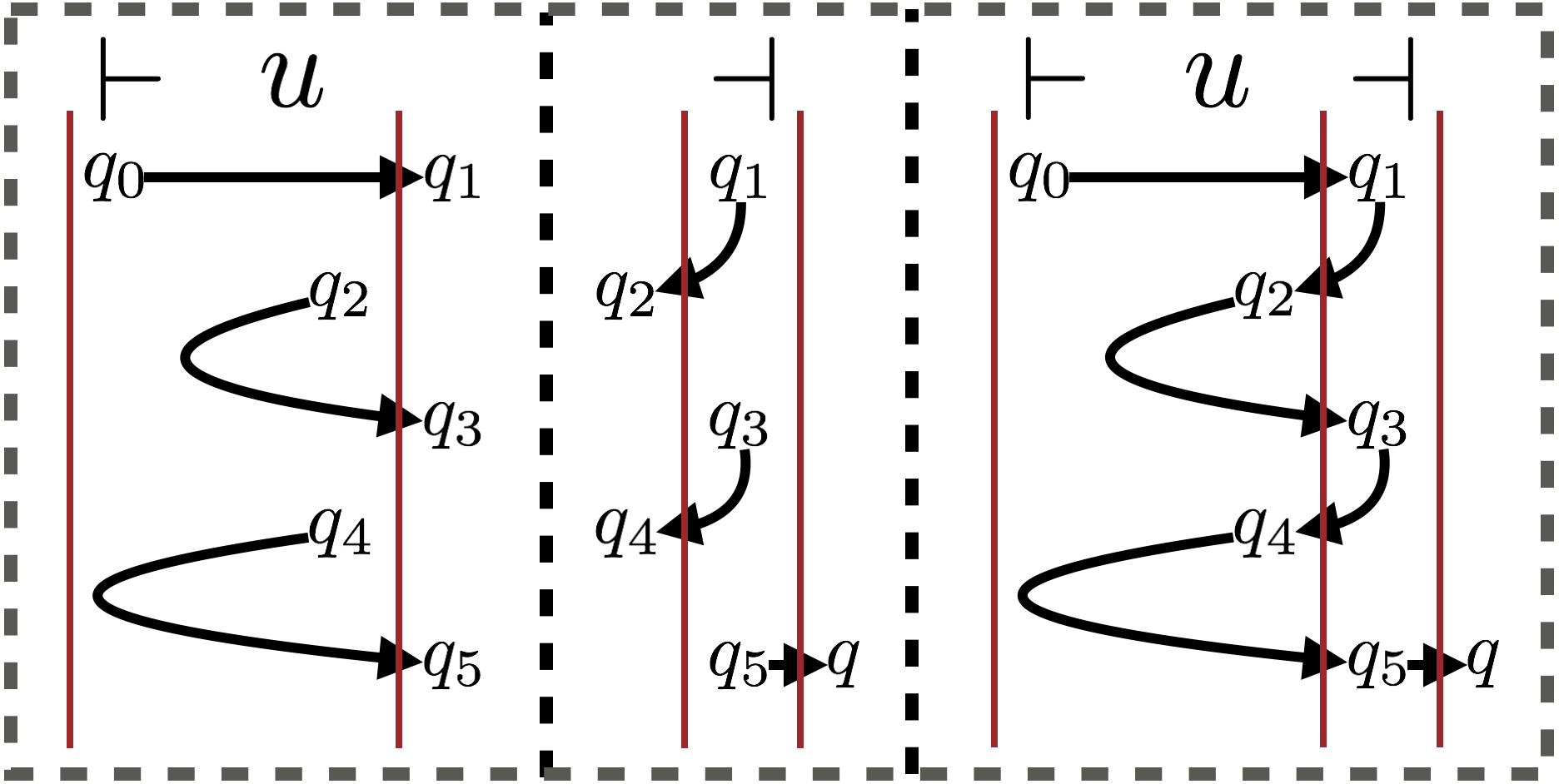}
    \end{center}
    \caption{}
    \label{steps-fin-3}
  \end{figure}
  
  Similarly, let $s_\rightend$ be the set of steps $(p,\leftleft,q)$ such that
  there is a transition $\delta(p,\rightend)=(q,\gamma_p,-1)$ in $\A$ or
  $(p,\leftright,q)$ such that there is a transition
  $\delta(p,\rightend)=(q,\gamma_p,+1)$ in $\A$. From the initial state $q_0$ 
  of $\A$, there is a unique sequence of steps 
  $x_1=(q_0,\leftright,q_1)$,
  $x_2=(q_1,\leftleft,q_2)$,
  $x_3=(q_2,\rightright,q_3)$,
  $x_4=(q_3,\leftleft,q_4)$, \ldots,
  $x_i=(q_{i-1},\rightright,q_i)$,
  $x_{i+1}=(q_i,\leftright,q)$ 
  with $i\geq1$, 
  and $x_1,x_3,\ldots,x_i$ are steps where $C_{\leftend F_s}$ is defined and
  $x_2,x_4,\ldots,x_{i+1}\in s_\rightend$ (see Figure~\ref{steps-fin-3}). 
  
  Notice that this sequence of steps corresponds to an accepting run iff $q\in
  F$ is an accepting state of $\A$.  Therefore, either $q\notin F$ and
  $\dom{\A}\cap(\TrMorph^{-1}(s)\setminus\{\varepsilon\})=\emptyset$ so we set
  $C_s=\bot$.  Or, $q\in F$ and
  $\TrMorph^{-1}(s)\setminus\{\varepsilon\}\subseteq\dom{\A}$ so we define
  $$
  C_s=C_{\leftend F_s}(x_1)\odot\gamma_{q_1}\odot 
  C_{\leftend F_s}(x_3)\odot\gamma_{q_3}\odot 
  \cdots\odot C_{\leftend F_s}(x_{i})\odot\gamma_{q_i} \,.
  $$
  We have $\dom{C_s}=\TrMorph^{-1}(s)\setminus\{\varepsilon\}$ and for all 
  $u\in\dom{C_s}$ we have $\sem{C_s}(u)=\sem{\A}(u)$.
\end{proof}

\section{Infinite Words}
\label{sec:inf}

In this section, we start looking at regular functions on infinite words.  As in
Section~\ref{sec:fin}, we restrict our attention to two way transducers as the
model for computing regular functions.  Given a finite alphabet $\Sigma$, let
$\Sigma^{\omega}$ denote the set of infinite words over $\Sigma$, and let
$\Sigma^{\infty}=\Sigma^* \cup \Sigma^{\omega}$ be the set of all finite or
infinite words over $\Sigma$.

\subsection{Two-way transducers over $\omega$-words (\twoDMTla)}
Let $\Sigma$ be a finite input alphabet and let $\Gamma$ be a finite output alphabet.  
Let $\leftend$ be a left end marker symbol not in $\Sigma$ and let 
$\Sigma_\leftend=\Sigma\cup\{\leftend\}$. 
The input word is presented as $\leftend w$ where $w \in \Sigma^{\omega}$. 

Let $\Rr$ be a finite set of \emph{look-ahead} $\omega$-regular languages.
For the $\omega$-regular languages in $\Rr$, we may use any finite descriptions
such as $\omega$-regular expressions or automata.  Below, we will use
\emph{complete unambiguous \buchi automata} (\CUBA) \cite{Carton:2003rt}, also
called \emph{backward deterministic \buchi automata} \cite{Wilke-fsttcs17}).  A
deterministic two-way transducer (\twoDMTla) over $\omega$-words is given by
$\A=(Q,\Sigma,\Gamma,q_0,\delta,\mathcal{F},\Rr)$, where $Q$ is a finite set of
states, $q_0 \in Q$ is a unique initial state, and $\delta\colon Q \times
\Sigma_\leftend \times \Rr \mapsto Q \times \Gamma^*\times \{-1,+1\}$ is the
partial transition function.  We request that for every pair $(q,a)\in
Q\times\Sigma_\leftend$, the subset $\Rr(q,a)$ of languages $R\in\Rr$ such that
$\delta(q,a,R)$ is defined forms a partition of $\Sigma^\omega$.  This ensures
that $\A$ is complete and behaves deterministically.  The set $\mathcal{F}
\subseteq 2^Q$ specifies the Muller acceptance condition.  As in the finite
case, the reading head cannot move left while on $\leftend$.  A configuration is
represented by $w'qaw''$ where $w'a\in\leftend\Sigma^*$, $w''\in\Sigma^\omega$
and $q$ is the current state, scanning letter $a$.  From configuration $w'qaw''$,
let $R$ be the unique $\omega$-regular language in $\Rr(q,a)$ such that $w''\in
R$, the automaton outputs $\gamma$ and moves to
$$
\begin{cases}
  w'aq'w'' & \text{if } \delta(q,a,R)=(q',\gamma,+1) \\
  w'_1q'baw'' & \text{if } \delta(q,a,R)=(q',\gamma,-1) \text{ and } w'=w'_1b 
  \,.
\end{cases}
$$
The output $\gamma \in \Gamma^*$ is appended at the end of the output produced
so far.  A run $\rho$ of $\A$ on $w\in\Sigma^{\omega}$ is a sequence of
transitions starting from the initial configuration $q_0\leftend w$ where the
reading head is on $\leftend$:
$$
q_0\leftend w \xrightarrow{\gamma_1} w'_1q_1w''_1 \xrightarrow{\gamma_2}
w'_2q_2w''_2 \xrightarrow{\gamma_3} w'_3q_3w''_3 \xrightarrow{\gamma_4} 
w'_4q_4w''_4 \cdots
$$
We say that $\rho$ reads the whole word $w$ if $\supr\{|w'_n| \mid n>0\}=\infty$.
The set of states visited by $\rho$ infinitely often is denoted
$\infi(\rho)\subseteq Q$.  The word $w$ is accepted by $\A$, i.e.,
$w\in\dom{\A}$ if $\rho$ reads the whole word $w$ and $\infi(\rho)\in\Ff$ is a
Muller set.  In this case, we let
$\sem{\A}(w)=\gamma_1\gamma_2\gamma_3\gamma_4\cdots$ be the output produced by
$\rho$.

The notation \twoDMTla signifies the use of the look-ahead (la) using the
$\omega$-regular languages in $\Rr$.  It must be noted that without look-ahead,
the expressive power of two-way transducers over infinite words is lesser than
regular transformations over infinite words \cite{lics12}.  A classical
example of this is given in Example \ref{eg-two-way-inf}, where the look-ahead
is necessary to obtain the required transformation.

\begin{example}\label{eg-two-way-inf}
  On the right of Figure~\ref{fig:intro} we have
  an \twoDMTla $\Aa'$ over $\Sigma=\{a,b,\#\}$ that defines
  the transformation $\sem{\Aa'}(u_1\#u_2\#\cdots \#u_n\#v)=u^R_1u_1\#u^R_2u_2\#
  \cdots \#u^R_nu_n\#v$ where $u_1, \ldots, u_n \in (a+b)^*$, $v \in
  (a+b)^{\omega}$ and $u^R$ denotes the reverse of $u$.  The Muller acceptance
  set is $\{\{q_5\}\}$.  From state $q_1$ reading $\leftend$, or state $q_4$ 
  reading $\#$, $\Aa'$ uses the look ahead
  partition $\Rr(q_1,\leftend)=\Rr(q_4,\#)=
  \{\Sigma^*\#\Sigma^\omega, (\Sigma\setminus \{\#\})^\omega\}$,
  which indicates the presence or absence of a $\#$ in the remaining suffix of
  the word being read.  For all other transitions, the look-ahead
  langage is $\Sigma^\omega$, hence it is omitted.  Also, to keep the picture 
  light, the automaton is not complete, i.e., we have omitted the transitions 
  going to a sink state.
  It can be seen that any maximal string $u$ between two consecutive occurrences
  of $\#$ is replaced with $u^Ru$; the infinite suffix over $\{a,b\}^{\omega}$
  is then reproduced as it is.
\end{example}

\begin{remark}
  Note that, an equivalent way to define \twoDMTla is using look-behind and
  look-ahead automata \cite{lics12} instead of $\omega$-regular languages in
  $\Rr$.  See Appendix \ref{app:two-way-eq} for a proof of equivalence.
\end{remark}

\subsection{$\omega$-Regular Transducer Expressions (\oRTE)}\label{sec:oRTE}

As in the case of finite words, we define regular transducer expressions for
infinite words.  Let $\Sigma$ and $\Gamma$ be finite input and output alphabets
and let $\bot$ stand for undefined.  We define the output domain as
$\D=\Gamma^{\infty} \cup \{\bot\}$ , with the usual concatenation of a finite
word on the left with a finite or infinite word on the right.  Again, $\bot$
acts as zero and the unit is the empty word $1_{\D}=\varepsilon$.

The syntax of \emph{$\omega$-Regular Transducer Expressions} (\oRTE) from
$\Sigma^\omega$ to $\D$ is defined by:
$$
C ::= 
\Ifthenelse{L}{C}{C}  \mid C \odot C \mid
E \lrdot C \mid E^\omega \mid \twoomega{K}{E}
$$
where $K\subseteq\Sigma^+$ ranges over \emph{regular} languages of \emph{finite
non-empty words}, $L\subseteq\Sigma^{\omega}$ ranges over $\omega$-\emph{regular}
languages of \emph{infinite words} and $E$ is an \RTE over finite words as
defined in Section~\ref{sec:RegExp}.
The semantics $\sem{E}\colon\Sigma^*\to\Gamma^*\cup\{\bot\}$ of the finitary
combinator expressions $E\in\RTE$ is unchanged (see Section~\ref{sec:RegExp}).
The semantics of an \oRTE $C$ is a function $\sem{C}\colon\Sigma^\omega\to\D$.
Given a regular language $K\subseteq\Sigma^+$, an $\omega$-regular language
$L\subseteq\Sigma^\omega$, and functions
$f\colon\Sigma^*\to\Gamma^*\cup\{\bot\}$, $g,h\colon\Sigma^{\omega}\to\D$,
we define
\begin{description}
  \item[If then else.]  We have 
  $\dom{\Ifthenelse{L}{g}{h}}=(\dom{g}\cap L)\cup(\dom{h}\setminus L)$.  
  
  Moreover, $(\Ifthenelse{L}{g}{h})(w)$ is defined
  as $g(w)$ for $w\in\dom{g}\cap L$, and $h(w)$ for $w\in\dom{h}\setminus L$.

  \item[Hadamard product.] We have $\dom{g\odot h}=g^{-1}(\Gamma^*)\cap\dom{h}$. 
  
  Moreover, $(g \odot h)(w)=g(w)\cdot h(w)$ for 
  $w\in\dom{g}\cap\dom{h}$ with $g(w)\in\Gamma^*$.
  
  \item[Unambiguous Cauchy product.] If $w\in\Sigma^\omega$ admits a unique 
  factorization $w=u\cdot v$ with $u\in\dom{f}$ and $v\in\dom{g}$ then we set 
  $(f\lrdot g)(w)=f(u)\cdot g(v)$. Otherwise, we set 
  $(f\lrdot g)(w)=\bot$.

  \item[Unambiguous $\omega$-iteration.]  If $w\in\Sigma^\omega$ admits a unique
  infinite factorization $w=u_1 u_2 u_3 \cdots$ with $u_i\in\dom{f}$
  for all $i\geq1$ then we set $f^\omega(w)=f(u_1)f(u_2)f(u_3)\cdots\in\Gamma^\infty$.
  Otherwise, we set $f^\omega(w)=\bot$.

  \item[Unambiguous 2-chained $\omega$-iteration.]  If $w\in\Sigma^\omega$
  admits a unique factorization $w=u_1 u_2 u_3 \cdots$ with $u_i\in K$ for all
  $i\geq1$ and if moreover $u_iu_{i+1}\in\dom{f}$ for all $i\geq1$ then we set
  $\twoomega{K}{f}(w)=f(u_1u_2)f(u_2u_3)f(u_3u_4)\cdots$.  Otherwise, we set
  $\twoomega{K}{f}(w)=\bot$.
\end{description}

\begin{remark}\label{rem:omega-iteration-redundant}
  Let $C_\varepsilon=(\Ifthenelse{\Sigma}{\varepsilon}\bot)^\omega$.
  We have $\dom{C_\varepsilon}=\Sigma^\omega$ and 
  $\sem{C_\varepsilon}(w)=\varepsilon$ for all $w\in\Sigma^\omega$.
  Now, for $\gamma\in\Gamma^+$, let 
  $C_\gamma=(\Ifthenelse{\Sigma}{\gamma}\bot)\lrdot C_\varepsilon$. 
  We have $\dom{C_\gamma}=\Sigma^\omega$ and 
  $\sem{C_\gamma}(w)=\gamma$ for all $w\in\Sigma^\omega$.
  Therefore, we can freely use constants $\gamma\in\Gamma^*$ when defining 
  \oRTEs.
\end{remark}

\begin{remark}
  We can express the $\omega$-iteration with the 2-chained $\omega$-iteration 
  as follows: \\
  $f^\omega =
  \twoomega{\dom{f}}{f\lrdot(\Ifthenelse{\dom{f}}{\varepsilon}\bot)}$.
  \label{derived}
\end{remark}

\begin{example}
  We now give the \oRTE for the transformation given in Example
  \ref{eg-two-way-inf}.

  \noindent
  Let $E_1=\Ifthenelse{a}{a}{(\Ifthenelse{b}{b}{(\Ifthenelse{\#}{\#}{\bot})})}$,
  $E_2=\Ifthenelse{a}{a}{(\Ifthenelse{b}{b}{\bot})}$ and
  $E_3=\Ifthenelse{a}{a}{(\Ifthenelse{b}{b}{(\Ifthenelse{\#}{\varepsilon}{\bot})})}$.
  Then $\dom{E_1}=\dom{E_3}=(a+b+\#)$ and $\dom{E_2}=(a+b)$.
  \\
  Let $E_4=\Ifthenelse{((a+b)^*\#)}{(\rlplus{E_3} \odot \lrplus{E_1})}{\bot}$.
  We have $\dom{E_4}=(a+b)^*\#$ and, for $u\in(a+b)^*$, $\sem{E_4}(u\#)=u^Ru\#$ 
  where $u^R$ denotes the reverse of $u$.
  Next, let $C_1={\lrplus{E_4}\lrdot E_2^{\omega}}$.
  Then, $\dom{C_1} = [(a+b)^*\#]^+(a+b)^{\omega}$, and
  $\sem{C_1}(u_1\#u_2\#\cdots u_n\#v)=u_1^Ru_1\#u_2^Ru_2\# \cdots \#u_n^Ru_n\#v$
  when $u_i \in (a+b)^*$ and $v \in (a+b)^{\omega}$. Finally, let
  $C=\Ifthenelse{(a+b)^{\omega}}{E_2^{\omega}}{C_1}$. We have 
  $\dom{C}=[(a+b)^*\#]^*(a+b)^{\omega}$ and $\sem{C}=\sem{\Aa'}$ where $\Aa'$ is 
  the transducer on the right of Figure~\ref{fig:intro}.
\label{omega-eg}
\end{example}

The main theorem connecting  \twoDMTla and \oRTE is as follows.
\begin{theorem}\label{thm:2WST=oRTE}
  \twoDMTla and \oRTEs define the same class of functions. More precisely,
  \begin{enumerate}
    \item\label{item:oRTEto2WST} given an \oRTE $C$, we can construct an \twoDMTla $\A$ such that 
    $\sem{\A}=\sem{C}$.

    \item\label{item:2WSTto-oRTE} given an \twoDMTla $\A$, we can construct an \oRTE C such that 
    $\sem{\A}=\sem{C}$,
  \end{enumerate}
\end{theorem}

The proof of \eqref{item:oRTEto2WST} is given in the next section, while the
proof of \eqref{item:2WSTto-oRTE} will be given in Section~\ref{sec:2WSTto-oRTE}
after some preparatory work on backward deterministic B\"uchi automata
(Section~\ref{sec:BDBA}) which are used to remove the look-ahead of \twoDMTla
(Section~\ref{sec:look-ahead}), and the notion of transition monoid for
\twoDMTla (Section~\ref{sec:TrM-2DMTla}) used in the unambiguous forest
factorization theorem extended to infinite words
(Theorem~\ref{thm:U-forest-omega}).

\subsection{\oRTE to \twoDMTla}

In this section, we prove one direction of Theorem~\ref{thm:2WST=oRTE}: given an
\oRTE $C$, we can construct an \twoDMTla $\A$ such that $\sem{\A}=\sem{C}$.
The proof is by structural induction and follows immediately from

\begin{lemma}\label{lem:induction-oRTE-to-2DMT}
  Let $K\subseteq\Sigma^*$ be regular and $L\subseteq\Sigma^\omega$ be
  $\omega$-regular.  Let $f$ be an \RTE with $\sem{f}=\sem{M_f}$ for some 2DFT
  $M_f$.  Let $g,h$ be \oRTEs with $\sem{g}=\sem{M_g}$ and $\sem{h}=\sem{M_h}$
  for \twoDMTla $M_g$ and $M_h$ respectively.  Then, one can construct
  \begin{enumerate}[nosep]
    \item\label{item:oifthenelse} an \twoDMTla $\Aa$ such that
    $\sem{\Ifthenelse{L}{g}h}=\sem{\Aa}$,
    
    \item\label{item:oHadamard} an \twoDMTla $\Aa$ such that 
    $\sem{\Aa}=\sem{g\odot h}$,
    
    \item\label{item:oCauchy} an \twoDMTla $\Aa$ such that 
    $\sem{\Aa}=\sem{g\lrdot h}$,
    
    \item\label{item:omega} an \twoDMTla $\Aa$ such that 
    $\sem{\Aa}=\sem{f^{\omega}}$,
   
    \item\label{item:twoomega} an \twoDMTla $\Aa$ such that 
    $\sem{\Aa}=\sem{\twoomega{K}{f}}$.
  \end{enumerate}    
\end{lemma}

\begin{proof}
  Throughout the proof, we let $M_g = (Q_g, \Sigma, \Gamma, s_g, \delta_g
  \mathcal{F}_g, \Rr_g)$ and $M_h = (Q_h, \Sigma, \Gamma, s_h, \delta_h,
  \mathcal{F}_h, \Rr_h)$ be the be the \twoDMTla such that $\sem{M_g}=\sem{g}$
  and $\sem{M_h}= \sem{h}$.

  \medskip\noindent\eqref{item:oifthenelse} \textbf{If then else.}
  The set of states of $\Aa$ is
  $Q_{\Aa} = \{q_0\} \cup Q_{g} \cup Q_{h}$ with $q_0\notin Q_g\cup Q_h$.
	In state $q_0$, we have the transitions $\delta_{\Aa}(q_0, (\leftend,R\cap
	L))=(q, \gamma, +1)$ if $\delta_g(s_g, (\leftend, R)) = (q, \gamma, +1)$ and
	$\delta_{\Aa}(q_0, (\leftend, R'\setminus L))=(q',\gamma',+1)$ if
	$\delta_h(s_h, (\leftend, R')) = (q', \gamma', +1)$.  This invokes $M_g$
	($M_h$) iff the input $w$ is in $L$ (not in $L$).  The Muller set
	$\mathcal{F}$ is simply a union $\mathcal{F}_g \cup \mathcal{F}_h$ of the
	respective Muller sets of $M_g$ and $M_h$.  It is clear that $\sem{\Aa}$
	coincides with $\sem{M_g}$ iff the input string is in $L$, and otherwise,
	$\sem{\Aa}$ coincides with $\sem{M_h}$.

  \medskip\noindent\eqref{item:oHadamard} \textbf{Hadamard product.}
  We create a look ahead which indicates the position where we can stop reading
  the input word $w$ for the transducer $M_g$.
	The look ahead should satisfy two conditions for this purpose:
	\begin{itemize}[nosep]
		\item We cannot visit any position  to the left of the current position in 
    the remaining run of $M_g$ on $w$.
		\item The output produced by running $M_g$ on the suffix should be $\varepsilon$.
	\end{itemize}
	To accommodate these two conditions, we create look ahead automata $A_q$ for
	each state $q \in Q_g$ and let $L_q=\dom{A_q}$.  The structure of $A_q$ is same
	as $M_g$ except that we
	\begin{itemize}[nosep]
		\item add a new initial state $\iota_q$ and the transition
		$\delta_q(\iota_q, \leftend, \Sigma^\omega) = (q, \varepsilon, +1)$,

    \item remove all transitions from $M_g$ where the output is $\gamma \neq
    \varepsilon$,

    \item remove all transitions from $M_g$ where the input symbol is
    $\leftend$.
	\end{itemize}
	We explain the construction of the \twoDMTla $\Aa$ such that $\sem{g \odot h}=\sem{\Aa}$.
	The set of states of $\Aa$ are $Q_{\Aa} = Q_{g} \cup Q_{h} \cup \{\mathsf{reset}\}$.
  Backward transitions in $\Aa$ and $M_g$ are same:
  $\delta_\Aa(q,a,R)=(q',\gamma,-1)$ iff $\delta_g(q,a,R)=(q',\gamma,-1)$.
  Forward
  transitions of $M_g$ are divided into two depending on the look ahead.  If we
  have $\delta_g(q,a,R)=(q',\gamma,+1)$ in $M_g$ for an $a\in\Sigma_\leftend$,
  then \par
  \hfil$\delta_\Aa(q,a,R\setminus L_{q'})=(q',\gamma,+1)$ 
  \hfil and \hfil
  $\delta_\Aa(q,a,R\cap L_{q'})=(\mathsf{reset},\gamma,+1)$.
  \\
	From the $\mathsf{reset}$ state, we go to the left until $\leftend$ is reached
	and then start running $M_h$.  So, $\delta_\Aa(\mathsf{reset}, a, \Sigma^\omega) =
	(\mathsf{reset}, \varepsilon, -1)$ for all $a \in \Sigma$ and
	$\delta_\Aa(\mathsf{reset}, \leftend, R) = (q'', \gamma, +1)$ if $\delta_h(s_h,
	\leftend, R) = (q'', \gamma, +1)$.  The accepting set is same as the Muller
	accepting set $\Ff_h$ of $M_h$.

  \medskip\noindent\eqref{item:oCauchy} \textbf{Cauchy product.}
	From the transducers $M_f$ and $M_g$, we can construct a DFA $\Dd_f = (Q_f,
	\Sigma, \delta_f, s_f, F_f)$ that accepts $\dom{M_f}$ and a 
  deterministic Muller automaton (DMA)
	$\Dd_g = (Q_g, \Sigma, \delta_g, s_g, \mathcal{F}_g)$ that
	accepts $\dom{M_g}$.  
	
	Now, the set $L$ of words $w$ having at least two factorizations $w = u_1v_1 =
	u_2v_2$ with $u_1,u_2\in\dom{f}$, $v_1,v_2\in\dom{g}$ and $u_1\neq u_2$ is
	$\omega$-regular.  This is easy since $L$ can be written as $L=\bigcup_{p\in
	F_f,q\in Q_g} L_p\cdot M_{p,q}\cdot R_q$ where
	\begin{itemize}[nosep]
		\item $L_p\subseteq\Sigma^{*}$ is the regular set of words which admit a run
		in $\Dd_f$ from its initial state to state $p$,
	
		\item $M_{p,q}\subseteq\Sigma^{*}$ is the regular set of words which admit a
		run in $\Dd_f$ from state $p$ to some final state in $\Dd_f$, and also admit
		a run in $\Dd_g$ from the initial state to some state $q$ in $\Dd_g$,
	
		\item $R_{q}\subseteq\Sigma^{\omega}$ is the $\omega$-regular set of words
		which (i) admit an \emph{accepting} run from state $q$ in $\Dd_g$ and also (ii)
		admit an \emph{accepting} run in $\Dd_g$ from its initial state $s_g$.
 	\end{itemize}	
	Therefore, $\dom{f\lrdot g} = (\dom{f}\cdot\dom{g})\setminus L$ is
	$\omega$-regular.
	
  First we construct an $\omega$-1DMT\textsubscript{la} $\Dd$ such that $\dom{\Dd} =
  \dom{f \lrdot g}$ and on an input word $w = uv$ with $u\in\dom{f}$ and
  $v\in\dom{g}$, it produces the output $u \# v$ where $\# \notin \Sigma$ is a
  new symbol.  From its initial state while reading $\leftend$, $\Dd$ uses the 
  look-ahead to check whether the input word $w$ is in $\dom{f\lrdot g}$ or not.
  If yes, it moves right and enters the initial state of $\Dd_f$. If not, it 
  goes to a sink state and rejects.
  While running $\Dd_f$, $\Dd$ copies each input letter to output.  
  Upon reaching a final state of $\Dd_f$, we use the look-ahead $\dom{g}$ to 
  see whether we should continue running $\Dd_f$ or we should switch to 
  $\Dd_g$. Formally, if $\delta_f(q,a)=q'\in F_f$ the corresponding transitions 
  of $\Dd$ are \par
  \hfil$\delta_\Dd(q,a,\dom{g})=(s_g,a \#,+1)$ 
  \hfil and \hfil
  $\delta_\Dd(q,a,\Sigma^{\omega}\setminus\dom{g})=(q',a,+1)$.
  \\
  While running $\Dd_g$, $\Dd$ copies each input letter to output.  
  Accepting sets of $\Dd$ are the accepting sets of the DMA
  $\Dd_g$.  Thus, $\Dd$ produces an output $u \# v$ for an input string $w=uv$
  which is in $\dom{f\lrdot g}$ such that $u \in \dom f$ and $v \in \dom g$.
	
	Next we construct an \twoDMTla $\Tt$ which takes input words of the form $u \# v$
	with $u \in \Sigma^*$ and $v \in \Sigma^\omega$, runs $M_f$ on $u$ and $M_g$
	on $v$.  To do so, $u$ is traversed in either direction depending on $M_f$ and
	the symbol $\#$ is interpreted as right end marker $\rightend$ for $M_f$.
	While simulating a transition of $M_f$ moving right of $\rightend$, producing
	the output $\gamma$ and reaching state $q$, there are two possibilities.  If
	$q$ is not a final state of $M_f$ then $\Tt$ moves to the right of $\#$, goes
	to some sink state and rejects.  If $q$ is a final state of $M_f$, then $\Tt$
	stays on $\#$ producing the output $\gamma$ and goes to the initial state of
	$M_g$.  Then, $\Tt$ runs $M_g$ on $v$ interpreting $\#$ as $\leftend$.  The
	Muller accepting set of $\Tt$ is same as $M_g$.
  	
  We construct an \twoDMTla $\Aa$ as the composition of $\Dd$ and $\Tt$.  Regular
  transformations are definable by \twoDMTla \cite{lics12} and are closed under
  composition \cite{Courcelle:1997:EGP:278918.278932}.  Thus the composition of
  an $\omega$-1DMT\textsubscript{la} and an \twoDMTla is an \twoDMTla. We deduce that $\Aa$
  is an \twoDMTla. Moreover $\sem{\Aa} = \sem{f \lrdot g}$.

  \medskip\noindent\eqref{item:omega} \textbf{$\omega$-iteration.}
	By the remark above Example~\ref{omega-eg}, this is a derived operator
	and hence the result follows from the next case.

  \medskip\noindent\eqref{item:twoomega} \textbf{2-chained $\omega$-iteration.}
	First we show that the set of words $w$ in $\Sigma^\omega$ having an
	unambiguous decomposition $w = u_1 u_2 \cdots$ with $u_i \in K$ for each $i$ is
	$\omega$-regular.  As in case \eqref{item:oCauchy} above, the
	language $L$ of words $w$ having at least two factorizations $w=u_1v_1=u_2v_2$
	with $u_1,u_2\in K$, $v_1,v_2\in K^\omega$ and $u_1\neq u_2$ is $\omega$-regular.
	Hence, $L'=K^*\cdot L$ is $\omega$-regular and contains all words in $\Sigma^\omega$
	having several factorizations as products of words in $K$.  We deduce
	that $\Sigma^\omega \setminus L'$ is $\omega$-regular.
	
	As in case \eqref{item:oCauchy} above, we construct an
	$\omega$-1DMT\textsubscript{la} $\Dd$ which takes as input $w$ and outputs
	$u_1 \# u_2 \# \cdots$ iff there is an unambiguous decomposition of $w$ as $u_1
	u_2 \cdots$ with each $u_i \in K$.
	We then construct an $\omega$-2DMT $\Dd'$ that takes as input words of the
	form $u_1 \# u_2 \# \cdots$ with each $u_i \in \Sigma^*$ and produces $u_1u_2
	\# u_2u_3 \# \cdots$.
	
	Next we construct an $\omega$-2DMT $\Tt$ that takes as input words of the form
	$w_1 \# w_2 \# \cdots $ with each $w_i \in \Sigma^*$ and runs $M_f$ on each
	$w_i$ from left to right.  The transducer $\Tt$ interprets $\#$ as $\leftend$
	(resp.\ $\rightend$) when it is reached from the right (resp.\ from left).
  While simulating a transition of $M_f$ moving right of $\rightend$, we 
  proceed as in case \eqref{item:oCauchy} above, except
	that $\Tt$ goes to the initial state of $M_f$ instead.
	
	The \twoDMTla $\Aa$ is then obtained as the composition of $\Dd$, $\Dd'$ and $\Tt$.
	The output produced by $\Aa$ is thus $\sem{M_f}(u_1u_2) \sem{M_f}(u_2u_3)
	\cdots$.
\end{proof}

\subsection{Backward deterministic \buchi automata (\BDBA)}\label{sec:BDBA}

A \buchi automaton over the input alphabet $\Sigma$ is a tuple
$\Bb=(P,\Sigma,\Delta,\Fin)$ where $P$ is a finite set of states, $\Fin\subseteq
P$ is the set of final (accepting) states, and $\Delta\subseteq
P\times\Sigma\times P$ is the transition relation.  A run of $\Bb$ over an
infinite word $w=a_1a_2a_3\cdots$ is a sequence
$\rho=p_0,a_1,p_1,a_2,p_2,\ldots$ such that $(p_{i-1},a_i,p_i)\in\Delta$ for all
$i\geq 1$.  The run is final (accepting) if $\infi(\rho)\cap\Fin\neq\emptyset$ where
$\infi(\rho)$ is the set of states visited infinitely often by $\rho$. 

The \buchi automaton $\Bb$ is \emph{backward deterministic} (\BDBA) or
\emph{complete unambiguous} (\CUBA) if for all infinite words
$w\in\Sigma^\omega$, there is \emph{exactly one} run $\rho$ of $\Bb$ over $w$
which is final, this run is denoted $\Bb(w)$.  The fact that we request at
least/most one final run on $w$ explains why the automaton is called
complete/unambiguous.  Wlog, we may assume that all states of $\Bb$ are
\emph{useful}, i.e., for all $p\in P$ there exists some $w\in\Sigma^\omega$ such
that $\Bb(w)$ starts from state $p$.  In that case, it is easy to check that the
transition relation is \emph{backward deterministic and complete}: for all
$(p,a)\in P\times\Sigma$ there is exactly one state $p'$ such that
$(p',a,p)\in\Delta$.  We write $p'\xleftarrow{a}p$ and state $p'$ is denoted
$\Delta^{-1}(p,a)$.  In other words, the inverse of the transition relation
$\Delta^{-1}\colon P\times\Sigma\to P$ is a total function.

For each state $p\in P$, we let $\Lang{\Bb,p}$ be the set of infinite words 
$w\in\Sigma^\omega$ such that $\Bb(w)$ starts from $p$. For every subset 
$I\subseteq P$ of initial states, the language $\Lang{\Bb,I}=\bigcup_{p\in 
I}\Lang{\Bb,p}$ is $\omega$-regular. 

\begin{example}\label{ex:BDBA}
  For instance, the automaton $\Bb$ below is a \BDBA. Morover, we have 
  $\Lang{\Bb,p_2}=(\Sigma\setminus\{\#\})^{\omega}$, 
  $\Lang{\Bb,p_4}=(\#\Sigma^{*})^{\omega}$, and
  $\Lang{\Bb,\{p_1,p_3,p_4\}}=\Sigma^{*}\#\Sigma^{\omega}$.
  
  \medskip
  \centerline{\includegraphics[page=4,scale=1]{gpicture-pics.pdf}}
\end{example}

Deterministic \buchi automata (\DBA) are 
strictly weaker than non-deterministic \buchi automata (\NBA) but backward 
determinism keeps the full expressive power.

\begin{theorem}[Carton \& Michel \cite{Carton:2003rt}]\label{thm:BDBA}
  A language $L\subseteq\Sigma^\omega$ is $\omega$-regular iff $L=\Lang{\Bb,I}$ 
  for some \BDBA $\Bb$ and initial set $I$.
\end{theorem}

The proof in \cite{Carton:2003rt} is constructive, starting with an \NBA with 
$m$ states, they construct an equivalent \BDBA with $(3m)^m$ states.

A crucial fact on \BDBA is that they are easily closed under boolean operations.
In particular, the complement, which is quite difficult for \NBAs, becomes
trivial with \BDBAs: $\Lang{\Bb,P\setminus
I}=\Sigma^\omega\setminus\Lang{\Bb,I}$.  For intersection and union, we simply
use the classical cartesian product of two automata $\Bb_1$ and $\Bb_2$.  This
clearly preserves the backward determinism.  For intersection, we use a
generalized \buchi acceptance condition, i.e., a conjunction of \buchi
acceptance conditions.  For \BDBAs, generalized and classical \buchi acceptance
conditions are equivalent \cite{Carton:2003rt}.  We obtain immediately

\begin{corollary}\label{cor:BDBA}
  Let $\Rr$ be a finite family of $\omega$-regular languages.  There is a \BDBA
  $\Bb$ and a tuple of initial sets $(I_R)_{R\in\Rr}$ such that
  $R=\Lang{\Bb,I_R}$ for all $R\in\Rr$.
\end{corollary}

\subsection{Replacing the look-ahead of an \twoDMTla with a \BDBA}\label{sec:look-ahead}

Let $\A=(Q,\Sigma,\Gamma,q_0,\delta,\mathcal{F},\Rr)$ be an \twoDMTla.  By
Corollary~\ref{cor:BDBA} there is a \BDBA $\Bb=(P,\Sigma,\Delta,\Fin)$ and a
tuple $(I_R)_{R\in\Rr}$ of initial sets for the finite family $\Rr$ of
$\omega$-regular languages used as look-ahead by the automaton $\A$.  Recall
that for every pair $(q,a)\in Q\times\Sigma_\leftend$, the subset $\Rr(q,a)$ of
languages $R\in\Rr$ such that $\delta(q,a,R)$ is defined forms
a partition of $\Sigma^\omega$.  We deduce that $(I_R)_{R\in\Rr(q,a)}$
is a partition of $P$.

We construct an \twoDMT
$\tilde{\A}=(Q,\tilde{\Sigma},\Gamma,q_0,\tilde{\delta},\mathcal{F})$ \emph{without
look-ahead} over the extended alphabet $\tilde{\Sigma}=\Sigma\times P$ which is
equivalent to $\A$ in some sense made precise below.  Intuitively, in a pair 
$(a,p)\in\tilde{\Sigma_\leftend}$, the state $p$ of $\Bb$ gives the 
look-ahead information required by $\A$. Formally, the deterministic 
transition function $\tilde{\delta}\colon Q\times\tilde{\Sigma_\leftend}\to 
Q\times\Gamma^*\times\{-1,+1\}$ is defined as follows: for $q\in Q$ and 
$(a,p)\in\tilde{\Sigma_\leftend}$ we let 
$\tilde{\delta}(q,(a,p))=\delta(q,a,R)$ for 
the unique $R\in\Rr(q,a)$ such that $p\in I_R$.

\begin{example}\label{ex:tilde-A}
  For instance, the automaton $\tilde{\A}$ constructed from the automaton on 
  the right of Figure~\ref{fig:intro} and the \BDBA $\Bb$ of 
  Example~\ref{ex:BDBA} is depicted below, where $\bullet$ stands for an 
  arbitrary state of $\Bb$. 
  
  \medskip
  \centerline{\includegraphics[page=5,scale=1]{gpicture-pics.pdf}}
\end{example}

Let $w=a_1a_2a_3\cdots\in\Sigma^\omega$ and let
$\Bb(w)=p_0,a_1,p_1,a_2,p_2,\ldots$ be the unique final run of $\Bb$ on $w$.  We
define $\tilde{\leftend w}=
(\leftend,p_0)(a_1,p_1)(a_2,p_2)\cdots\in\tilde{\Sigma_\leftend}^\omega$.
We can easily check by induction that the unique run of $\A$ on $w$
$$
q_0\leftend w \xrightarrow{\gamma_1} w'_1q_1w''_1 \xrightarrow{\gamma_2}
w'_2q_2w''_2 \xrightarrow{\gamma_3} w'_3q_3w''_3 \xrightarrow{\gamma_4} 
w'_4q_4w''_4 \cdots
$$
corresponds to the unique run of $\tilde{\A}$ on $\tilde{\leftend w}$
$$
q_0\tilde{\leftend w} \xrightarrow{\gamma_1} \tilde{w'_1}q_1\tilde{w''_1} \xrightarrow{\gamma_2}
\tilde{w'_2}q_2\tilde{w''_2} \xrightarrow{\gamma_3} \tilde{w'_3}q_3\tilde{w''_3} \xrightarrow{\gamma_4} 
\tilde{w'_4}q_4\tilde{w''_4} \cdots
$$
where for all $i>0$ we have $\tilde{\leftend w}=\tilde{w'_i}\tilde{w''_i}$ and
$|w'_i|=|\tilde{w'_i}|$.  Indeed, assume that in a configuration $w'qaw''$ with
$\leftend w=w'aw''$ the transducer $\A$ takes the transition
$q\xrightarrow{(a,R)}(q',\gamma,+1)$ and reaches configuration $w'aq'w''$.
Then, $w''\in R$ and the corresponding configuration $\tilde{w'}q(a,p)\tilde{w''}$ with
$\tilde{\leftend w}=\tilde{w'}(a,p)\tilde{w''}$ and $|w'|=|\tilde{w'}|$ is such 
that $p\in I_R$. Therefore, the
transducer $\tilde{\A}$ takes the transition $q\xrightarrow{(a,p)}(q',\gamma,+1)$ and
reaches configuration $\tilde{w'}(a,p)q'\tilde{w''}$. The proof is similar for 
backward transitions. We have shown that $\A$ and $\tilde{\A}$ are equivalent 
in the following sense:

\begin{lemma}\label{lem:removing-look-ahead}
  For all words $w\in\Sigma^\omega$, the transducer $\A$ starting from $\leftend
  w$ accepts iff the transducer $\tilde{\A}$ starting from $\tilde{\leftend w}$
  accepts, and in this case they compute the same output in $\Gamma^\infty$.
\end{lemma}

\subsection{Transition monoid of an \twoDMTla}\label{sec:TrM-2DMTla}

We use the notations of the previous sections, in particular for the \twoDMTla 
$\A$, the \BDBA $\Bb$ and the corresponding \twoDMT $\tilde{\A}$.
As in the case of 2NFAs over finite words, we will define a congruence on 
$\Sigma^+$ such that two words $u,v\in\Sigma^+$ are equivalent iff they behave 
the same in the \twoDMTla $\A$, when placed in an arbitrary right context 
$w\in\Sigma^\omega$. The right context $w$ is abstracted with the first state 
$p$ of the unique final run $\Bb(w)$.

The \twoDMT $\tilde{\A}$ does not use look-ahead, hence, we may use for
$\tilde{\A}$ the classical notion of transition monoid.  Actually, in order to
handle the Muller acceptance condition of $\tilde{\A}$, we need a slight
extension of the transition monoid defined in Section~\ref{sec:TrM-2NFA}.  More
precisely, the abstraction of a finite word $\tilde{u}\in\tilde{\Sigma}^+$ will
be the set $\tilde{\TrMorph}(\tilde{u})$ of tuples $(q,d,X,q')$ with $q,q'\in
Q$, $X\subseteq Q$ and $d\in\{\leftright,\leftleft,\rightright,\rightleft\}$
such that the unique run of $\tilde{\A}$ on $\tilde{u}$ starting in state $q$ on
the left of $\tilde{u}$ if $d\in\{\leftright,\leftleft\}$ (resp.\ on the right
if $d\in\{\rightright,\rightleft\}$) exits in state $q'$ on the left of
$\tilde{u}$ if $d\in\{\leftleft,\rightleft\}$ (resp.\ on the right if
$d\in\{\leftright,\rightright\}$) and visits the set of states $X$ \emph{while
in} $\tilde{u}$ (i.e., including $q$ but not $q'$ unless $q'$ is also visited
before the run exits $\tilde{u}$).

For instance, with the automaton $\tilde{\A}$ of Example~\ref{ex:tilde-A}, we
have $(q_4,\leftright,\{q_2,q_3,q_4\},q_5)\in\tilde{\TrMorph}(\tilde{u})$ when
$\tilde{u}\in((a,p_1)+(b,p_1))^{*}(\#,p_1)((a,p_1)+(b,p_1))^{*}(\#,p_2)$.

We denote by
$\tilde{\TrMon}=\{\tilde{\TrMorph}(\tilde{u})\mid\tilde{u}\in\tilde{\Sigma}^+\}\cup\{\unit_{\tilde{\TrMon}}\}$
the transition monoid of $\tilde{\A}$ with unit $\unit_{\tilde{\TrMon}}$.  
The classical product is extended by taking the union of the sets $X$ occurring 
in a sequence of steps. For instance, if we have steps
$(q_0,\leftright,X_1,q_1)$, $(q_2,\rightright,X_3,q_3)$, \ldots, 
$(q_{i-1},\rightright,X_i,q_i)$ in $\tilde{\TrMorph}(\tilde{u})$ and
$(q_1,\leftleft,X_2,q_2)$, $(q_3,\leftleft,X_4,q_4)$, \ldots, 
$(q_i,\leftright,X_{i+1},q_{i+1})$ in $\tilde{\TrMorph}(\tilde{v})$ then there is a 
step $(q_0,\leftright,X_1\cup\cdots\cup X_{i+1},q_{i+1})$ in 
$\tilde{\TrMorph}(\tilde{u}\cdot\tilde{v})=
\tilde{\TrMorph}(\tilde{u})\cdot\tilde{\TrMorph}(\tilde{v})$.
We denote by $\tilde{\TrMorph}\colon\tilde{\Sigma}^*\to\tilde{\TrMon}$ the
canonical morphism.

Let $u=a_1\cdots a_n\in\Sigma^+$ be a finite word of length $n>0$ and let $p\in
P$.  We define the sequence of states $p_0,p_1,\ldots,p_n$ by $p_n=p$ and for
all $0\leq i<n$ we have $p_{i}\xleftarrow{a_{i+1}}p_{i+1}$ in $\Bb$.  Notice
that for all infinite words $w\in\Lang{\Bb,p}$, the unique run $\Bb(uw)$ starts
with $p_0,a_1,p_1,\ldots,a_n,p_n$. We define 
$\tilde{u}^p=(a_1,p_1)(a_2,p_2)\cdots(a_n,p_n)\in\tilde{\Sigma}^+$.

We are now ready to define the finite abstraction $\TrMorph(u)$ of a finite
word $u\in\Sigma^+$ with respect to the pair $(\A,\Bb)$: we let
$\TrMorph(u)=(r^p,b^p,s^p)_{p\in P}$ where for each $p\in P$,
$s^p=\tilde{\TrMorph}(\tilde{u}^p)\in\tilde{\TrMon}$ is the abstraction of
$\tilde{u}^p$ with respect to $\tilde{\A}$, $r^p\in P$ is the unique state of
$\Bb$ such that $r^p\xleftarrow{u}p$, $b^p=1$ if the word $\tilde{u}^p$ contains
a final state of $\Bb$ and $b^p=0$ otherwise.

The transition monoid of $(\A,\Bb)$ is the set $\TrMon=\{\TrMorph(u)\mid
u\in\Sigma^+\}\cup\{\unit_\TrMon\}$ where $\unit_\TrMon$ is the unit.  The
product of $\sigma_1=(r_1^p,b_1^p,s_1^p)_{p\in P}$ and
$\sigma=(r^p,b^p,s^p)_{p\in P}$ is defined to be
$\sigma_1\cdot\sigma=(r_1^{r^p},b_1^{r^p}\vee b^p, s_1^{r^p}\cdot s^p)_{p\in
P}$.  We can check that this product is associative, so that
$(\TrMon,\cdot,\unit_\TrMon)$ is a monoid.  Moreover, let $u,v\in\Sigma^+$ be
such that $\TrMorph(u)=\sigma_1$ and $\TrMorph(v)=\sigma$.  For each $p\in P$,
we can check that $\tilde{uv}^p=\tilde{u}^{r^p}\cdot\tilde{v}^p$.  We deduce
easily that $\TrMorph(uv)=\sigma_1\cdot\sigma=\TrMorph(u)\cdot\TrMorph(v)$.
Therefore, $\TrMorph\colon\Sigma^*\to\TrMon$ is a morphism.

\subsection{\twoDMTla to \oRTE}\label{sec:2WSTto-oRTE}

We prove in this section that from an \twoDMTla $\A$ we can construct an
equivalent \oRTE. The proof follows the ideas already used for finite words in
Section~\ref{sec:2DFT-RTE}.  We will use the following generalization to
infinite words of the unambiguous forest factorization Theorem~\ref{thm:U-forest}.

\begin{theorem}[Unambiguous Forest Factorization \cite{GastinKrishna-UFF}]\label{thm:U-forest-omega}
  Let $\varphi\colon\Sigma^*\to S$ be a morphism to a finite monoid
  $(S,\cdot,\unitS)$.  There is an \emph{unambiguous} rational expression
  $G=\bigcup_{k=1}^{m}F_k\cdot G_k^\omega$ over $\Sigma$ such that
  $\Lang{G}=\Sigma^\omega$ and for all $1\leq k\leq m$ the expressions $F_k$ and
  $G_k$ are $\varepsilon$-free $\varphi$-good rational expressions and $s_{G_k}$
  is an idempotent, where $\varphi(G_k)=\{s_{G_k}\}$.
\end{theorem}

We will apply this theorem to the morphism $\TrMorph\colon\Sigma^*\to\TrMon$
defined in Section~\ref{sec:TrM-2DMTla}.  We use the unambiguous expression
$G=\bigcup_{k=1}^{m}F_k\cdot G_k^\omega$ as a \emph{guide} when constructing
\oRTEs corresponding to the \twoDMTla $\A$.

\begin{lemma}\label{lem:C_F^p}
  Let $G$ be an $\varepsilon$-free $\TrMorph$-good rational expression and let
  $\TrMorph(G)=\sigma_G=(r_G^p,b_G^p,s_G^p)_{p\in P}$ be the corresponding
  element of the transition monoid $\TrMon$ of $(\A,\Bb)$.  For each state $p\in
  P$, we can construct a map $C_G^p\colon s_G^p\to\RTE$ such that for each step
  $x=(q,d,X,q')\in s_G^p$ the following invariants hold:
  \begin{enumerate}[nosep,label=($\mathsf{J}_{\arabic*}$),ref=$\mathsf{J}_{\arabic*}$]
    \item\label{oInv1} $\dom{C_G^p(x)}=\Lang{G}$,

    \item\label{oInv2} for each $u\in\Lang{G}$, $\sem{C_G^p(x)}(u)$ is the output
    produced by $\tilde{\A}$ when running step $x$ on $\tilde{u}^p$ (i.e., running
    $\tilde{\A}$ on $\tilde{u}^p$ from $q$ to $q'$ following direction $d$).
  \end{enumerate}
\end{lemma}

\begin{proof}
  The proof is by structural induction on the rational expression.  For each
  subexpression $E$ of $G$ we let
  $\TrMorph(E)=\sigma_E=(r_E^p,b_E^p,s_E^p)_{p\in P}$ be the corresponding
  element of the transition monoid $\TrMon$ of $(\A,\Bb)$.  We start with atomic
  regular expressions.  Since $G$ is $\varepsilon$-free and $\emptyset$-free, we
  do not need to consider $E=\varepsilon$ or $E=\emptyset$.  The construction is
  similar to the one given in Section~\ref{sec:2DFT-RTE}.  The interesting cases
  are concatenation and Kleene-plus.

  \begin{description}
    \item[atomic] Assume that $E=a\in\Sigma$ is an atomic subexpression.  Notice
    that $\tilde{a}^p=(a,p)$ for all $p\in P$.  Since the \twoDMT $\tilde{\A}$ is
    deterministic and complete, for each state $q\in Q$ we have
    \begin{itemize}[nosep]
      \item either $\tilde{\delta}(q,(a,p))=(q',\gamma,1)$ and we let
      $C_a^p((q,\leftright,\{q\},q'))=C_a^p((q,\rightright,\{q\},q'))=\Ifthenelse{a}{\gamma}\bot$,
    
      \item or $\tilde{\delta}(q,(a,p))=(q',\gamma,-1)$ and we let
      $C_a^p((q,\leftleft,\{q\},q'))=C_a^p((q,\rightleft,\{q\},q'))=\Ifthenelse{a}{\gamma}\bot$.
    \end{itemize}
    Clearly, invariants \eqref{oInv1} and \eqref{oInv2} hold for all 
    $x\in s_E^p$.
    
    \item[Union] Assume that $E=E_1\cup E_2$.  Since $E$ is good,
    we deduce that
    $\sigma_E=\sigma_{E_1}=\sigma_{E_2}$.  For each $p\in P$ and $x\in s_E^p$ we
    define $C_E^p(x)=\Ifthenelse{E_1}{C_{E_1}^p(x)}{C_{E_2}^p(x)}$.  Since $E$ is
    unambiguous we have $\Lang{E_1}\cap\Lang{E_2}=\emptyset$.  As in
    Section~\ref{sec:2DFT-RTE} we can prove easily that invariants \eqref{oInv1}
    and \eqref{oInv2} hold for all $x\in s_E^p$.

    \item[concatenation] Assume that $E=E_1\cdot E_2$ is a concatenation.  
    Since $E$ is good, we deduce that
    $\sigma_E=\sigma_{E_1}\cdot \sigma_{E_2}$. Let $p\in P$ and $p_1=r_{E_2}^p$. 
    We have $s_E^p = s_{E_1}^{p_1}\cdot s_{E_2}^p$. Let $x\in s_E^p$.

      If $x=(q,\leftright,X,q')$ then, by definition of the product in the 
      transition monoid $\tilde{\TrMon}$, there is a unique sequence of steps 
      $x_1=(q,\leftright,X_1,q_1)$,
      $x_2=(q_1,\leftleft,X_2,q_2)$,
      $x_3=(q_2,\rightright,X_3,q_3)$,
      $x_4=(q_3,\leftleft,X_4,q_4)$, \ldots,
      $x_i=(q_{i-1},\rightright,X_i,q_i)$,
      $x_{i+1}=(q_i,\leftright,X_{i+1},q')$ 
      with $i\geq1$, 
      $x_1,x_3,\ldots,x_i\in s_{E_1}^{p_1}$ and
      $x_2,x_4,\ldots,x_{i+1}\in s_{E_2}^p$ 
      and $X=X_1\cup\cdots\cup X_{i+1}$ (see Figure~\ref{fig:CFp} top left).
      We define
      $$
      C_E^p(x)=(C_{E_1}^{p_1}(x_1)\lrdot C_{E_2}^{p}(x_2))\odot
      (C_{E_1}^{p_1}(x_3)\lrdot C_{E_2}^{p}(x_4))\odot
      \cdots\odot(C_{E_1}^{p_1}(x_i)\lrdot C_{E_2}^{p}(x_{i+1})) \,.
      $$
      Notice that when $i=1$ we  have $C_E^p(x)=C_{E_1}^{p_1}(x_1)\lrdot 
      C_{E_2}^{p}(x_2)$ with $x_2=(q_1,\leftright,X_2,q')$.
      
      The concatenation $\Lang{E}=\Lang{E_1}\cdot\Lang{E_2}$ is unambiguous.
      Therefore, for all $y\in s_{E_1}^{p_1}$ and $z\in s_{E_2}^{p}$, using \eqref{oInv1} for
      $E_1$ and $E_2$, we obtain $\dom{C_{E_1}^{p_1}(y)\lrdot C_{E_2}^{p}(z)}=\Lang{E}$.  We
      deduce that $\dom{C_E(x)}=\Lang{E}$ and \eqref{oInv1} holds for $E$ and $x=(q,\leftright,X,q')$.

      Now, let $u\in\Lang{E}$ and let $u=u_1u_2$ be its unique factorization with
      $u_1\in\Lang{E_1}$ and $u_2\in\Lang{E_2}$.  We have
      $\tilde{u_1u_2}^p=\tilde{u_1}^{p_1}\cdot\tilde{u_2}^p$.  Hence, the step
      $x=(q,\leftright,X,q')$ performed by $\tilde{\A}$ on $\tilde{u}^p$ is
      actually the concatenation of steps $x_1$ on $\tilde{u_1}^{p_1}$, followed
      by $x_2$ on $\tilde{u_2}^p$, followed by $x_3$ on $\tilde{u_1}^{p_1}$,
      followed by $x_4$ on $\tilde{u_2}^p$, \ldots, until $x_{i+1}$ on
      $\tilde{u_2}^p$.  Using \eqref{oInv2} for $E_1$ and $E_2$, we deduce that
      the output produced by $\tilde{\A}$ while making step $x$ on $\tilde{u}^p$
      is
      $$
      \sem{C_{E_1}^{p_1}(x_1)}(u_1) \cdot \sem{C_{E_2}^{p}(x_2)}(u_2) 
      \cdots
      \sem{C_{E_1}^{p_1}(x_i)}(u_1) \cdot \sem{C_{E_2}^{p}(x_{i+1})}(u_2)
      = \sem{C_E^p(x)}(u) \,.
      $$
      Therefore, \eqref{oInv2} holds for $E$ and step $x=(q,\leftright,X,q')$.
      The proof is obtained mutatis mutandis for the other cases 
      $x=(q,\leftleft,X,q')$ or $x=(q,\rightright,X,q')$ or 
      $x=(q,\rightleft,X,q')$.
      
      \item[Kleene-plus] Assume that $E=F^+$.  Since $E$ is good,
      we deduce that 
      $\sigma_E=\sigma_F=\sigma=(r^p,b^p,s^p)_{p\in P}$ is an
      idempotent of the transition monoid $\TrMon$.  Notice that for all $p\in 
      P$, since $\sigma$ is an idempotent, we have $r^{r^p}=r^p$.
      
      We first define $C_E^p$ for states $p\in P$ such that $p=r^p$.  Let $x\in s^p$.
      \begin{itemize}[nosep]
        \item If $x=(q,\leftleft,X,q')$.  Since $F^+$ is unambiguous, a word
        $u\in\Lang{F^+}$ admits a unique factorization $u=u_1u_2\cdots u_n$ with
        $n\geq1$ and $u_i\in\Lang{F}$.  Now, $\TrMorph(u_i)=\sigma$ for all $1\leq
        i \leq n$ and since $p=r^p$ we deduce that
        $\tilde{u}^p=\tilde{u_1}^p\tilde{u_2}^p\cdots\tilde{u_n}^p$.  Since
        $x=(q,\leftleft,X,q')\in s^p$, the unique run $\rho$ of $\tilde{\A}$
        starting in state $q$ on the left of $\tilde{u_1}^p$ exits on the left in
        state $q'$.  Therefore, the unique run of $\tilde{\A}$ starting in state
        $q$ on the left of $\tilde{u}^p$ only visits $\tilde{u_1}^p$ and is
        actually $\rho$ itself.  Therefore, we set
        $C_E^p(x)=C_F^p(x)\lrdot(\Ifthenelse{F^*}{\varepsilon}\bot)$ and we can easily
        check that (\ref{oInv1}--\ref{oInv2}) are satisfied.
        
        \item Similarly for $x=(q,\rightright,X,q')$ 
        we set $C_E^p(x)=(\Ifthenelse{F^*}{\varepsilon}\bot)\lrdot C_F^p(x)$.
        
        \item If $x=(q,\leftright,X,q')$.  Since $\sigma$ is an idempotent, we
        have $x\in s^p\cdot s^p$.  We distinguish two cases depending on whether 
        the step $y\in s^p$ starting in state $q'$ from the left goes to the right or goes 
        back to the left.
        
        First, if $y=(q',\leftright,X_2,q_2)\in s^p$ goes to the right. Since $s^p$ 
        is an idempotent, following $x$ in $s^p\cdot s^p$ is same as following 
        $x$ in (the first) $s^p$ an then $y$ in (the second) $s^p$. Therefore, we 
        must have $q_2=q'$ and $X_2\subseteq X$. In this case, we set
        $C_E^p(x)=\Ifthenelse{F}{C_F^p(x)}{\big(C_F^p(x)\lrdot\lrplus{(C_F^p(y))}\big)}$.
        
        Second, if $y=(q',\leftleft,X_2,q_2)\in s^p$ goes to the left. Since $s^p$ 
        is an idempotent, there exists a unique sequence of steps in $s^p$:
        $x_1=x$,
        $x_2=y$,
        $x_3=(q_2,\rightright,X_3,q_3)$,
        $x_4=(q_3,\leftleft,X_4,q_4)$, \ldots,
        $x_i=(q_{i-1},\rightright,X_i,q_i)$,
        $x_{i+1}=(q_i,\leftright,X_{i+1},q')$ 
        with $i\geq3$ (see Figure~\ref{fig:CFp} middle). 
        We define
        \begin{align*}
          C_E^p(x) &= 
          \big(C_F^p(x)\lrdot(\Ifthenelse{F^*}{\varepsilon}\bot)\big)
          \odot\twoplus{F}{C'}
          \\
          C' &= 
          \big((\Ifthenelse{F}{\varepsilon}\bot)\lrdot C_F^p(x_2)\big)
          \odot(C_F^p(x_3)\lrdot C_F^p(x_4))\odot
          \cdots\odot(C_F^p(x_i)\lrdot C_F^p(x_{i+1})) 
        \end{align*}
        The proof of correctness, i.e., that (\ref{oInv1}--\ref{oInv2}) are 
        satisfied for $E$, is as in Section~\ref{sec:2DFT-RTE}.

        \item If $x=(q,\rightleft,X,q')$, the proof is obtained mutatis 
        mutandis, using the backward unambiguous (2-chained) Kleene-plus 
        $\rlplus{C}$ and $\rltwoplus{K}{C}$.
      \end{itemize}
      
      \medskip
      Now, we consider $p\in P$ with $r^p\neq p$.  We let $p'=r^p$.  We have
      already noticed that since $\sigma$ is idempotent we have $r^{p'}=p'$. 
      Consider a word $u\in\Lang{F^+}$.  Since $F^+$ is unambiguous, $u$ admits a
      unique factorization $u=u_1\cdots u_{n-1} u_n$ with $n\geq1$ and
      $u_i\in\Lang{F}$.  Now, $\TrMorph(u_i)=\sigma$ for all $1\leq i \leq n$.
      Using $r^p=p'$ and $r^{p'}=p'$ we deduce that
      $\tilde{u}^p=\tilde{u_1}^{p'}\cdots\tilde{u_{n-1}}^{p'}\tilde{u_n}^p$.  So
      when $n>1$, the expression $C_E^p$ that we need to compute is like the
      concatenation of $C_E^{p'}$ on the first $n-1$ factors with $C_F^p$ on the
      last factor. Since $r^{p'}=p'$ we have already seen how to compute 
      $C_E^{p'}$. We also know how to handle concatenation. So it should be no 
      surprise that we can compute $C_E^p$ when $p\neq r^p$.
      We define now formally $C_E^p(x)$ for $x\in s^p$.
      \begin{itemize}[nosep]
        \item If $x=(q,\leftleft,X,q')\in s^p$.  There are two cases depending on
        whether the step $y\in s^{p'}$ starting in state $q$ from the left goes
        back to the left or goes to the right.
        
        If it goes back to the left, then $y=(q,\leftleft,X,q')=x$ since
        $s^{p}=s^{p'}\cdot s^{p}$ (recall that $\sigma$ is idempotent) and we
        define
        $$
        C_E^p(x)=\Ifthenelse{F}{C_F^p(x)}{(C_F^{p'}(x)\lrdot(\Ifthenelse{F^+}{\varepsilon}\bot))} \,.
        $$
        If it goes to the right, then $y=(q,\leftright,X_1,q_1)$ and 
        there exists a unique sequence of steps:
        $x_1=y$,
        $x_2=(q_1,\leftleft,X_2,q_2)$,
        $x_3=(q_2,\rightright,X_3,q_3)$,
        $x_4=(q_3,\leftleft,X_4,q_4)$, \ldots,
        $x_i=(q_{i-1},\rightleft,X_i,q')$
        with $i\geq3$, $x_1,x_3,\ldots,x_i\in s^{p'}$ and $x_2,\ldots,x_{i-1}\in
        s^p$ (see Figure~\ref{fig:CFp} top right).  Notice that
        $X=X_1\cup\cdots\cup X_i$.  We define $C_E^p(x) =
        \Ifthenelse{F}{C_F^p(x)}{C'}$ where
        \begin{align*}
          C' &= 
          (C_E^{p'}(x_1)\lrdot C_F^p(x_2))\odot\cdots\odot
          (C_E^{p'}(x_{i-2})\lrdot C_F^p(x_{i-1}))\odot
          \big(C_E^{p'}(x_i)\lrdot (\Ifthenelse{F}{\varepsilon}\bot)\big) \,.
        \end{align*}
        We can check that (\ref{oInv1}--\ref{oInv2}) are 
        satisfied for $(E,p,x)$.

        \item If $x=(q,\rightleft,X,q')\in s^p$.  There are two cases depending on
        whether the step $y\in s^{p'}$ starting in state $q'$ from the right goes
        to the left or goes back to the right.
        
        If it goes to the left, then $y=(q',\rightleft,X',q')$ with 
        $X'\subseteq X$ and we define
        $$
        C_E^p(x)=\Ifthenelse{F}{C_F^p(x)}{(C_E^{p'}(y) \rldot C_F^{p}(x))} \,.
        $$
        If it goes back to the right, then $y=(q',\rightright,X_2,q_2)$ and 
        there exists a unique sequence of steps:
        $x_1=x$,
        $x_2=y$,
        $x_3=(q_2,\leftleft,X_3,q_3)$,
        $x_4=(q_3,\rightright,X_4,q_4)$, \ldots,
        $x_i=(q_{i-1},\leftleft,X_i,q_i)$
        $x_{i+1}=(q_i,\rightleft,X_{i+1},q')$
        with $i\geq3$, $x_1,x_3,\ldots,x_i\in s^{p}$ and $x_2,\ldots,x_{i+1}\in 
        s^{p'}$. Notice that $X_2\cup\cdots\cup X_{i+1}\subseteq X$.  We define
        $C_E^p(x) = \Ifthenelse{F}{C_F^p(x)}{C'}$ where
        \begin{align*}
          C' &= 
          (C_E^{p'}(x_2)\rldot C_F^p(x_1))\odot\cdots\odot
          (C_E^{p'}(x_{i-1})\rldot C_F^p(x_{i-2}))\odot
          (C_E^{p'}(x_{i+1})\rldot C_F^p(x_{i})) \,.
        \end{align*}
        We can check that (\ref{oInv1}--\ref{oInv2}) are 
        satisfied for $(E,p,x)$.
      
        \item The cases $x=(q,\leftright,X,q')\in s^p$ and 
        $x=(q,\rightright,X,q')\in s^p$ can be handled similarly.
        \qedhere
      \end{itemize}
  \end{description}
\end{proof}

\begin{figure}[tbp]
  \centering
  \includegraphics[scale=0.16]{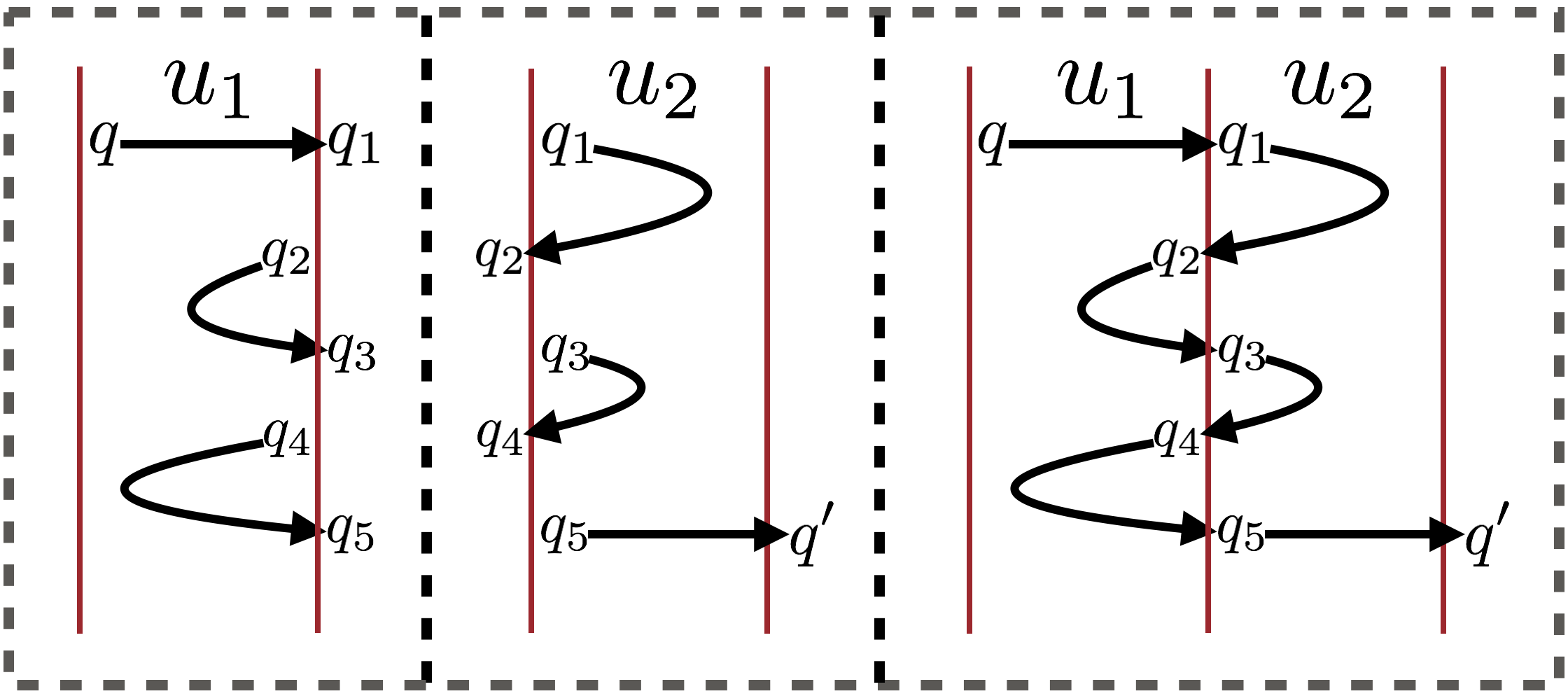}
  \hfill
  \includegraphics[scale=0.16]{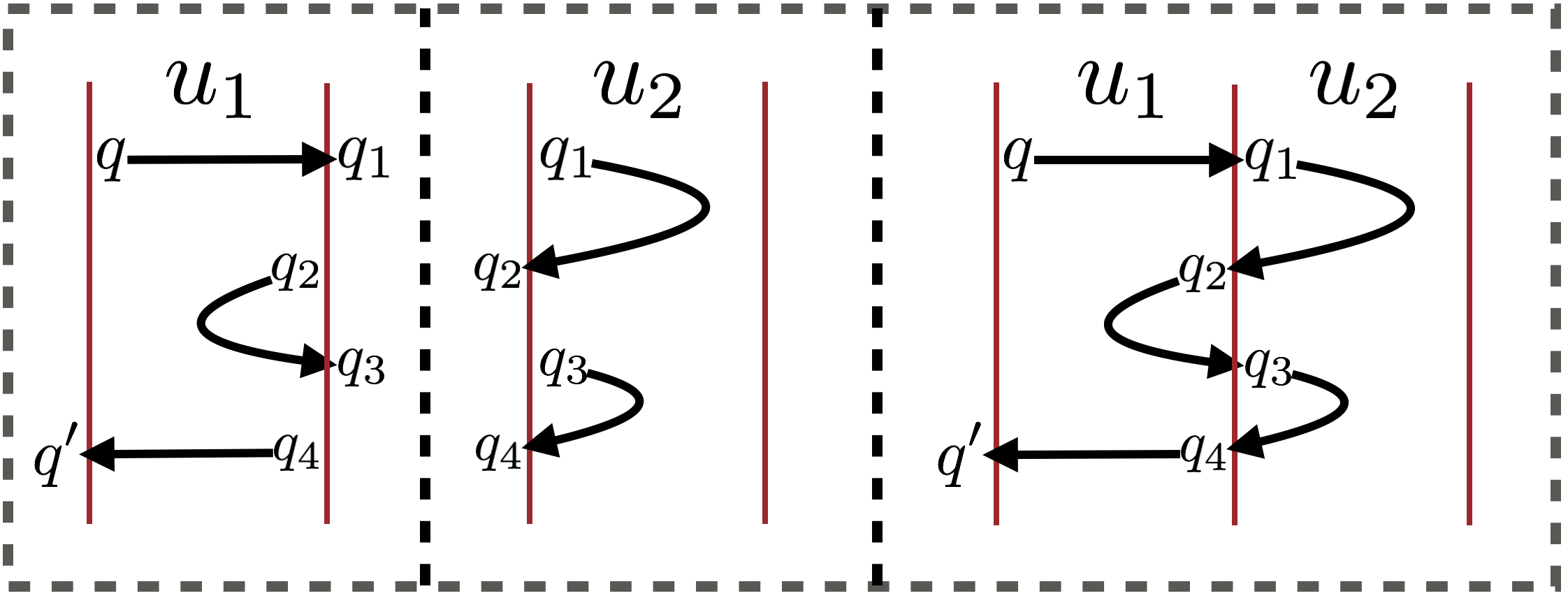}

  \medskip
  \includegraphics[scale=0.27]{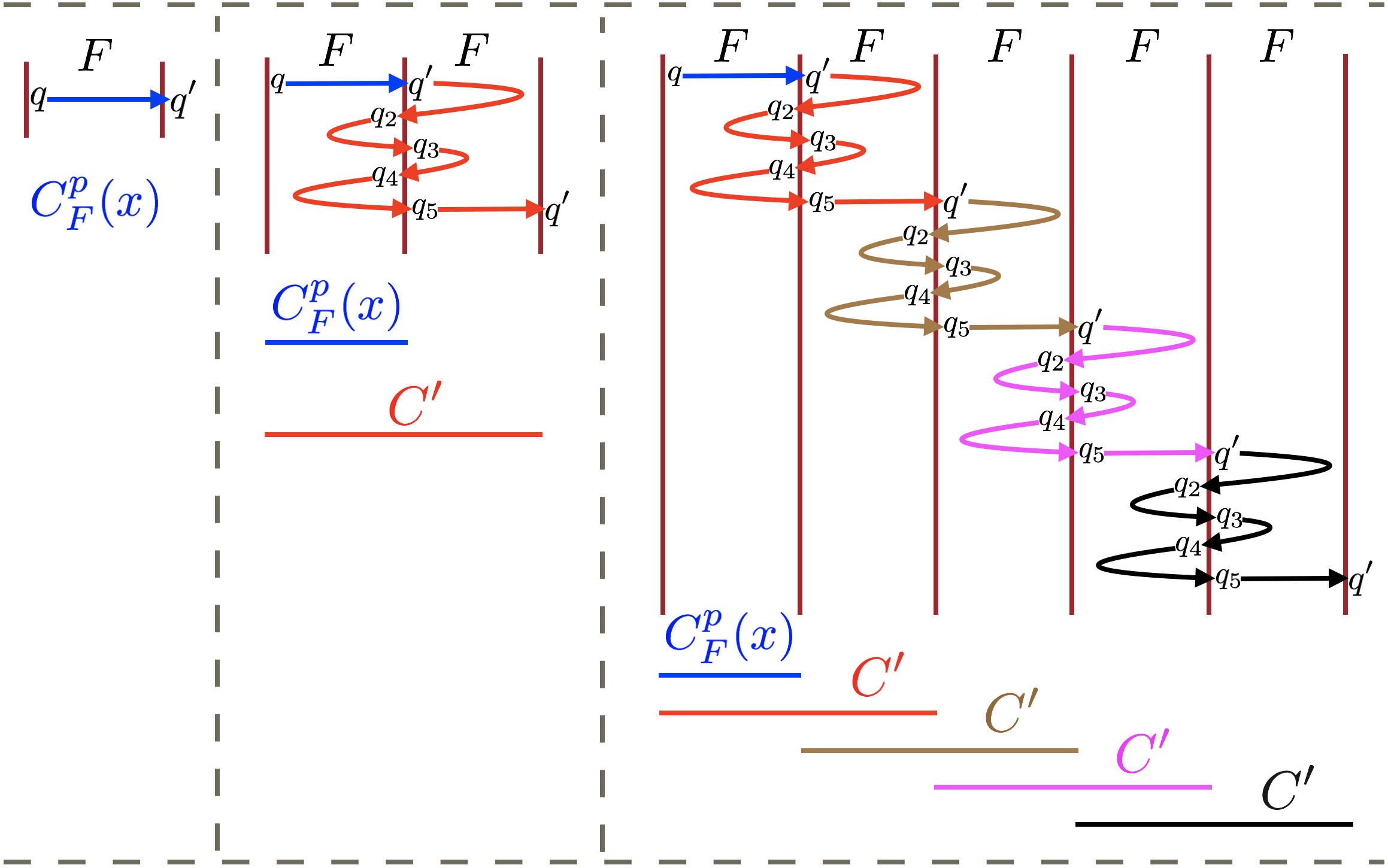}

  \medskip
  \includegraphics[scale=0.16]{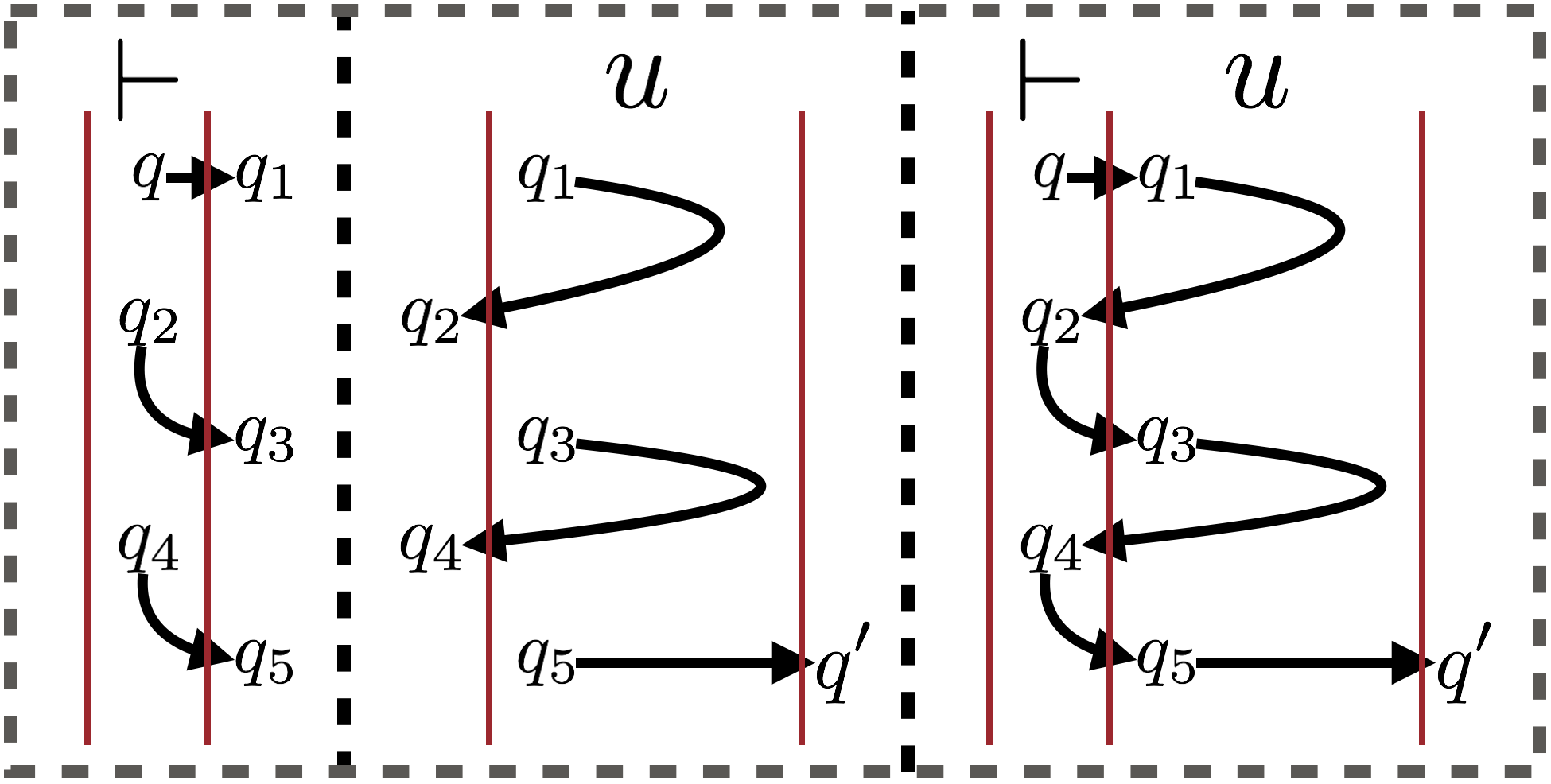}
  \hfill
  \includegraphics[scale=0.16]{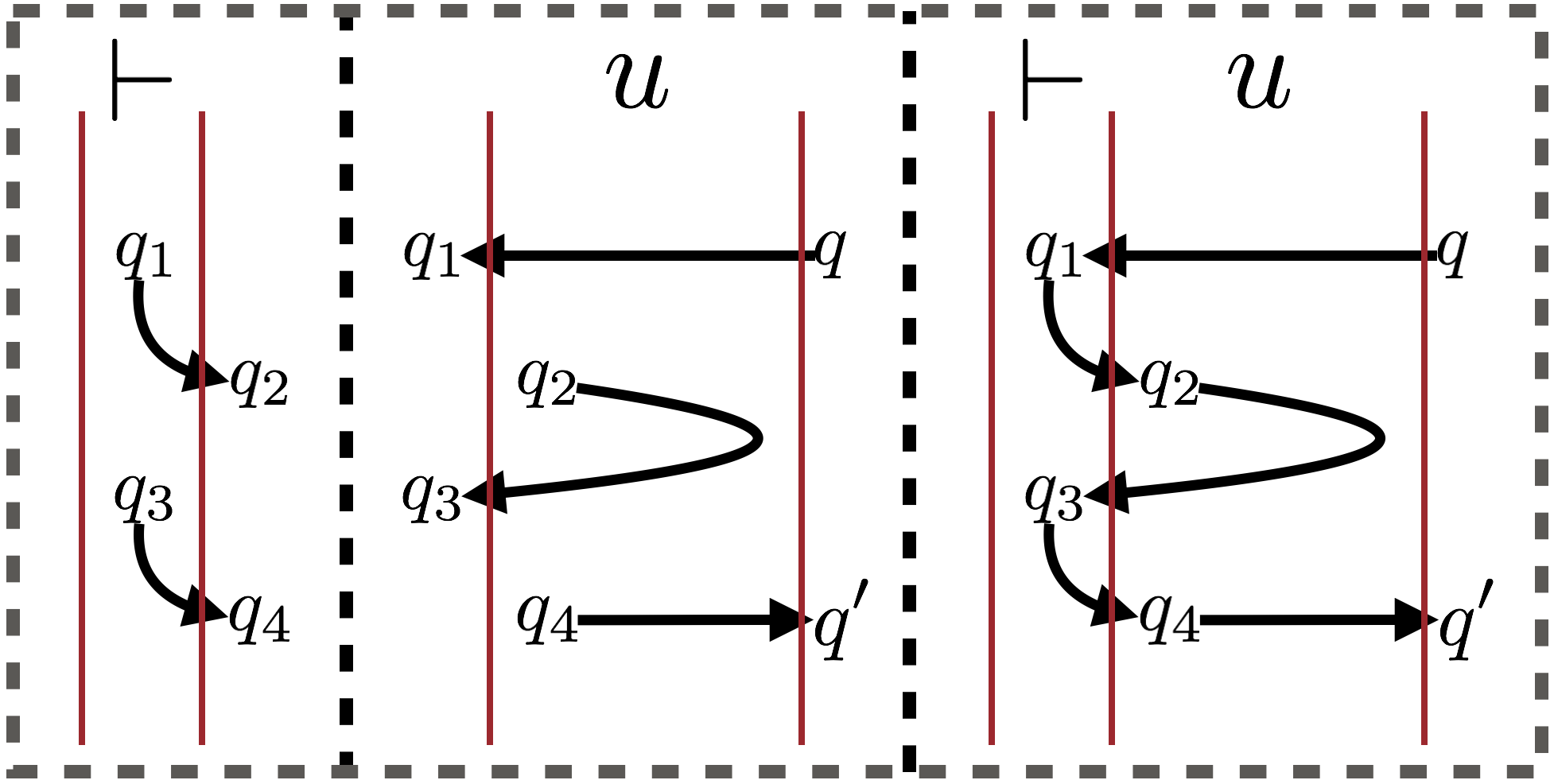}
  \caption{}
  \label{fig:CFp}
\end{figure}

We now define \RTEs corresponding to the left part of the computation of the
\twoDMTla $\A$, i.e., on some input $\leftend u$ consisting of the left
end-marker and some finite word $u\in\Sigma^+$.  As before, the look-ahead is
determined by the state of the \BDBA $\Bb$.

\begin{lemma}\label{lem:C_leftend_F^p}
  Let $F$ be an $\varepsilon$-free $\TrMorph$-good rational expression.
  For each state $p\in P$ and $q\in Q$, there is a unique state $q'\in Q$ and an
  \RTE $C_{\leftend F}^p(q,\leftright,q')$ (resp.\ $C_{\leftend
  F}^p(q,\rightright,q')$) such that the following invariants hold:
  \begin{enumerate}[nosep,label=($\mathsf{K}_{\arabic*}$),ref=$\mathsf{K}_{\arabic*}$]
    \item\label{kInv1} $\Lang{F}=\dom{C_{\leftend F}^p(q,\leftright,q')}$
    (resp.\ $\Lang{F}=\dom{C_{\leftend F}^p(q,\rightright,q')}$),

    \item\label{kInv2} for each $u\in\Lang{F}$, 
    $\sem{C_{\leftend F}^p(q,\leftright,q')}(u)$ (resp.\ 
    $\sem{C_{\leftend F}^p(q,\rightright,q')}(u)$) is the output produced by $\tilde{\A}$
    on $\tilde{\leftend u}^p$ when starting on the left (resp.\ right) in state $q$ until it
    exists on the right in state $q'$.
  \end{enumerate}
\end{lemma}

\begin{proof}
  Let $\sigma=(r^p,b^p,s^p)_{p\in P}=\TrMorph(F)$.  We fix some state $p\in P$.
  For all words $u\in\Lang{F}$, we have $\tilde{\leftend
  u}^p=(\leftend,r^p)\tilde{u}^p$.  Let $s_\leftend^p$ be the set of steps
  $(q,\leftright,q'),(q,\rightright,q')$ such that
  $\tilde{\delta}(q,(\leftend,r^p))=(q',\gamma_q^p,+1)$ in $\tilde{\A}$.
  
  For each $q\in Q$, there is a unique sequence of steps
  $x_1=(q,\leftright,q_1)$,
  $x_2=(q_1,\leftleft,X_2,q_2)$,
  $x_3=(q_2,\rightright,q_3)$,
  $x_4=(q_3,\leftleft,X_4,q_4)$, \ldots,
  $x_i=(q_{i-1},\rightright,q_i)$,
  $x_{i+1}=(q_i,\leftright,X_{i+1},q')$ 
  with $i\geq1$, 
  $x_1,x_3,\ldots,x_i\in s_\leftend^p$ and
  $x_2,x_4,\ldots,x_{i+1}\in s^p$ (see Figure~\ref{fig:CFp} bottom left). 
  We define
  $$
  C_{\leftend F}^p(q,\leftright,q')=\gamma_{q}^p\odot C_{F}^p(x_2)\odot
  \gamma_{q_2}^p\odot C_{F}^p(x_4)\odot
  \cdots\odot\gamma_{q_{i-1}}^p\odot C_{F}^p(x_{i+1}) \,.
  $$
  Using Lemma~\ref{lem:C_F^p}, we can show that $\Lang{F}=\dom{C_{\leftend
  F}^p(q,\leftright,q')}$ and also that for each $u\in\Lang{F}$, $\sem{C_{\leftend
  F}^p(q,\leftright,q')}(u)$ is the output produced by $\tilde{\A}$ on
  $\tilde{\leftend u}^p$ when starting on the left in state $q$ until it exists
  on the right in state $q'$.
  
  For each $q\in Q$, there is a unique sequence of steps
  $x_1=(q,\rightleft,X_1,q_1)$,
  $x_2=(q_1,\rightright,q_2)$,
  $x_3=(q_2,\leftleft,X_3,q_3)$,
  $x_4=(q_3,\rightright,q_4)$, \ldots,
  $x_i=(q_{i-1},\rightright,q_i)$,
  $x_{i+1}=(q_i,\leftright,X_{i+1},q')$ 
  with $i\geq2$, 
  $x_2,x_4,\ldots,x_i\in s_\leftend^p$ and
  $x_1,x_3,\ldots,x_{i+1}\in s^p$ (see Figure~\ref{fig:CFp} bottom right). 
  We define
  $$
  C_{\leftend F}^p(q,\rightright,q')=
  C_{F}^p(x_1)\odot \gamma_{q_1}^p\odot 
  C_{F}^p(x_3)\odot \gamma_{q_3}^p\odot 
  \cdots\odot\gamma_{q_{i-1}}^p\odot C_{F}^p(x_{i+1}) \,.
  $$
  Using Lemma~\ref{lem:C_F^p}, we can show that $\Lang{F}=\dom{C_{\leftend
  F}^p(q,\rightright,q')}$ and also that for each $u\in\Lang{F}$, $\sem{C_{\leftend
  F}^p(q,\rightright,q')}(u)$ is the output produced by $\tilde{\A}$ on
  $\tilde{\leftend u}^p$ when starting on the right in state $q$ until it exists
  on the right in state $q'$.
\end{proof}

\begin{lemma}\label{lem:C_E_F^omega}
  Let $F\cdot G^\omega$ be an unambiguous rational expression such that $F$ and
  $G$ are $\varepsilon$-free $\TrMorph$-good rational expresions and
  $\TrMorph(G)=\sigma=(r^p,b^p,s^p)_{p\in P}$ is an idempotent in the
  transition monoid $\TrMon$ of $(\A,\Bb)$.  We can construct an \oRTE
  $C_{FG^\omega}$ such that $\dom{C_{FG^\omega}}=\Lang{FG^\omega}\cap\dom{\A}$ and for each
  $w\in\dom{C_{FG^\omega}}$, $\sem{C_{FG^\omega}}(w)=\sem{\A}(w)$.
\end{lemma}

\begin{proof}
  We first show that there exists one and only one state $p\in P$ such that
  $r^p=p$ and $b^p=1$.  For the existence, consider a word
  $w=u_1u_2u_3\cdots\in\Lang{FG^\omega}$ with $u_1\in\Lang{F}$ and
  $u_n\in\Lang{G}$ for all $n\geq2$.  By definition of \BDBA there is a unique
  final run of $\Bb$ over $w$: $p_0,u_1,p_1,u_2,p_2,\ldots$.  Let us show first
  that $p_n=p_1$ for all $n\geq1$.  Since $\sigma$ is idempotent, we have
  $\TrMorph(u_{2}\cdots u_{n+1})=\TrMorph(u_{n+1})$.  Since
  $p_{1}\xleftarrow{u_{2}\cdots u_{n+1}}p_{n+1}$ and
  $p_{n}\xleftarrow{u_{n+1}}p_{n+1}$, we deduce that $p_{1}=r^{p_{n+1}}=p_{n}$.
  This implies $p_1=r^{p_2}=r^{p_1}$.
  Let $p=p_1$ so that $p=r^{p}$ and the final run of $\Bb$ on $w$ is
  $p_0,u_1,p,u_2,p,\ldots$.  Now, for all $n\geq2$ we have
  $\TrMorph(u_n)=\sigma$ and we deduce that $p\xleftarrow{u_n}p$ visits a final
  state from $\Fin$ iff $b^p=1$.  Since the run is accepting, we deduce that
  indeed $b^p=1$.  To prove the unicity, let $p\in P$ with $p=r^p$ and $b^p=1$.
  Let $v\in\Lang{G}$.  We have
  $p\xleftarrow{v}p$ and this subrun visits a final state from $\Fin$.
  Therefore, $p,v,p,v,p,v,p,\ldots$ is a final run of $\Bb$ on $v^\omega$.
  Since $\Bb$ is \BDBA, there is a unique final run of $\Bb$ on $v^\omega$, 
  which proves the unicity of $p$.
  
  We apply Lemma~\ref{lem:C_leftend_F^p}. We denote by $s'_{\leftend F}$ the 
  set of triples $(q,d,q')\in Q\times\{\leftright,\rightright\}\times Q$ such 
  that the \RTE $C_{\leftend F}^p(q,d,q')$ is defined.
  
  Starting from the initial state $q_0$ of $\A$, there exists a unique sequence 
  of steps 
  $x'_1=(q_0,\leftright,q'_1)$,
  $x'_2=(q'_1,\leftleft,X'_2,q'_2)$,
  $x'_3=(q'_2,\rightright,q'_3)$,
  $x'_4=(q'_3,\leftleft,X'_4,q'_4)$, \ldots,
  $x'_i=(q'_{i-1},\rightright,q'_i)$,
  $x'_{i+1}=(q'_i,\leftright,X'_{i+1},q)$ 
  with $i\geq1$, 
  $x'_1,x'_3,\ldots,x'_i\in s'_{\leftend F}$ and
  $x'_2,x'_4,\ldots,x'_{i+1}\in s^p$. 
  We define
  \begin{align*}
    C_1 &= \big( C_{\leftend F}^p(x'_1)\lrdot C_{G}^p(x'_2) \big) \odot
    \big( C_{\leftend F}^p(x'_3)\lrdot C_{G}^p(x'_4) \big) \odot\cdots\odot
    \big( C_{\leftend F}^p(x'_i)\lrdot C_{G}^p(x'_{i+1}) \big) 
    \\
    C_2 &= C_1 \lrdot (\Ifthenelse{G^\omega}{\varepsilon}\bot) \,.
  \end{align*}
  We have $\dom{C_1}=FG$ and $\tilde{\leftend u_1u_2}^p=\tilde{\leftend
  u_1}^p\tilde{u_2}^p$ for all $u_1\in F$ and $u_2\in G$.  Moreover,
  $\sem{C_1}(u_1u_2)$ is the output produced by $\tilde{\A}$ on $\tilde{\leftend
  u_1u_2}^p$ when starting on the left in the initial state $q_0$ until it exists on
  the right in state $q$.
  Now, $C_2$ is an \oRTE with $\dom{C_2}=FG^\omega$ and for all 
  $w=u_1u_2u_3\ldots\in FG^\omega$ with $u_1\in F$ and $u_n\in G$ for all 
  $n>1$, we have $\sem{C_2}(w)=\sem{C_1}(u_1u_2)\in\Gamma^*$.
  
  Now, we distinguish two cases. First, assume that there is a step 
  $x=(q,\leftright,X,q')\in s^p$. Since $\sigma$ is idempotent, so is $s^p$, and 
  since $x'_{i+1}=(q'_i,\leftright,X'_{i+1},q)\in s^p$ we deduce that $q'=q$.  
  Therefore, the unique run of $\tilde{\A}$ on $\tilde{\leftend 
  w}=\tilde{\leftend u_1}^p\tilde{u_2}^p\tilde{u_3}^p\cdots$ follows the steps 
  $x'_1 x'_2 \cdots x'_i x'_{i+1} x x x \cdots$. Hence, the set of states visited 
  infinitely often along this run is $X$ and the run is accepting iff 
  $X\in\mathcal{F}$ is a Muller set. Therefore, if $X\notin\mathcal{F}$ we have 
  $FG^\omega\cap\dom{\A}=\emptyset$ and we set $C_{FG^\omega}=\bot$. Now,  
  if $X\in\mathcal{F}$ we have $FG^\omega\subseteq\dom{\A}$ and we set
  $$
  C_{FG^\omega}= C_2 \odot \big( (\Ifthenelse{FG}{\varepsilon}\bot)
  \lrdot C_G^p(x)^\omega \big) \,.
  $$
  We have $\dom{C_{FG^\omega}}=FG^\omega$ and for all 
  $w=u_1u_2u_3\ldots\in FG^\omega$ with $u_1\in F$ and $u_n\in G$ for all 
  $n>1$, we have 
  $$
  \sem{C_{FG^\omega}}(w)=\sem{C_1}(u_1u_2)\sem{C_G^p(x)}(u_3)\sem{C_G^p(x)}(u_4)\cdots \,.
  $$
  By \eqref{oInv2}, we know that for all $n\geq3$, $\sem{C_G^p(x)}(u_n)$  is the output
  produced by $\tilde{\A}$ when running step $x=(q,\leftright,X,q)$ on $\tilde{u_n}^p$.
  We deduce that $\sem{C_{FG^\omega}}(w)=\sem{\tilde{\A}}(\tilde{\leftend 
  w})=\sem{\A}(w)$ as desired.
  
  The second case is when the unique step $x_1=(q,\leftleft,X_1,q_1)$ in $s^p$
  which starts from the left in state $q$ exits on the left. Since $s^p$ is 
  idempotent and $x'_{i+1}=(q'_i,\leftright,X'_{i+1},q)\in s^p$, by definition 
  of the product $s^p\cdot s^p$, there is a unique sequence of steps
  $x_2=(q_1,\rightright,X_2,q_2)$,
  $x_3=(q_2,\leftleft,X_3,q_3)$,
  \ldots,
  $x_j=(q_{j-1},\rightright,X_j,q_j)$,
  $x_{j+1}=(q_j,\leftright,X_{j+1},q)$ 
  in $s^p$ with $j\geq2$.
  Therefore, for all $w=u_1u_2u_3\ldots\in FG^\omega$ with $u_1\in F$ and
  $u_n\in G$ for all $n>1$, the unique run of $\tilde{\A}$ on $\tilde{\leftend
  w}=\tilde{\leftend u_1}^p\tilde{u_2}^p\tilde{u_3}^p\cdots$ follows the steps
  $x'_1 x'_2\cdots x'_i x'_{i+1} (x_1 x_2 x_3 \cdots x_j x_{j+1})^\omega$.
  Hence, the set of states visited infinitely often along this run is $X=X_1\cup
  X_2\cup\cdots\cup X_{j+1}$.  We deduce that the run is accepting iff 
  $X\in\mathcal{F}$ is a Muller set. Therefore, if $X\notin\mathcal{F}$ we have 
  $FG^\omega\cap\dom{\A}=\emptyset$ and we set $C_{FG^\omega}=\bot$. Now,  
  if $X\in\mathcal{F}$ we have $FG^\omega\subseteq\dom{\A}$ and we set
  \begin{align*}
    C_3 &= \big( (\Ifthenelse{G}{\varepsilon}\bot)\lrdot C_{G}^p(x_1) \big) \odot
    \big( C_{G}^p(x_2)\lrdot C_{G}^p(x_3) \big) \odot\cdots\odot
    \big( C_{G}^p(x_j)\lrdot C_{G}^p(x_{j+1}) \big) 
    \\
    C_{FG^\omega} &= C_2 \odot \big( (\Ifthenelse{F}{\varepsilon}\bot)
    \lrdot \twoomega{G}{C_3} \big) \,.
  \end{align*}
  We have $\dom{C_{FG^\omega}}=FG^\omega$ and for all 
  $w=u_1u_2u_3\ldots\in FG^\omega$ with $u_1\in F$ and $u_n\in G$ for all 
  $n>1$, we have 
  $$
  \sem{C_{FG^\omega}}(w)=\sem{C_1}(u_1u_2)\sem{C_3}(u_2u_3)\sem{C_3}(u_3u_4)\cdots \,.
  $$
  Using \eqref{oInv2}, we can check that this is the output produced by
  $\tilde{\A}$ when running on $\tilde{\leftend w}$.  We deduce that
  $\sem{C_{FG^\omega}}(w)=\sem{\tilde{\A}}(\tilde{\leftend w})=\sem{\A}(w)$ as
  desired.
\end{proof}

We are now ready to prove that \twoDMTla are no more expressive than \oRTEs.

\begin{proof}[Proof of Theorem~\ref{thm:2WST=oRTE}~\eqref{item:2WSTto-oRTE}]
  We use the notations of the previous sections, in particular for the \twoDMTla
  $\A$, the \BDBA $\Bb$.  We apply Theorem~\ref{thm:U-forest-omega} to the
  canonical morphism $\TrMorph$ from $\Sigma^*$ to the transition monoid
  $\TrMon$ of $(\A,\Bb)$.  We obtain an \emph{unambiguous} rational expression
  $G=\bigcup_{k=1}^{m}F_k\cdot G_k^\omega$ over $\Sigma$ such that
  $\Lang{G}=\Sigma^\omega$ and for all $1\leq k\leq m$ the expressions $F_k$ and
  $G_k$ are $\varepsilon$-free $\TrMorph$-good rational expressions and
  $\sigma_{G_k}$ is an idempotent, where $\TrMorph(G_k)=\{\sigma_{G_k}\}$.  For
  each $1\leq k\leq m$, let $C_k=C_{F_kG_k^\omega}$ be the \oRTE given by
  Lemma~\ref{lem:C_E_F^omega}.
  We define the final \oRTE as
  $$
  C = 
  \Ifthenelse{F_1G_1^\omega}{C_{1}}{
  (\Ifthenelse{F_2G_2^\omega}{C_{2}}{\cdots
  (\Ifthenelse{F_{m-1}G_{m-1}^\omega}{C_{{m-1}}}{C_{m}})
  })} \,.
  $$
  From Lemma~\ref{lem:C_E_F^omega}, we can easily check that $\dom{C}=\dom{\A}$ 
  and $\sem{C}(w)=\sem{\A}(w)$ for all $w\in\dom{C}$.
\end{proof}

\section{Conclusion}

The main contribution of the paper is to give a characterisation of regular
string transductions using some combinators, giving rise to regular transducer
expressions (RTE).  Our proof uniformly works well for finite and infinite
string transformations.  RTE are a succint specification mechanism for regular
transformations just like regular expressions are for regular languages.  It is
worthwhile to consider extensions of our technique to regular tree
transformations, or in other settings where more involved primitives such as
sorting or counting are needed.  The minimality of our combinators in achieving
expressive completeness, as well as computing complexity measures for the
conversion between RTEs and two-way transducers are open.

\bibliographystyle{plain}
\bibliography{papers}

\newpage
\appendix

\newpage
\appendix
\onecolumn

\newcommand{\err}{\ensuremath{\mathsf{err}}\xspace}
\newcommand{\nl}{\ensuremath{\backslash\mathsf{n}}\xspace}

\newcommand{\twentyoneplus}[2]{[#1,#2]^{21\scriptstyle{\boxplus}}}
\newcommand{\twentyoneomega}[2]{[#1,#2]^{21\omega}}
\newcommand{\komega}[2]{[#1,#2]^{k\omega}}
\newcommand{\kplus}[2]{[#1,#2]^{k\scriptstyle{\boxplus}}}
\newcommand{\rplus}[2]{[#1,#2]^{r\scriptstyle{\boxplus}}}

\section{Examples}\label{app:intro}

\subsection{More details on the Example in the Introduction}
\label{app:intro}

We continue with the computation of the \RTE $C_{E_3}(q_0,\leftright,q_2)$ for
$E_3=[(ba^{+})^{3}]^{+}b \subseteq \dom{\Aa}$. 
This involves the use of the 2-chained Kleene-plus.

\begin{figure*}[h]
  \centerline{\includegraphics[scale=0.3, angle=270]{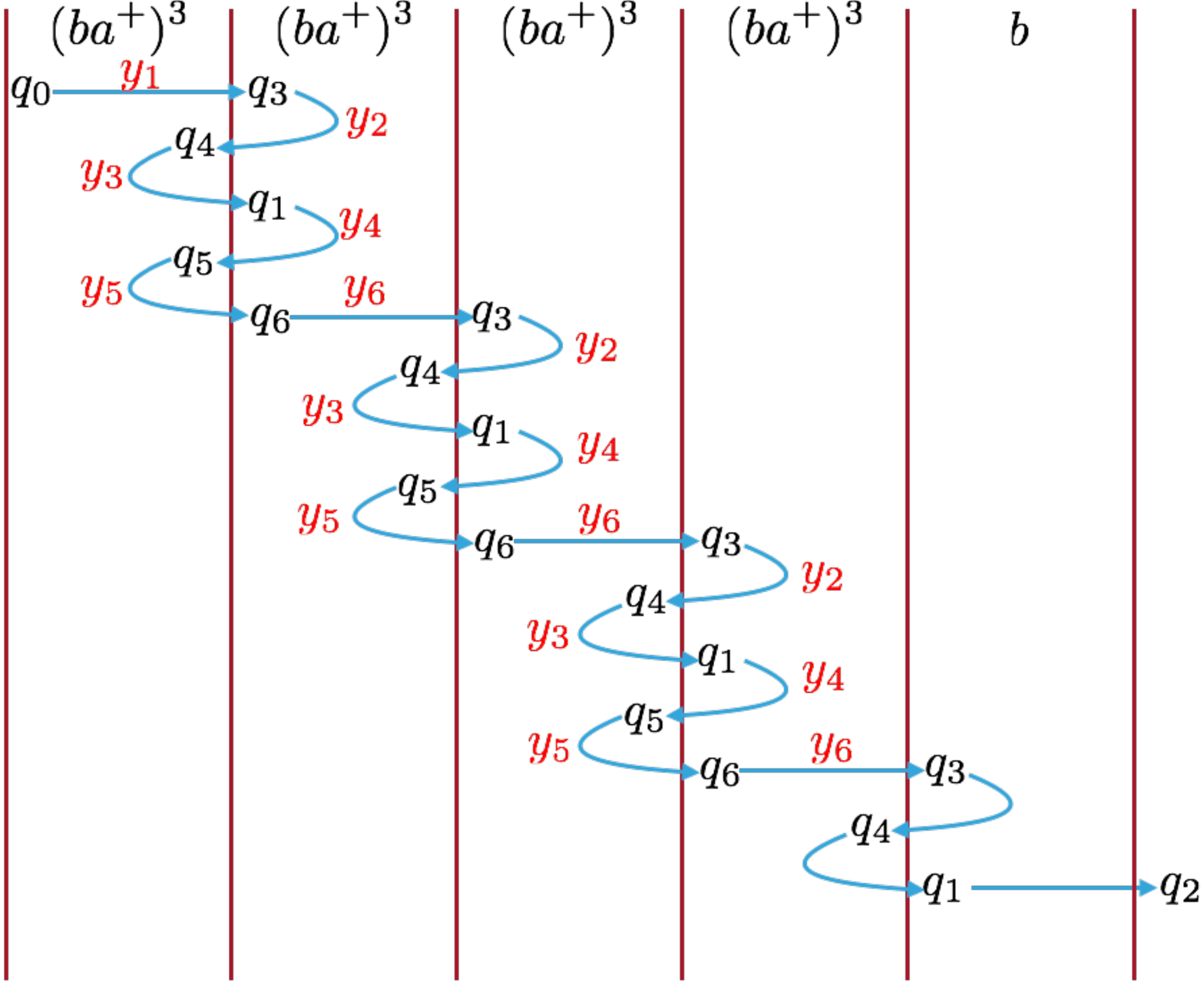}}
  \caption{run of a word in $E_3=[(ba^{+})^{3}]^{+}b$ }
  \label{fig:runapp}
\end{figure*}

We want to compute the \RTE for the step $(q_0,\leftright,q_2)$ on a word $u\in
E_3$.  It can be decomposed as shown in Figure~\ref{fig:runapp}.  Unlike the
case of $E_2$, we have to use the 2-chained Kleene plus.
Let $F=(ba^{+})^{3}$ so that $E_3=F^{+}b$. 
We have (see Figure~\ref{fig:runapp}), 
$$
C_{E_3}(q_0,\leftright,q_2)=
(C_{F^{+}}(q_0,\leftright,q_3)\lrdot C_b(q_3,\leftleft,q_4)) \odot
(C_{F^{+}}(q_4,\rightright,q_1)\lrdot C_b(q_1,\leftright,q_2)) \,.
$$
We know that $C_b(q_3,\leftleft,q_4)=(b/\varepsilon)=C_b(q_1,\leftright,q_2)$ hence 
it remains to compute $C_{F^{+}}(q_0,\leftright,q_3)$ and 
$C_{F^{+}}(q_4,\rightright,q_1)$. 
First we define \RTEs associated with atomic expressions and steps which are
going to be used in constructing $C_{E_3}(q_0,\leftright,q_2)$.  They are
$C_b(q_3, \leftleft, q_4) = C_b(q_6, \leftright, q_1) = C_b(q_1, \leftright,
q_2) = C_b(q_5, \leftright, q_6) = (b/\varepsilon)$ and $C_{a^+}(q_2,
\leftright, q_3) = C_{a^+}(q_1, \leftright, q_1) = \lrplus{(a/\varepsilon)}$,
$C_{a^+}(q_4, \rightleft, q_4) = \rlplus{(a/a)}$, $C_{a^+}(q_5, \rightleft, q_5)
= \rlplus{(a/b)}$.  We compute \RTE $C_{F}(x)$ for the relevant steps $x$
in the monoid element $X = \TrMorph(F)$.  $F$
is an unambiguous catenation of $E_2 = ba^+ba^+b$ with $a^+$ and from 
Figure~\ref{fig:run}, it can be seen that:
\begin{enumerate}
  \item For $y_1 = (q_0, \leftright, q_3)$
  \begin{align*}
    C_{F}(y_1)& = C_{E_2}(q_0, \leftright, q_2) \lrdot C_{a^+}(q_2, \leftright, q_3) \\
    & = \big( (b/\varepsilon)\rldot\rlplus{(a/b)}\rldot
    (b/\varepsilon)\rldot\rlplus{(a/a)} \big) \lrdot (b/\varepsilon) \lrdot \lrplus{(a/\varepsilon)}
  \end{align*}
  where $C_{E_2}(q_0, \leftright, q_2)$ has been computed in 
  Section~\ref{sec:intro}. \\
  For example,  $\sem{C_{F}(y_1)}(ba^{m_1}ba^{m_2}ba^{m_3}) = a^{m_2}b^{m_1}$. 
  
  \item  Continuing with the computation for $(ba^+)^3$ as in Figure~\ref{fig:runapp}, for $y_2 = (q_3, \leftleft, q_4)$,
  we take the Cauchy product of $C_b(q_3, \leftleft, q_4)$ with $(\Ifthenelse{a^+ba^+ba^+}{\varepsilon}{\bot})$. 
  \begin{align*}
    C_{F}(y_2) & = C_b(q_3, \leftleft, q_4) \lrdot 
    (\Ifthenelse{a^+ba^+ba^+}{\varepsilon}{\bot}) 
    \approx (\Ifthenelse{(ba^+)^{3}}{\varepsilon}{\bot}) 
  \end{align*}
  $\sem{C_{F}(y_2)}(ba^{m_1}ba^{m_2}ba^{m_3}) = \varepsilon$.
  
  \item   For $y_3 = (q_4, \rightright, q_1)$, we have
  \begin{align*}
    C_{F}(y_3)& = (\Ifthenelse{ba^+}{\varepsilon}{\bot}) \rldot 
    C_{ba^+ba^+}(y_3) \\
    &= (\Ifthenelse{ba^+}{\varepsilon}{\bot}) \rldot \big( (b/\varepsilon)\rldot\rlplus{(a/b)}\rldot
    (b/\varepsilon)\rldot\rlplus{(a/a)} \big) \\
    &\approx (\Ifthenelse{ba^+b}{\varepsilon}{\bot})\lrdot \big(\rlplus{(a/b)}\rldot
    (b/\varepsilon)\rldot\rlplus{(a/a)} \big) 
  \end{align*}
  where $C_{ba^+ba^+}(y_3)$  is already computed in Section~\ref{sec:intro}. \\
  As an example,   $\sem{C_{F}(y_3)}(ba^{m_1}ba^{m_2}ba^{m_3}) = a^{m_3}b^{m_2}$. 

  \item For $y_4 = (q_1, \leftleft, q_5)$, it is similar to the $C_E(y)$
  computed for $C_{E_2}$ in Section~\ref{sec:intro}.  Here we have
  \begin{align*}
    C_{F}(y_4)& = C_{ba^+b}(y_4) \lrdot (\Ifthenelse{a^+ba^+}{\varepsilon}{\bot}) \\
    & = \big( (C_b(q_1, \leftright, q_2) \lrdot C_{a^+}(q_2, \leftright, q_3) \lrdot (b/\varepsilon)) \odot{} \\
    &~~~~~~ (C_b(q_4, \rightleft, q_5) \rldot C_{a^+}(q_4, \rightleft, q_4) 
    \rldot C_b(q_3, \leftleft, q_4)) \big) 
    \lrdot (\Ifthenelse{a^+ba^+}{\varepsilon}{\bot}) \\
    & = (((b/\varepsilon) \lrdot \lrplus{(a/\varepsilon)} \lrdot (b/\varepsilon)) \odot 
    ((b/\varepsilon) \rldot \rlplus{(a/a)} \rldot (b/\varepsilon))) \lrdot (\Ifthenelse{a^+ba^+}{\varepsilon}{\bot}) \\
    & \approx \big( (b/\varepsilon) \rldot \rlplus{(a/a)} \rldot 
    (b/\varepsilon) \big) \lrdot (\Ifthenelse{a^+ba^+}{\varepsilon}{\bot})
  \end{align*}
  As an example, $\sem{C_{F}(y_4)}(ba^{m_1}ba^{m_2}ba^{m_3}) = a^{m_1}$. 

  \item For $y_5 = (q_5, \rightright, q_6)$, in the computation of $C_F(y_5)$ we need $C_{ba^+}(y_5)$. Thus, we compute $C_{ba^+}(y_5)$ below whose computation is similar to $C_E(y)$ computed in Section~\ref{sec:intro}.
  \begin{align*}
    C_{ba^+}(y_5)& = ((b/\varepsilon) \rldot C_{a^+}(q_5, \rightleft, q_5)) \odot (C_b(q_5, \leftright, q_6) \lrdot C_{a^+}(q_6, \leftright, q_6)) \\
    & = ((b/\varepsilon) \rldot \rlplus{(a/b)}) \odot ((b/\varepsilon) \lrdot 
    \lrplus{(a/\varepsilon)}) 
    \approx (b/\varepsilon) \rldot \rlplus{(a/b)}
  \end{align*}
  We can compute $C_{F}(y_5)$ as 
  \begin{align*}
    C_{F}(y_5)& = (\Ifthenelse{ba^+ba^+}{\varepsilon}{\bot}) \lrdot C_{ba^+}(y_5) 
    \approx (\Ifthenelse{ba^+ba^+b}{\varepsilon}{\bot}) \rldot \rlplus{(a/b)}
  \end{align*}
  As an example, $\sem{C_{F}(y_5)}(ba^{m_1}ba^{m_2}ba^{m_3}) = b^{m_3}$. 
  
  \item For $y_6 = (q_6, \leftright, q_3)$, the computation of $C_{F}(y_6)$ is
  similar to $C_{ba^+ba^+}(q_0, \leftright, q_2)$ in Section~\ref{sec:intro}.
  We need $C_{ba^+ba^+}(q_6, \leftright, q_3)$ and $C_{ba^+ba^+b}(q_6,
  \leftright, q_2)$.  We see the computation of these below step by step.
  \begin{align*}
    C_{ba^+ba^+}(q_6, \leftright, q_3)& = C_b(q_6, \leftright, q_1) \lrdot C_{a^+}(q_1, \leftright, q_1) \lrdot C_b(q_1, \leftright, q_2) \lrdot C_{a^+}(q_2, \leftright, q_3) \\
    & = (b/\varepsilon) \lrdot \lrplus{(a/\varepsilon)} \lrdot (b/\varepsilon) \lrdot \lrplus{(a/\varepsilon)} 
    \approx (\Ifthenelse{ba^+ba^+}{\varepsilon}{\bot})
    \\
    C_{ba^+ba^+b}(q_6, \leftright, q_2)& = (C_{ba^+ba^+}(q_6, \leftright, q_3) \lrdot C_b(q_3, \leftleft, q_4)) \odot (C_{ba^+ba^+}(q_4, \rightright, q_1) \lrdot C_b(q_1, \leftright, q_2)) \\
    & \approx (\Ifthenelse{ba^+ba^+b}{\varepsilon}{\bot}) \odot (((b/\varepsilon)\rldot\rlplus{(a/b)}\rldot
    (b/\varepsilon)\rldot\rlplus{(a/a)}) \lrdot (b/\varepsilon)) \\
    & \approx \big( (b/\varepsilon)\rldot\rlplus{(a/b)}\rldot
    (b/\varepsilon)\rldot\rlplus{(a/a)} \big) \lrdot (b/\varepsilon)
  \end{align*}
  Note that $C_{ba^+ba^+}(q_4, \rightright, q_1)$ has been computed in Section~\ref{sec:intro}.  
  Now we concatenate with $C_{a^+}(q_2, \leftright, q_3)$ needed in the computation. 
  \begin{align*}
    C_{F}(y_6)& = C_{ba^+ba^+b}(q_6, \leftright, q_2) \lrdot C_{a^+}(q_2, \leftright, q_3) \\
    & = \big( (b/\varepsilon)\rldot\rlplus{(a/b)}\rldot
    (b/\varepsilon)\rldot\rlplus{(a/a)} \big) \lrdot (b/\varepsilon) \lrdot (\lrplus{a/\varepsilon})\\
    & \approx \big( (b/\varepsilon)\rldot\rlplus{(a/b)}\rldot 
    (b/\varepsilon)\rldot\rlplus{(a/a)} \big) \lrdot
    (\Ifthenelse{ba^+}{\varepsilon}{\bot})
  \end{align*}
  As an example, $\sem{C_{F}(y_6)}(ba^{m_1}ba^{m_2}ba^{m_3}) = a^{m_2}b^{m_1}$.
\end{enumerate}
Now we are in a position to compute \RTE $C_{F^{+}}(q_0, \leftright, q_3)$.  As
shown in figure\ref{fig:runapp}, it is a concatenation of step $y_1$ and then
steps $y_2$, $y_3$, $y_4$, $y_5$ and $y_6$ repetitively.  Consecutive pairs of
$(ba^+)^3$ are needed to compute the \RTE and thanks to the 2-chained Kleene
plus, we can define the \RTE for the same.
\begin{align*}
  C_{F^{+}}(y_1) & = (C_{F}(y_1) \lrdot (\Ifthenelse{F^*}{\varepsilon}{\bot})) \odot \twoplus{F}{C'} \\
  C' & =  ((\Ifthenelse{F}{\varepsilon}{\bot}) \lrdot C_{F}(y_2) ) \odot (C_{F}(y_3) \lrdot C_{F}(y_4)) \odot (C_{F}(y_5) \lrdot C_{F}(y_6))
\end{align*}
As an example,
$\sem{C_{F^{+}}(y_1)}(ba^{m_1}ba^{m_2}ba^{m_3}ba^{m_4}ba^{m_5}ba^{m_6}) =
a^{m_2}b^{m_1}a^{m_3}b^{m_2}a^{m_4}b^{m_3}a^{m_5}b^{m_4}$.

Finally, we compute \RTE for $y = (q_0, \leftright, q_2)$ for the expression $E_3
= [(ba^{+})^{3}]^{+}b$ by concatenating $b$ with the above \RTE.
\begin{align*}
  C_{E_3}(y)& = (C_{F^{+}}(q_0, \leftright, q_3) \lrdot C_b(q_3, \leftleft, q_4)) \odot (C_{F^{+}}(q_4, \rightright, q_1) \lrdot C_{b}(q_1, \leftright, q_2))
\end{align*}
Notice that 
$C_{F^{+}}(q_4,\rightright,q_1)=(\Ifthenelse{F^*}{\varepsilon}{\bot})\lrdot 
C_F(y_3)$.

We have already seen that $C_{E_3}(y)$ computes the output produced by a
successful run on a word $w \in E_3$.  Applying the \RTE as above, we have, for
example, 
$$
\sem{C_{E_3}(y)}(ba^{m_1}ba^{m_2}ba^{m_3}ba^{m_4}ba^{m_5}ba^{m_6}b) =
a^{m_2}b^{m_1}a^{m_3}b^{m_2}a^{m_4}b^{m_3}a^{m_5}b^{m_4}a^{m_6}b^{m_5} \,.
$$

\subsection{A Motivating Example}
\label{app:motiv}
Apart from theoretical interest, regular
expressions have great practical utility, being used in search engines, or in
search and replace patterns in text processors, or in lexical analysis in
compilers.  Many programming languages like Java and Python also support regular
expressions using a regexp engine, as part of their standard libraries.  We
believe that our extension of the beautiful theory of regular expressions to
regular transducer expressions over both finite and infinite words has many
useful applications.

As a specific example, we consider the $\mathsf{dmesg}$ command (see
Figure~\ref{fig:dmesg} for a sample output) used to write the kernel messages in
Linux and other Unix-like operating systems to standard output (which by default
is the display screen).  The output is often captured in a permanent system
logfile via a logging daemon, such as syslog.  The kernel is the first part of
the operating system that is loaded into memory when a computer boots up.  The
numerous messages generated by the kernel that appear on the display screen as a
computer boots up show the hardware devices that the kernel detects and indicate
whether it is able to configure them.  $\mathsf{dmesg}$ obtains its data by
reading the kernel ring buffer, a portion of a computer's memory that is set
aside as a temporary holding place for data that is being sent to or received
from an external device, such as a hard disk drive (HDD), printer or keyboard.
Using $\mathsf{dmesg}$ along with the $\mathsf{-w}$ option gives real time
updates while the $\mathsf{-xH}$ option provides extra information with each
line of the kernel message relating to various (external) devices.  This
information can be one of $\mathsf{err}$ (for error), $\mathsf{emerg}$ (for
emergency), $\mathsf{warn}$ (for warning), $\mathsf{info}$ (for information) and
so on.  $\mathsf{dmesg}$ can thus be very useful when troubleshooting or just
trying to obtain information about the hardware on a system by analyzing this
output.  
We can extract from the output produced by $\mathsf{dmesg}$ some messages 
with contextual information.
For instance, if we are searching for $\mathsf{err}$ messages, and wish to
resolve them, we need some contextual information, like 10 lines before and 10
lines after the message.  This is a regular transformation which can be
specified with an \oRTE and implemented with a two-way transducer as described
below.  It takes as input the (unbounded) log produced by $\mathsf{dmesg -wHx}$
and produces as output lines containing $\mathsf{err}$ messages with their
contexts.

\subsubsection{Detecting context of error in $\mathsf{dmesg}$}
In this section, we give details of how the errors in $\mathsf{dmesg}$ command 
are detected and their context is given as output using \oRTEs or transducers. 

\begin{figure*}[h]
  \centerline{\includegraphics[scale=0.5]{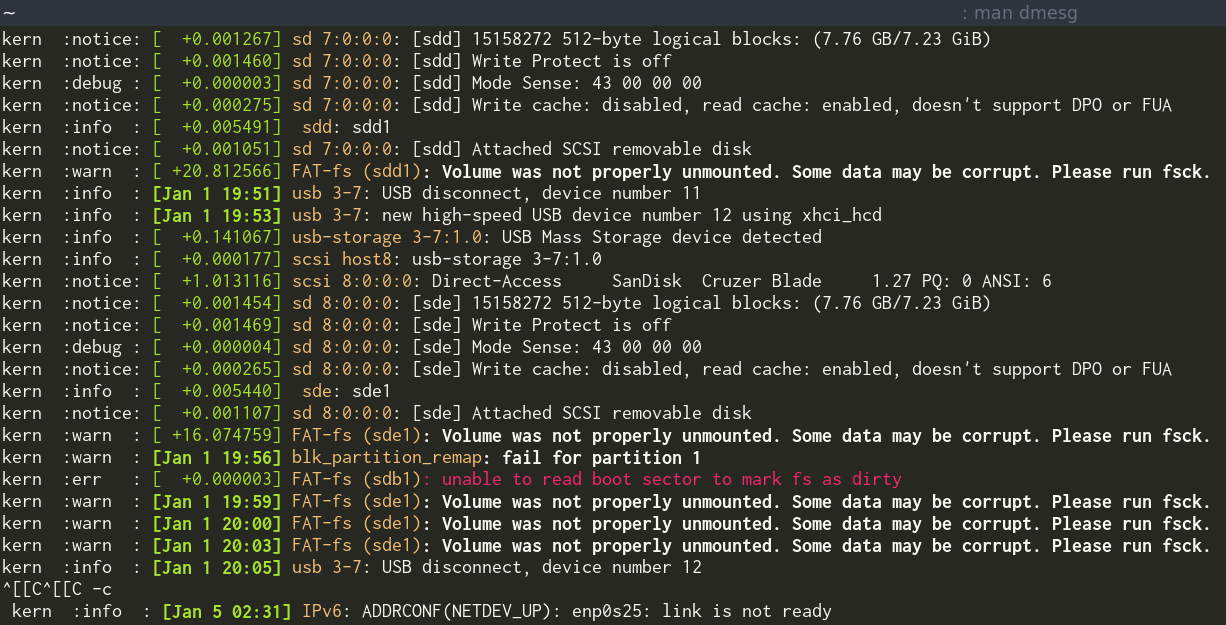}}
  \caption{Screenshot of $\mathsf{dmesg}$}
  \label{fig:dmesg}
\end{figure*}

\subsubsection{An \oRTE for $\mathsf{dmesg}$}
We first give an \oRTE which analyzes the output of $\mathsf{dmesg}$, and
produces the appropriate contexts.  The RTE is a specification language which is
easier to understand than the transducer which describes the same computation.

The required \oRTE
inspects the output of $\mathsf{dmesg}$, and if it detects a line containing the
message \err, then it outputs 10 lines before this line, 10 lines after this
line as the necessary context needed to investigate the reason for the error.
This is done for all lines containing \err.

The \oRTE is broken into two parts.  We first look at the first 10 lines of the
file.  If any of these lines (say line $i$) has \err, we output the lines $1,
\dots, i+ 10$.
\begin{enumerate}
	\item Let $\nl$ denotes newline, $\mathsf{char}=\Sigma\setminus\{\nl\}$,
	$\mathsf{line}=\mathsf{char}^+ \cdot\nl$ and $\err_i=\mathsf{line}^i \cdot 
  \mathsf{char}^{7}\cdot \err\cdot \Sigma^{\omega}$ be an expression
	which says that there is an ``\err'' in the $i+1$th line for $i \geq 0$.  The
	message starts from the 8th character on each line, so we scan from the 8th
	character for \err.
	
  \item Let $C_i$ =
  $\Ifthenelse{\err_i}{[(\Ifthenelse{\mathsf{line}^{i+11}}{\mathsf{copy}}\bot)
  \lrdot \Ifthenelse{(\mathsf{line}^\omega}{\varepsilon}{\bot})]}{\varepsilon}$
  be an \oRTE which gives the context of the error if the error is found in the
  $(i+1)$th line where $0 \leq i \leq 9$.  The \RTE $\mathsf{copy}$ defines the
  identity function on $\Sigma^{+}$: For instance, if $\Sigma=\{a,b,c\}$ then
  $\mathsf{copy}=\lrplus{\Ifthenelse{a}{a}{(\Ifthenelse{b}{b}{(\Ifthenelse{c}{c}\bot)})}}$.
  Thus, if there is an ``\err'' in the first 10 lines, the context is generated
  using $C_i$.
	
	\item To catch occurrences of ``\err'' in lines 11 and later in the file, we
	use $\twentyoneomega{\mathsf{line}}{C'}$, the 21-chained $\omega$-iteration.
  Here, $C'$ is an \RTE which copies the context for each line containing
  ``\err'' starting from line number 11:
	$$
  C' = \Ifthenelse{[\mathsf{line}^{10} \cdot \mathsf{char}^{7} \cdot \err
  \cdot \mathsf{char}^+ \cdot \nl \cdot 
  \mathsf{line}^{10}]}{\mathsf{copy}}{\varepsilon} \,.
  $$
	
	\item The required \oRTE $C$ = $C_0\odot C_1 \odot \cdots \odot C_9 \odot
	\twentyoneomega{\mathsf{line}}{C'}$.
\end{enumerate}
Thus, the first 10 lines are checked for ``\err'' and the respective context is
output if a line $0 \leq i <10$ has ``\err''; the remaining lines in the file
are treated using the $\twentyoneomega{K}{C'}$, where we look at blocks of 21
lines, and reproduce them as is, if the 11th line in the block has an ``\err'';
this is repeated from the next line and so on.  The $\twentyoneomega{K}{C'}$ is
not a new combinator, it can be written in terms of 2-chained $\omega$-iteration
as shown by Lemma~\ref{lem:komega}.

\subsubsection{Machine description for $\mathsf{dmesg}$}

Now we describe a transducer that produces the contexts 
based on the output of $\mathsf{dmesg}$. 
Since the  $\mathsf{dmesg}$ command continuously monitors 
not only external devices, but also the RAM, processors, cache and the hard disk, 
the output is continuously updated in real time. Hence, the output is
arbitrarily long. 

The transducer $\Aa$ takes as input, the output as generated above, and
needs the two way functionality.  We do not use the look-ahead option here.
Whenever $\Aa$ reads a line containing ``\err'', it goes back 10 lines (or $i \leq
10$ lines, if there are only $i < 10$ lines before the current line), and
outputs the next 21 lines (or $11+i$ lines).  Then it goes back 10 lines to
check the message in the next line.  The transducer $\Aa$ can be obtained 
inductively from the \oRTE $C$ described above, or it can be directly 
constructed. Because it has to count up to 10 several times, the automaton is 
rather large, but it can be constructed easily.

\subsubsection{Reducing  $k$-chained $\omega$-iteration to  2-chained $\omega$-iteration for $k >2$}

\begin{lemma}	\label{lem:komega}
	For $k >2$, the \RTE $\komega{L}{f}$  can be derived from \RTEs defined in 
  Section~\ref{sec:oRTE}.
\end{lemma}

\begin{proof}
	Let $L\subseteq\Sigma^{+}$ be a regular language, let $f\colon\Sigma^{+}\to\D$
	and let $g=\komega{L}{f}$.  First, recall that by definition, $g(w)$ is
	defined as $f(u_1u_2\cdots u_k) \cdot f(u_2u_3\cdots u_{k+1}) \cdot f(u_3 u_4
	\cdots u_{k+2}) \cdots$ iff $w$ admits a unique factorization $w=u_1u_2u_3
	\cdots$ with each $u_i\in L$ for $i\geq 1$.
  
  For $n>0$, we let $L_n$ be the set of words $w$ which admit a unique
  factorization $w=u_1u_2 \cdots u_n$ with each $u_i\in L$ for $1\leq i\leq n$.
  Notice that $L_n\subseteq L^{n}\subseteq\Sigma^{+}$ is a regular language.  We
  also define $L_\omega$ as the set of words $w\in\Sigma^{\omega}$ which admit a
  unique factorization $w=u_1u_2u_3 \cdots$ with $u_i\in L$ for all $i\geq 1$. 
  Indeed, $L_\omega\subseteq L^{\omega}$ is an $\omega$-regular language.
  
  We define by induction for $n\geq 0$ a function $f_n$ such that
  $\dom{f_n}\subseteq L_{k+n}$ and if $w=u_1u_2\cdots u_{k+n}$ is the unique
  factorization of $w$ with $u_i\in L$ for $1\leq i\leq k+n$ then
  $$
  f_n(w)=f(u_1u_2\cdots u_k) f(u_2u_3\cdots u_{k+1}) \cdots
  f(u_{n+1}u_{n+2}\cdots u_{n+k}) \,.
  $$
  We let $f_0=\Ifthenelse{L_k}{f}\bot$ and for $n>0$ we define
  $f_n = (f_{n-1}\lrdot(\Ifthenelse{L}{\varepsilon}\bot)) \odot
  ((\Ifthenelse{L_{n}}{\varepsilon}{\bot})\lrdot f_0)$.
  Note that this gives us 
  $f_1 = (f_{0}\lrdot(\Ifthenelse{L}{\varepsilon}\bot)) \odot
  ((\Ifthenelse{L}{\varepsilon}{\bot})\lrdot f_0)$, which works on strings 
  of length $k+1$, and produces $f(u_1 \cdots u_k) f(u_2 \cdots u_{k+1})$.  
  Likewise, $f_2 = (f_{1}\lrdot(\Ifthenelse{L}{\varepsilon}\bot)) \odot
  ((\Ifthenelse{L_2}{\varepsilon}{\bot})\lrdot f_0)$, which works on strings 
  of length $k+2$, and produces $f_1(u_1 \cdots u_{k+1}) f(u_3 \cdots u_{k+2})$, which 
  in turn expands to  $f(u_1 \cdots u_k) f(u_2 \cdots u_{k+1}) f(u_3 \cdots u_{k+2})$.

  Finally, let $h=\Ifthenelse{L_\omega}{\twoomega{L_k}{f'}}\bot$ with 
  $f'=f_{k-1}\lrdot(\Ifthenelse{L}{\varepsilon}\bot)$.
  Notice that $\dom{f'}\subseteq L_{2k}$.
  We claim that $g(w)=h(w)$ for all $w\in\Sigma^{\omega}$.
  
  Let $w\in L_\omega\supseteq\dom{g}\cup\dom{h}$.  Consider the unique
  factorization $w=u_1u_2u_3\cdots$ with $u_i\in L$ for all $i\geq1$.  For all
  $i\geq1$, let $w_i=u_{(i-1)k+1}\cdots u_{ik}$.  Clearly, 
  $w=w_1w_2w_3\cdots$ is the unique factorization of $w$ with $w_i\in L_k$ for
  all $i\geq1$.
  Now, $h(w)=f'(w_1w_2)f'(w_2w_3)\cdots$ and for each $i\geq1$ we 
  have
  $$
  f'(w_iw_{i+1})=f(u_{(i-1)k+1}\cdots u_{ik}) f(u_{(i-1)k+2}\cdots u_{ik+1})
  \cdots f(u_{ik}\cdots u_{ik+k-1}) \,.
  $$
  We deduce that $h(w)=g(w)$.
\end{proof} 
  
\begin{example}
	Let $k = 4$.  We have $g = [L,f]^{4\omega} =
	\Ifthenelse{L_\omega}{\twoomega{L_4}{f'}}\bot$ by lemma~\ref{lem:komega}.  We
	show the computation of $[L, f]^{4\omega}$ on a word $w\in L_\omega$ using
	this new \RTE. Consider the unique factorization $w=u_1u_2u_3\cdots$ with
	$u_i\in L$ for all $i\geq1$.  Also, let $w_1 = u_1 \cdots u_4$, $w_2 = u_5
	\cdots u_8$ and so on. We have $w_i \in L_4$ for all $i \geq 1$.
  Now,
  $$
  [L,f]^{4\omega}(w) = f(u_1 \cdots u_4) f(u_2 \cdots u_5) \cdots
  \qquad\text{and}\qquad
	[L_4, f']^{2\omega}(w) = f'(u_1 \cdots u_8) f'(u_5 \cdots u_{12}) \cdots
  $$
  We show the expansion of $f'(u_1 \cdots u_8)$ below.
	\begin{align*}
	&f'(u_1 \cdots u_8) \\
	&= (f_3  \lrdot (\Ifthenelse{L}{\varepsilon}{\bot} ))(u_1 \cdots u_8) & 
  \text{def.\ of } f' \\
	& = f_3(u_1 \cdots u_7) & \\
	& = ((f_2 \lrdot (\Ifthenelse{L}{\varepsilon}{\bot})) \odot ((\Ifthenelse{L_3}{\varepsilon}{\bot}) \lrdot f_0))(u_1 \cdots u_7) & \text{def.\ of } f_3\\
	& = f_2(u_1 \cdots u_6) \cdot f_0(u_4 \cdots u_7) & \\
	& = ((f_1 \lrdot (\Ifthenelse{L}{\varepsilon}{\bot})) \odot ((\Ifthenelse{L_2}{\varepsilon}{\bot}) \lrdot f_0))(u_1 \cdots u_6) 
  \cdot f_0(u_4 \cdots u_7) & \text{def.\ of } f_2\\
	& = f_1(u_1 \cdots u_5) \cdot f_0(u_3 \cdots u_6) \cdot f_0(u_4 \cdots u_7) & \\
	& = ((f_0 \lrdot (\Ifthenelse{L}{\varepsilon}{\bot})) \odot ((\Ifthenelse{L}{\varepsilon}{\bot}) \lrdot f_0))(u_1 \cdots u_5)
  \cdot f_0(u_3 \cdots u_6) \cdot f_0(u_4 \cdots u_7) & \text{def.\ of } f_1\\
	& = f_0(u_1 \cdots u_4) \cdot f_0(u_2 \cdots u_5) \cdot f_0(u_3 \cdots u_6) \cdot f_0(u_4 \cdots u_7) & \\
	& = f(u_1 \cdots u_4) \cdot f(u_2 \cdots u_5) \cdot f(u_3 \cdots u_6) \cdot f(u_4 \cdots u_7) & \text{def.\ of } f_0 
	\end{align*}
	Recursion tree for the same is shown in the figure~\ref{fig:tree}.   Hence, \\ 
	\noindent
	$
  [L_4, f']^{2\omega}(w) =  f(u_1 \cdots u_4) \cdot f(u_2 \cdots u_5) \cdot 
  f(u_3 \cdots u_6) \cdot f(u_4 \cdots u_7) \cdot 
	f(u_5 \cdots u_8) \cdot f(u_6 \cdots u_9) \cdot f(u_7 \cdots u_{10}) \cdot f(u_8 \cdots u_{11}) \cdots 
	$
\end{example}

\begin{figure}[h]
  \centerline{\includegraphics[scale=0.3]{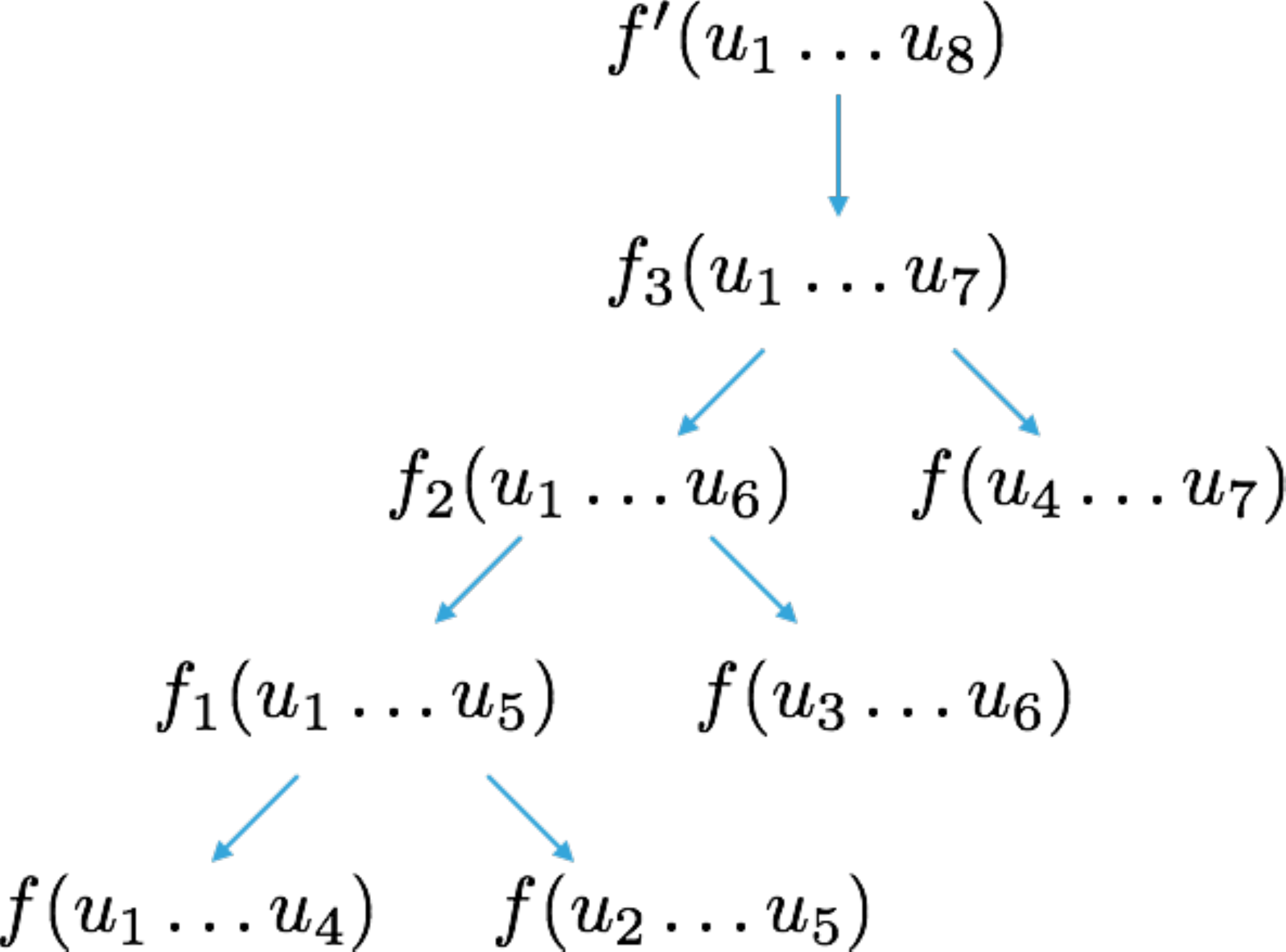}}
  \caption{Computation of $f'(u_1 \cdots u_8)$}
  \label{fig:tree}
\end{figure}

\begin{lemma}	\label{lem:kplus}
	Similar to lemma~\ref{lem:komega}, the \RTE $\kplus{L}{f}$ can be derived from
	\RTEs defined in Section~\ref{sec:RegExp} for $k > 2$.
\end{lemma}

\begin{proof}
	The proof technique is similar to Lemma~\ref{lem:komega}.
	Definitions of $L_n, f_n, f'$ are as defined in this lemma.
  In addition we use $L_+$, the set of words having a unique decomposition
	$w = u_1 u_2 \cdots u_n$ with $n\geq1$ and $u_i\in L$ for all $1\leq i\leq 
  n$. We also use $L_{<k}=L_0\cup L_1\cup\cdots\cup L_{k-1}$, where 
  $L_0=\{\varepsilon\}$. We define
  \begin{align*}
    h &=
    \Ifthenelse{L_+}{[\big(\twoplus{L_k}{f'}\lrdot(\Ifthenelse{L_{<k}}{\varepsilon}\bot)\big)
    \odot \big((\Ifthenelse{L_k^{*}}{\varepsilon}\bot)\lrdot h' \big)]}\bot
    \\
    h' &= \Ifthenelse{L_k}{f}{( \Ifthenelse{L_{k+1}}{f_1}{( 
    \Ifthenelse{L_{k+2}}{f_2}{( \cdots (\Ifthenelse{L_{2k-1}}{f_{k-1}}\bot)
    \cdots )})})} \,.
  \end{align*}
  We can show as in Lemma~\ref{lem:komega} that $h=\kplus{L}{f}$.
\end{proof}

\section{Equivalence of Models}\label{app:two-way-eq}

In this section, we look at the equivalence of the automata models for regular
transformations on infinite words.  The model used here is a two-way,
deterministic Muller automaton, which has for each pair $(q,a)$ consisting of a
state and symbol, a tuple of look-ahead $\omega$-regular languages which are 
mutually exclusive.  The model (denoted 2\textsf{WST}\textsubscript{la}) used in
\cite{lics12} however is a two-way deterministic Muller automaton which is
equipped with a look-behind automaton (a NFA) and a look-ahead automaton (a
possibly non-deterministic Muller automaton).  Here, we show that these two
models are equivalent in expressiveness.
     
\begin{lemma}
	\WST and 2\textsf{WST}\textsubscript{la} (defined in \cite{lics12}) are
	equivalent representations for $\omega-$regular transformations.
\end{lemma}

\begin{proof}[Proof sketch.]
	\begin{itemize}
		\item[$(\subseteq)$] Given an \WST $\Aa
		=(Q,\Sigma,\Gamma,q_0,\delta,\Ff,\Rr)$, one can construct a
		2\textsf{WST}\textsubscript{la} $\Aa'$ such that $\sem{\Aa} = \sem{\Aa'}$.
		We create a look ahead automata $A_L$ as the disjoint union of all $\Aa_i =
		(Q_i, \Sigma, s_i, \delta_i, \Ff_i)$ corresponding to each look ahead
		language $R_i$ in $\Rr$: $\mathcal{L}(\Aa_i) = R_i$.  The set of states of
		the 2\textsf{WST}\textsubscript{la} $\Aa'$ is $Q$, with initial state $q_0$
		and Muller accepting set $\Ff$.  The transition function $\delta'$ of $\Aa'$
		is defined by $\delta'(q, a, \Sigma^*, {s_i}) = (q', \gamma, d)$, if
		$\delta(q, (a, R_i)) = (q', \gamma, d)$.  Note that since we do not use any
		look-behind in $\Aa$, the look-behind automaton needed for the model in
		\cite{lics12} is the trivial one which accepts all strings.  Rather than
		writing the single state look-behind automaton (where the state is both
		accepting and initial), we write the expression $\Sigma^*$.
		
    \item[$(\supseteq)$] Given 2\textsf{WST}\textsubscript{la} $\Tt=((Q, \Sigma,
    \Gamma, \delta, q_0, \Ff), \Aa, \Bb)$ with look-ahead automata $\Aa=(Q_A,
    \Sigma, \delta_A, \Ff_A)$ and look-behind automata $\Bb= (Q_B, \Sigma,
    \delta_B, F_B)$, one can construct an \WST 
    $\Aa'=(Q',\Sigma,\Gamma,q'_0,\delta',\Ff',\Rr)$ such that $\sem{\Tt} =
    \sem{\Aa'}$.  For each state $p\in Q_B$ we let $K_p\subseteq\Sigma^{*}$ be
    the regular languages accepted by $\Bb$ with initial state $p$.  Similarly,
    for each state $p\in Q_A$ we let $L_p\subseteq\Sigma^{\omega}$ be the
    $\omega$-regular languages accepted by $\Aa$ with initial state $p$.  From
    the determinism of $\Tt$ we deduce that the languages $(K_p)_{p\in Q_B}$ are
    mutually exclusive, and similarly the languages $(L_p)_{p\in Q_A}$ are
    mutually exclusive. 
    
    Now, a transition $\delta(q,(s_B,a,s_A))=(q',\gamma,d)$ of $\Tt$ can be
    replaced with the more abstract
    $\delta'(q,(K_{s_B},a,L_{s_A}))=(q',\gamma,d)$ using regular languages
    instead of states for look-ahead and look-behind. In order to obtain the 
    \twoDMTla it remains to remove the look-behind.
    
    Consider a DFA $\Dd=(Q_D, \Sigma, \delta_D, s_D, (F_p)_{p\in Q_B})$ which
    simultaneously recognizes all the $(K_p)_{p\in Q_B}$: $K_p$ is the set of
    words accepted by $\Dd$ when using $F_p$ as set of final states.  We define
    $\Aa'$ as the synchronized product of $\Tt$ and $\Dd$.  More precisely, we
    let $Q'=Q\times Q_D$, $q'_0=(q_0,s_D)$, $\Rr=\{L_{s_A}\mid s_A\in Q_A\}$,
    $\Ff'$ consists of all $X\subseteq Q'$ such that the projection of $X$ on
    $Q$ belongs to $\Ff$.  The transition function $\delta'$ is defined as
    follow.  If $\delta(q,(s_B,a,s_A))=(q',\gamma,d)$ is a transition of $\Tt$,
    then $\delta'((q,p),(a,L_{s_A}))=((q',p'),\gamma,d)$ is a transition of
    $\Aa'$ provided
    \begin{itemize}
      \item $p\in F_{s_B}$, i.e., the prefix read so far belongs to $K_{s_B}$, 
      or equivalently, is recognized by the look-behind automaton $\Bb$ when 
      starting from state $s_B$,
    
      \item $p'=\delta_D(p,a)$ if the transition moves forward ($d=+1$),
      
      \item and if the transition moves backward ($d=-1$) then $p'$ is obtained 
      form $p$ using the \emph{reverse-run} algorithm of Hopcroft and Ullman.
      \qedhere
    \end{itemize}
  \end{itemize}
\end{proof}

\end{document}